\documentclass[preprint]{elsarticle}

\usepackage{hyperref}
\usepackage[utf8]{inputenc}
\usepackage{amsmath,amsfonts,amssymb,mathrsfs,color,epsfig,amsthm}
\usepackage{tikz}

\newcommand\norm[1]{||#1||}
\newcommand\dist{\Delta}
\newcommand\asy{\,\overset{\infty}{=}\,}
\newcommand\gplus{\,\oplus\,}
\newcommand\gpconf{\,\overline\oplus\,}
\newcommand\oconf{\,\overline0\,}
\newcommand{\N}{\mathbb{N}}
\newcommand{\Z}{\mathbb{Z}}
\newcommand{\T}{\mathbb{T}}
\newcommand{\G}{\mathbb{G}}
\newcommand{\R}{\mathbb{R}}
\newcommand{\F}{\mathbb{F}}
\newtheorem{defi}{Definition}[section]
\newtheorem{theo}[defi]{Theorem}
\newtheorem{prop}[defi]{Proposition}
\newtheorem{lemm}[defi]{Lemma}
\newtheorem{coro}[defi]{Corollary}
\newtheorem{exam}[defi]{Example}
\newtheorem{rem}[defi]{Remark}
\newproof{pf}{Proof}

\begin{document}

\begin{frontmatter}
\title{Pre-Expansivity in Cellular Automata}

\author[conce,cmm]{A. Gajardo\fnref{fcyt,basal}}
\ead{angajardo@udec.cl}

\author[lyc]{V. Nesme}

\author[cnrs,cmm]{G. Theyssier\corref{cor}\fnref{basal}}
\ead{guillaume.theyssier@cnrs.fr}

\cortext[cor]{Corresponding author}
\fntext[fcyt]{This author acknowledges the support of CONICYT-FONDECYT\#1140684 and ENLACE project of Universidad de Concepción}
\fntext[basal]{This author acknowledges the support of CONICYT-Basal PFB03.}

\address[conce]{Departamento de Ingenier\'ia Matem\'atica and Centro de Investigaci\'on en Ingenier\'ia Matem\'atica (CI$^2$MA)\\ Universidad de Concepci\'on, Chile}
\address[cmm]{Centro de Modelamiento Matem\'atico (CMM), Universidad de Chile, Chile}
\address[lyc]{Lycée Georges Brassens, France}
\address[cnrs]{Aix Marseille Université, CNRS, Centrale Marseille, I2M, Marseille, France}

\begin{abstract}
 We introduce the notion of pre-expansivity for cellular automata (CA):
  it is the property of being positively expansive on asymptotic pairs of
  configurations (i.e. configurations that differ in only finitely many
  positions). Pre-expansivity therefore lies between positive expansivity and
  pre-injectivity, two important notions of CA theory.

  We show that there exist one-dimensional pre-expansive CAs which are
  not positively expansive and they can be chosen reversible (while
  positive expansivity is impossible for reversible CAs). We provide
  both linear and non-linear examples. In the one-dimensional setting,
  we also show that pre-expansivity implies sensitivity to initial
  conditions in any direction. We show however that no two-dimensional
  Abelian CA can be pre-expansive. We also consider the finer notion
  of $k$-expansivity (positive expansivity over pairs of
  configurations with exactly $k$ differences) and show examples of
  linear CA in dimension 2 and on the free group that are
  $k$-expansive depending on the value of $k$, whereas no (positively)
  expansive CA exists in this setting.
\end{abstract}


\end{frontmatter}

%\linenumbers

\section{Introduction}

The model of cellular automata is at the crossroads of several
domains and is often the source of surprisingly complex objects in
several senses (computationally, dynamically, etc). 

From the point of view of dynamical systems and symbolic dynamics, the
theory of cellular automata is very rich \cite{hedlund,kurkabook,ergca,cagroups} and tells us, on the one hand, that CA
are natural examples of chaotic systems that can perfectly fit the
standard notions developed in a general context, and, on the other
hand, that they have special properties allowing and justifying the
development of a refined and dedicated theory.
For instance, the
structure of the space of configurations allows to define the notion of
an asymptotic pair of configurations: two configurations that differ only 
on finitely many positions of the lattice.
The Garden of Eden theorem,
which has a long history \cite{hedlund,moore,myhill,AIF_1999,Bartholdi_2010,Bartholdi16,machi,cagroups} and is emblematic of this
CA specific theoretical development, then says that surjectivity is
equivalent to pre-injectivity (injectivity on asymptotic pairs) if and
only if the lattice is given by an amenable group.

Two important lines of questioning have been particularly developed
and provide some of the major open problems of the field \cite{BoyleOpen}:
\begin{itemize}
\item surjective CA and their dynamics;
\item how does CA theory changes when changing the lattice.
\end{itemize}
In particular, the classical notion of (positive) expansivity has been
applied to CA giving both a rich theory in the one-dimensional case
\cite{blanchardmaass,kurkabook,nasu} and a general inexistence result
in essentially any other setting \cite{shereshevsky,Pivato11}.
Even
in the one-dimensional case where positive expansivity is equivalent
to being conjugated to a one-sided subshift of finite type \cite{Kurka97}, it is
interesting to note that outside the linear and bi-permutative
examples, few construction
techniques are known to produce positively expansive CA
\cite{JadurYazlle}.
On the other hand, it is still unknown whether
positive expansivity is a decidable property, although it is indeed
decidable for some algebraic cellular automata \cite{dilena06,phdLukkarila}.

In this paper, we introduce a new dynamical property called \emph{pre-expansivity}
that both generalizes positive expansivity and refines pre-injectivity: it is
the property of being positively expansive on asymptotic pairs.
Our motivation
is to better understand surjective CA and expansive-like dynamics, in
particular in the higher-dimensional case or in lattices where the classical
notion of positive expansivity cannot be satisfied by any CA \cite{shereshevsky,Pivato11}.
Pre-expansivity is weaker than positive expansivity. %%%%% AG: comentario
We show examples of pre-expansive CAs which are not positively expansive.
In such CA, some perturbations on infinitely many cells does not propagate at all, while every finite perturbation will be eventually seen in the neighborhood of every cell: it is the finiteness of the perturbation that allows the propagation on every direction.

Pre-expansivity is interesting in that: %%%%% AG: borré ``However'' agregué salto de línea
\begin{enumerate}
\item a reversible CA can be pre-expansive (see section~\ref{sec:1D}), while none can be positively expansive~\cite{posexponetoone};
\item pre-expansivity implies sensitivity in all directions (see Proposition~\ref{prop:directional}), while some expansive CA (like the shift map) have equicontinuous directions.
\end{enumerate}
This shows that the notion is useful in the classical setting of one-dimensional cellular automata.

For other settings, the situation is left open: on one hand, we show an impossibility result for Abelian CA in dimension ${d\geq 2}$ (see Theorem~\ref{thm:abelian2D}). %%%%% AG: comentario
This means that for every Abelian CA and every finite window, there will be a finite configuration that will preserve the window in state 0 forever.
%%%% AG: agregué una coma
On the other hand, we give several examples of $k$-expansive CA in this setting and on the free group, where $k$-expansivity means positive expansivity over pairs of configurations with exactly $k$ differences.

The paper is organized as follows.
In Section~\ref{sec:def} we give
the main definitions and results we need to work on cellular automata
on groups.
In section~\ref{sec:abelian}, we focus on Abelian cellular automata and
develop a toolbox for this class that is used later in different
sections. As an aside, we prove that such CA are always predictable in logarithmic space complexity.
In Section~\ref{sec:prexp}, we introduce pre-expansivity and
$k$-expansivity, and we give some preliminary results which do not
depend on the group defining the space.
In Section~\ref{sec:1D},
we restrict to the group $\Z$ and give various examples of cellular automata
which are pre-expansive but not positively expansive, including a characterization for a particular family of non-linear CA, namely multiplication CA. %%%%% AG: agregué esta última frase
In Section~\ref{sec:freegroup}, we consider the free group and show that
$k$-expansivity is possible for infinitely many values of $k$ although
positive expansivity is impossible.
Finally, in Section~\ref{sec:2D},
we restrict to the group $\Z^2$ and we study some $k$-expansive
examples for particular values of $k$ but also show that there is no
pre-expansive Abelian cellular automaton. %%%%% AG: 

\section{Formal Setting and Classical Definitions}
\label{sec:def}

We will work on cellular automata defined over a \emph{finitely generated group} $\G$.
 We will consider Abelian and non-Abelian groups, but since most of our examples are given for Abelian groups, we will prefer the additive notation for $\G$.

Fixing a generator set $G$, that is closed under inversion, a \emph{norm} can be defined in $\G$: given $z\in \G$, $\norm{z}$ is the length of the shortest sequence $g_1g_2...g_n$ of elements in $G$ such that $z=g_1+g_2+...+g_n$.
This norm induces a metric in $\G$ naturally, and given a non-negative integer $r$ we can define also the ball of radius $r$ and center $z$ as the set $B_r(z)=\{ x\in\G\ |\ \norm{-z+x}\le r\}$.
Given a point $z\in\G$ and two sets $X$, $Y\subseteq \G$, we accept the following notation.
$$ z+S=\{z+x\ |\  x\in S\},\quad \textrm{ and } \quad X+Y=\{x+y\ |\ x\in X, y\in Y\} $$

Cellular automata are functions defined on the \emph{symbolic space} $Q^\G=\{c:\G\rightarrow Q\ |\ c\textrm{ is a function}\}$.
An element $c$, called \emph{configuration}, assigns a symbol of $Q$ to each element of the group $\G$, sometimes called cells.
We will use both $c(z)$ and $c_z$ to denote the value of $c$ at the cell $z$.
A natural $\G$-action on $Q^\G$ is the \emph{shift}: given $z\in\G$, the function: $\sigma_z:Q^\G\rightarrow Q^\G$ is defined by $\sigma_z(c)(x)=c(z+x)$ for every $x\in\G$.
The \emph{Cantor distance} in $Q^\G$ is defined for any two configurations $c$, $d$ as follows.

$$\dist(c,d)=\begin{cases}
2^{-\min\{\norm{z}:c(z)\not=d(z)\}} &\textrm{ if }c\not = d\\
0 & \textrm{ if } c=d
\end{cases}$$

\begin{defi}
  Two configurations $c,d$ are \emph{asymptotic}, denoted ${c\asy d}$, if they differ only
  in finitely many positions: ${\{z\in\G: c(z)\not=d(z)\}}$ is finite.
\end{defi}

A cellular automaton (CA) is an endomorphism of $Q^\G$, compatible with the shift $\G$-action and continuous for the Cantor distance. %%%%% AG: no es ``Cantor distance''?
From Curtis-Hedlund theorem~\cite{hedlund,cagroups}, every
cellular automaton $F$ is characterized by a \emph{local function}
$f:Q^V\rightarrow Q$, where $V\subset \G$ is finite and called \emph{neighborhood of $F$}, as follows.

$$ \forall c\in Q^\G, \forall z\in \G, F(c)(z)=f(\sigma_{z}(c)|_{V})$$ 

Every function defined in this way is a cellular automaton.

Basic properties of $F$ such as surjectivity and injectivity have been considered and played an important role in CA theory because they were proved to be efficiently decidable in dimension 1 but undecidable in higher dimensions \cite{Kari94,Amoroso72}.
The weaker notion of \emph{pre-injectivity} says that, for every
pair of different asymptotic configurations $c$ and $d$, their image by $F$
are different: 
\[c\asy d\text{ and }c\not=d\Rightarrow F(c)\not=F(d).\]
The so-called Garden-of-Eden theorem establishes that surjectivity is
equivalent to pre-injectivity, which in particular implies that
injective CAs are also bijective (equivalently reversible by
Curtis-Hedlund Theorem, i.e. having an inverse which is also a CA). It
was first proved in particular cases \cite{hedlund,moore,myhill} and
later it was shown that it holds exactly when the group $\G$ is
amenable, \textit{i.e.} when it admits a finitely additive measure
which is invariant under its action \cite{cagroups}.

The pair $(Q^\G,F)$ is a dynamical system and can be studied from the
point of view of topological dynamics.
 % We will be interested in the action of the powers of $F$ over $Q^\G$: $\{F^t\}_{t\in\T}$, where $\T$ is usually equal to $\N$ but it can be considered as $\Z$ when $F$ is reversible. \TODO{vale la pena decir esto? habría que definir las órbitas y hablar un poco más de eso en tal caso...}
The present work proposes a new particular kind of topological unpredictability.
Weaker and stronger notions in this area are the following.

A CA $F$ is \emph{sensitive} if there exists a number $\delta>0$, called \emph{sensitivity constant}, such that for every $c$ and every $\epsilon$ there exists an instant $t\in\N$ and a configuration $d\in B_\epsilon(c)$ such that $\dist(F^t(c),F^t(d))\ge\delta$.

A stronger notion is expansivity, which can be of two kinds, and depends on whether the CA is reversible or not.
Given $\T$ be equal to either $\N$ or $\Z$, a CA $F$ is called \emph{$\T$-expansive} if there exists a number $\delta>0$, called \emph{expansivity constant}, such that for every $c\not=d$ there exists an instant $t\in\T$ such that $\dist(F^t(c),F^t(d))\ge\delta$.
In this work we will almost always consider positive times only (\textit{i.e.} $\T=\N$) and refer to $\N$-expansivity as positive expansivity. Accordingly we will stick to positive times in the definition of pre-expansivity below. This choice is particularly relevant for one-dimensional reversible CA, since they all possess a direction of $\Z$-expansivity\footnote{It can be checked that any direction greater than both the radius of the CA and its inverse is a direction of $\Z$-expansivity, see \cite[Proposition 5.3]{Sablik08} for more details.}, none has a direction of $\N$-expansivity, but having directions of pre-expansivity is a non-trivial property illustrated for instance by Proposition~\ref{prop:directional}.

Denote by $T_m : Q^\G\rightarrow (Q^{B_m(0)})^\N$ the \emph{trace
  function} which to any configuration associates its orbit
restricted to $B_m(0)$:
\[T_m(c) = \left(t\mapsto \bigl(F^t(c)\bigr)|_{B_m}\right).\]

A CA $F$ is positively expansive if and only if $T_m$ is injective for some $m$, in which case we get a conjugacy between $F$ acting on $Q^\G$ and the one-sided shift acting on ${T_m\bigl(Q^\G\bigr)}$ \cite{kurkabook}.

If $\G=\Z$, given a one-dimensional CA with local rule
${f:Q^{[-l,r]}\rightarrow Q}$ with ${l, r>0}$, we
say that it is \emph{LR-permutive} when for any ${q_{-l},\ldots,q_r\in Q}$
the two following maps are bijective.
\begin{align*}
  a&\mapsto f(a,q_{-l+1},\ldots,q_r)\\
  a&\mapsto f(q_{-l},\ldots,q_{r-1},a)\\
\end{align*}

LR-permutive are always positively expansive.

\section{Abelian CA: Definitions and Toolbox}
\label{sec:abelian}

The class of Abelian CA will be an important source of examples in the sequel.
This section establishes a number of properties used later on for both positive and negative results.
These properties are essentially folklore knowledge or extensions of already published results, mainly in~\cite{fractal}.
We however give detailed proofs below because our setting is more general than the usual one.
In particular, as far as we know, Corollary~\ref{cor:logspace} was never written in this level of generality, and Lemma~\ref{lem:prefijotraza} is new.
The reader can skip this section in a first read: the following definition is used everywhere, but the main results established below are only used in section~\ref{sec:2D}. %%%%% AG: agregué esta observación para calmar al lector ansioso.

\begin{defi}
  Let $(Q,\gplus)$ be a finite group and denote by $\gpconf$
  the component-wise extension of $\gplus$ to $Q^\G$ and by $\oconf$ the configuration identically equal to 0. 
  A CA $F$ over $Q^\G$ is \emph{linear} if 
  \[\forall c,d\in Q^\G: F(c\gpconf d) = F(c)\gpconf F(d).\]
  When ${(Q,\gplus)}$ is an Abelian group we say that $F$ is \emph{Abelian}, which is equivalent to the fact that $F$ verifies an equation of the form: 
\[F(c)_z = \sum_{i\in V} h_i(c_{z+i})\]
for any configuration $c$, where the sum corresponds to the $\gplus$ law and where $V$ is the neighborhood of $F$ and $h_i$ are endomorphisms of ${(Q,\gplus)}$.
\end{defi}

Given $a\in Q$, we denote by $c^a$ the configuration that is equal to $e$ (identity of the group $Q$) everywhere except at cell $0$ where its value is $a$. Any space time diagram of a linear CA is a sum of translated copies of space-time diagrams with initial configuration of the form $c^a$.

The case where ${(Q,\gplus)}$ is a cyclic group has received much more attention in the literature than the general Abelian case. We would like to stress the importance of considering the general case. First, it was already established that some dynamical behaviors related to randomization are possible in the general case but not in the cyclic case \cite[Thms 3 and 4]{abelianrandomization}. Second, we will show in section~\ref{sec:1D} below (Proposition~\ref{prop:cyclic1D} and Theorem~\ref{prop:prexpexample}) that pre-expansivity is equivalent to positive expansivity in the cyclic case while there are reversible (therefore not positively expansive) pre-expansive CA in the general Abelian case.

The following lemma shows that Abelian CA can be
decomposed according to the structure of the group. It is a folklore
knowledge that appears often in the particular case of cyclic groups
\cite{itoosatonasu,Manzinimargara}, and also in the more general
Abelian case \cite{fractal}.

Recall that the product ${F\times G}$ of two CA $F$ and $G$ is the CA
defined on the product alphabet and applying $F$ and $G$ on each
component independently.

\begin{lemm}
  \label{lem:groupdecomp}
  Let ${Q=Q_p\times Q'}$ be an Abelian group (law $+$ and neutral
  element $(0,0)$) where $Q_p$ is a $p$-group (the order of every element is a
  power of $p$) for some prime $p$ and the order of $Q'$ is relatively
  prime with $p$.

  Then, any Abelian CA $F$ over $Q$ is isomorphic to ${F_p\times F'}$
  where $F_p$ is an Abelian CA over $Q_p$ and $F'$ is an Abelian CA over $Q'$.
\end{lemm}
\begin{proof}
  By linearity of $F$, if $c$ satisfies that ${n\cdot c
    = \overbrace{c + \cdots + c}^n = \overline{(0,0)}}$, then $F(c)$ must
  satisfy the same: ${n\cdot F(c)= \overline{(0,0)}}$. We deduce that the
  subset of states ${Q_1=Q_p\times\{0_{Q'}\}}$ induces a subautomaton
  $F_1$ of $F$ because any configuration $c\in Q_1^{\G}$ is such
  that ${p^k\cdot c=\overline{(0,0)}}$ for some $k$ and no configuration
  in $\bigl(Q\setminus Q_1\bigr)^{\G}$ has this property. Moreover
  if ${n\cdot c=\overline{(0,0)}}$ for $n$ relatively prime with $p$, it
  implies that ${c\in Q_2^\G=\bigl(\{0_{Q_p}\}\times Q'\bigr)^\G}$
  (because the order of an element must divide the order of the group
  it belongs to). Therefore $Q_2$ induces a subautomaton $F_2$ of $F$.

  Now, any ${c\in Q^\G}$ can be written ${c=c_1+c_2}$ where
  ${c_1\in Q_1^\G}$ and ${c_2\in Q_2^\G}$ through cellwise and
  componentwise decomposition, and ${F(c) = F_1(c_1) + F_2(c_2)}$. 
  $F_1$ is isomorphic to a linear CA $F_p$ over
  $Q_p$ and $F_2$ to a linear CA $F'$ over $Q'$, and then $F$ is
  isomorphic to ${F_p\times F'}$.
\end{proof}

Given an Abelian CA $F$, spot configurations (\textit{i.e.}
configurations $c$ everywhere $0$ except in one cell) form a basis of the
whole set of configurations and we get the orbit of any configuration
by summing the orbits of corresponding spot configurations.
The main
point of this section is the existence of a substitutive structure
describing the space time diagram of spot configurations in an Abelian
CA over $p$-groups when the underlying spatial structure is $\Z^d$. %%%%% AG: destaco que es solo valido en Z^d
This in turn comes from the fact that such CA
verify \emph{multiscale additive identities}. %%%%% AG: \emph
Intuitively, a multiscale additive identity is the generalization to the Abelian setting of a linear dependency present in any space-time diagram between a finite set of cells whose relative positions correspond to a basic shape or a blowup of it of factor $\alpha^p$ for some given $\alpha$.

These facts were established in~\cite{fractal} in the one-dimensional
setting.
We give below a proof in any dimension $d$ using essentially the
same approach.
As it is usual in $\Z^d$, we define for $z\in\Z^d$, $\|z\|_\infty=\max\{|z_i|\ : i\in\{1,..,n\}$.
%%%%%%%%%%%%%%%%%%%%%%%%%%%%%%%NORMA infinito%%%%%%%%%%%%%%%

\begin{defi}
  \label{def:multiscaleidentity}
  Let $F$ be a $d$-dimensional Abelian CA over ${(Q,\oplus)}$. We say that $F$ has a \emph{multi-scale additive identity} if there is some scale factor ${\alpha\geq 2}$, ${M>0}$, a finite set ${X=\{(\vec{z_1},t_1),\ldots,(\vec{z_k},t_k)\}}$ with ${0\leq t_i<M}$ for all ${1\leq i\leq k}$ and endomorphisms ${(h_i:Q\rightarrow Q)_{1\leq i\leq k}}$ such that for any ${n\in\N}$ and any configuration $c$ it holds: 
  \[F^{M\alpha^{n}}(c) = \sum_{1\leq i\leq k} \overline{h_i}\circ F^{t_i\alpha^n}\circ \sigma_{\alpha^n\vec{z_i}}(c)\]
  where the sum correspond to law $\overline{\oplus}$ over configurations and $\overline{h_i}$ is the componentwise extension of $h_i$ to $Q^\G$.
\end{defi}
%%%%% AG: uso \emph en cada definición para destacar qué es lo que se define. Explicito el dominio de h_i para facilitar la lectura, y defino \overline{h_i}

\begin{exam}
  Consider the CA $F:\Z_2^\Z\rightarrow \Z_2^\Z$, defined by ${F(c) = c \overline{+} \sigma(c)}$.
  It is straightforward to check that 
\[F^{2^{n}}(c) = c \overline{+} \sigma_{2^n}(c)\]
for all $n$.
In this case $k=2$, $M=1$, $X=\{(0,0),(1,0)\}$, $\alpha=2$ and $h_1=h_2=id$.
\end{exam}
%%%%% AG: preferí explicitar las constantes, para ser más didáctica y aparté el ejemplo como tal.

Similar multi-scale additive identities where no $F$ appears on the right hand side can be derived using the binomial formula as soon as the Abelian CA
\[F(c)_z = \sum_{i\in V} h_i(c_{z+i})\] is such that the $h_i$ are commuting endomorphisms. The situation is a bit more complex in the general case when the $h_i$ do not commute. However,  we have the following lemma.

\newcommand\vecu[1]{u^{\vec{#1}}}
\begin{lemm}
  \label{lem:multiscaleidentities}
  Any $d$-dimensional Abelian CA is a Cartesian product of Abelian CA admitting multi-scale additive identities.
\end{lemm}
\begin{proof}
  First, it is sufficient to prove that Abelian CA over $p$-groups ($p$ prime) admit multi-scale additive identities since any Abelian CA is a Cartesian product of such CA (Lemma~\ref{lem:groupdecomp}).
  Second, it is sufficient to consider $Q$ a $p$-group of the form
  ${Q=\Z_{p^l}^D}$ because any Abelian CA on a $p$-group is a
  subautomaton of some Abelian CA on a group of this form (Proposition
  1 of \cite{fractal}), and a multi-scale additive identity in some CA holds in any of its subautomata.
  Then, an Abelian CA of dimension $d$ over the group ${Q=\Z_{p^l}^D}$ can be viewed as a ${D\times D}$
  matrix whose coefficients are Laurent polynomials with $d$ variables
  $u_1,\ldots, u_d$ and coefficients in $\Z_{p^l}$ (see for instance~\cite{kari00} for more details on this representation).%$Q$.

  Formally,  we denote by ${\Z_{p^l}[u_i,u_i^{-1}]_{1\leq i\leq d}}$ the ring of
  Laurent polynomials with variables $u_1,\ldots,u_d$, \textit{i.e.}
  the ring of linear combinations of monomials made with positive or negative
  powers of the variables and coefficients in ${\Z_{p^l}}$. A monomial corresponds to a vector of
  ${\Z^d}$, hence we use the notation $\vecu{i}$ for any
  ${\vec{i}=(i_1,\ldots,i_d)\in\Z^d}$ to denote the monomial
  ${u_1^{i_1}\cdots u_d^{i_d}}$. A linear cellular automaton is identified with some
  ${T\in\mathcal{M}_D\bigl(\Z_{p^l}[u_i,u_i^{-1}]_{1\leq i\leq d}\bigr)}$
  where the coefficient ${a_{\vec{z}}\in\Z_{p^l}}$ of the monomial $\vecu{z}$ of
  the coefficient $T_{i,j}$ of the matrix $T$ means that, when
  applying the CA, the layer
  $j$ of cell $\vec{z_0}$ receives $a_{\vec{z}}$ times the content of
  the layer $i$ of cell $\vec{z}+\vec{z_0}$, and all these individual
  contributions are summed-up. 
  This matrix representation is correct in the sense that $F^n$ is represented by $T^n$.

  By the Cayley-Hamilton theorem (Laurent polynomials form a
  commutative ring), the characteristic polynomial of $T$ gives a
  relation of the form:
  \[T^m = \sum_{j=0}^{m-1}\sum_{\vec{i}\in
    I}\lambda_{\vec{i},j}\vecu{i}T^j\]
  for some $m\le|D|$, some finite ${I\subseteq\Z^d}$, and where
  $\lambda_{\vec{i},j}\in \Z_{p^l}$. %%%%% AG puse cota a m, para asegurar al lector

  By standard techniques (binomial theorem and Kummer's theorem), we
  have the following identity on any commutative ring of
  characteristic $p^l$ (this is done explicitly in Lemma~10 of \cite{abelianrandomization}):
  \[\left(\sum_i X_i\right)^{p^{n+l-1}} = \left(\sum_i
    X_i^{p^n}\right)^{p^{l-1}}\]
  for any positive $n$. Then, applying this to the expression of
  $T^m$ obtained above we get:
  \[T^{m\cdot p^{n+l-1}} =
  \left(\sum_{j=0}^{m-1}\sum_{\vec{i}\in
        I}\lambda_{\vec{i},j}^{p^n}u^{p^n\cdot\vec{i}}T^{p^{n}\cdot
      j}\right)^{p^{l-1}}.\]
  Noting that the sequence ${(\lambda_{\vec{i},j}^{p^n})_n}$ is
  ultimately periodic, we can choose a large enough common period $N$ so that for all $\vec{i},j$ and all $n\ge1$: 
  \[\lambda_{\vec{i},j}^{p^{nN}} = \lambda_{\vec{i},j}^{p^{N}}.\]
  Denoting ${\alpha=p^N}$ and expanding the right-hand side of the above equality, we have for some $M$ large enough, some finite ${I'\subseteq\Z^d}$, and ${\mu_{\vec{i},j}\in \Z_{p^l}}$ (${0\leq j<M}$ and ${\vec{i}\in I'}$):
  \[T^{M\alpha^n} = \sum_{j=0}^{M-1}\sum_{\vec{i}\in
    I'}\mu_{\vec{i},j}u^{\alpha^n\vec{i}}T^{j\alpha^n}\]
  for any ${n\geq 1}$.
  This is exactly a multi-scale additive identity of scale $\alpha$ expressed in the matrix representation of $F$ and the lemma follows.
\end{proof}

The following corollary shows that Abelian CA are computationally ``easy'' to predict. Particular cases of this statement is folklore knowledge (see \cite{moore98}) and it is mentioned in the generality of Abelian CA in \cite{fractal} (based on Lemma~\ref{lem:substitucion} bellow). We give here a simple proof of this fact using existence of multi-scale additive identities.

\begin{coro}
  \label{cor:logspace}
  For any $d$-dimensional Abelian CA $F$ with alphabet $Q$ and radius $r$, the following prediction problem is computable in LOGSPACE:
  \begin{itemize}
  \item \emph{input :} a finite pattern ${u\in Q^{B_{rn}(0)}}$ and ${q\in Q}$,
  \item \emph{question :} do we have ${F^n(u)=q}$, \textit{i.e.} ${F^n(c)=q}$ for any $c$ with ${c_{|B_{rn}(0)}=u}$? 
  \end{itemize}
\end{coro}
\begin{proof}
  Using Lemma~\ref{lem:multiscaleidentities} it is sufficient to give a LOGSPACE algorithm for Abelian CA with multi-scale additive identities (the Cartesian product is handled by sequentially computing each component which doesn't change the\break LOGSPACE complexity). Therefore let's suppose that $F$ has the following multi-scale additive identity (using notation of definition~\ref{def:multiscaleidentity}):
  \[F^{M\alpha^n}(c) = \sum_{1\leq i\leq k} \overline{h_i}\circ F^{t_i\alpha^n}\circ \sigma_{\alpha^n\vec{z_i}}(c).\]
  An algorithm to evaluate ${F^t(c)_{\vec{z}}}$ is a recursive application of the above identity with the maximal possible value for $n$ at each application and until reaching terms corresponding to time steps strictly less than $M\alpha$ (which can be evaluated in constant time). Concretely if ${n}$ is the largest integer with ${M\alpha^{n}\leq t}$ we get
  \[F^{t}(c)_{\vec{z}} = \sum_{1\leq i\leq k} \overline{h_i}\circ F^{\alpha^nt_i+t-M\alpha^{n}}(c)_{\alpha^n\vec{z_i}+\vec{z}}\]
  and ${\alpha^nt_i+t-M\alpha^n\leq t\frac{(M-1)\alpha^n+t-M\alpha^n}{t}\leq t\frac{M\alpha-1}{M\alpha}}$ since ${t<M\alpha^{n+1}}$ by hypothesis on $n$ and the ratio ${\frac{(M-1)\alpha^n+t-M\alpha^n}{t}}$ is increasing with $t$.
  This shows that the depth of recursive calls in the algorithm is logarithmic in $t$ (at most ${\lceil\log_{\frac{M\alpha}{M\alpha-1}}(t)\rceil}$) and since the recursive branching width is constant (exactly $k$) it is actually doable in LOGSPACE. More precisely, it can be done in the following way:
  \begin{center}
    \small\tt
    \begin{itemize}
    \item ${m\leftarrow \lceil\log_{\frac{M\alpha}{M\alpha-1}}(t)\rceil}$;
    \item ${\mathtt{sum}\leftarrow 0}$ (\emph{identity of group $(Q,\oplus)$});
    \item for each ${b\in\{1,\ldots,k\}^m}$ do:
      \begin{itemize}
      \item ${h\leftarrow Id\in Q^Q}$;
      \item ${t'\leftarrow t}$;
      \item ${\vec{z'}\leftarrow \vec{z}}$;
      \item for each $i$ from $1$ to $m$ and while ${t'\geq M\alpha}$ do:
        \begin{itemize}
        \item ${h\leftarrow h\circ h_{b_i}}$;
        \item ${n\leftarrow \max\{i : M\alpha^{i} \leq t'\}}$;
        \item ${r\leftarrow t' - M\alpha^{n}}$ ;
        \item ${t'\leftarrow \alpha^nt_{b_i} + r}$;
        \item ${\vec{z'}\leftarrow \vec{z'}+\alpha^n\vec{z_i}}$
        \end{itemize}
      \item ${\mathtt{sum}\leftarrow \mathtt{sum} \oplus h\circ F^{t'}(c)_{\vec{z'}}}$ (\emph{bounded computation since ${t'<M\alpha}$})
      \end{itemize}
    \item return ${\mathtt{sum}}$
    \end{itemize}
  \end{center}
  The algorithm explores successively each branch of the tree of recursive calls (variable $b$) and for each of them (which is of depth at most $m$) it does a descent from the root to the leaf (variable $i$) and accumulate the sequence of endomorphisms to be applied at each level (variable $h$) while computing the new current position of space-time (variables $t'$ and $z'$). For the prediction problem of checking whether ${F^n(u)=q}$ for some ${u\in Q^{B_{rn}(0)}}$, we apply the above algorithm with ${t=n}$ and ${z=\vec{0}}$ so $m$ is logarithmic in the input size and all variables used have a logarithmic size. 
\end{proof}

\newcommand\sub[2]{\Phi_{#1}^{#2}}

\begin{lemm}
  \label{lem:substitucion}
  Let $F$ be any $d$-dimensional Abelian CA with a multi-scale additive identity. Then there
  exists a substitution of factor $\alpha$ describing space-time dependency, that is to say, there exists:
\begin{itemize}
\item $\alpha\geq 2$,
\item a finite set $E$,
\item $e : \Z^d\times\N \rightarrow E$,
\item $\Psi:E\rightarrow (Q\rightarrow Q)$,
\item ${T\in\N}$ and, for any ${0\leq t_0 < \alpha}$ and ${\vec{z_0}\in\Z^d}$ with ${\|\vec{z_0}\|_\infty<\alpha}$, a function ${\sub{\vec{z_0}}{t_0} : E\rightarrow E}$ such that, for any ${t\geq T}$
\[e(\alpha\vec{z} + \vec{z_0},\alpha t+t_0) = \sub{\vec{z_0}}{t_0}\bigl(e(\vec{z},t)\bigr)\]
\item $\Gamma_{\vec{z}}^t = \Psi(e(\vec{z},t))$
\end{itemize}
where $\Gamma_{\vec{z}}^t$ is the space-time dependency function given by:
\[\Gamma_{\vec{z}}^t : q\mapsto \sigma_{\vec{z}}\circ F^t (c^q).\]
\end{lemm}
\begin{proof}
  Taking the notations of Definition~\ref{def:multiscaleidentity} we suppose that for any ${n\in\N}$ and any configuration $c$ it holds: 
  \[F^{M\alpha^{n}}(c) = \sum_{1\leq i\leq k} \overline{h_i}\circ F^{t_i\alpha^n}\circ \sigma_{\alpha^n\vec{z_i}}(c)\]
  where the sum correspond to law $\overline{\oplus}$ over configurations.
  This identity extends to the space-time dependency function by choosing ${c=c^q}$ and composing both sides by a common translation and a common power of $F$, so that for any ${n\in\N}$ and any ${r\in\N}$ and any ${\vec{z}\in\Z^d}$ we have:
  \begin{equation}
  \Gamma_{\vec{z}}^{\alpha^{n}M+r}(c) = \sum_{1\leq i\leq k} h_i\circ \Gamma_{\alpha^n\vec{z_i}+\vec{z}}^{\alpha^n t_i+r}.\label{eq:gammamultiscale}
\end{equation}
  Applying this identity recursively, we can reduce any ${\Gamma_{\vec{z}}^t}$ to a sum of terms of the form ${h\circ\Gamma_{\vec{z'}}^{t'}}$ where ${t'<M\alpha}$. More precisely, we define the labeled $k$-regular DAG $D_F$ whose vertex set is ${\Z^d\times\N}$ and such that each vertex ${(\vec{z},t)}$
  \begin{itemize}
  \item is a leaf if ${t<M\alpha}$;
  \item has the following $k$ children:
    \[\chi_i(\vec{z},t) = (\alpha^n\vec{z_i}+\vec{z}, \alpha^n t_i+r)\]
    for ${1\leq i\leq k}$ where ${n=\max\{m : M\alpha^m\leq t\}}$ and ${r=t-M\alpha^n}$, and the edge ${e=\bigl((\vec{z},t),\chi_i(\vec{z},t)\bigr)}$ is labeled by ${\lambda(e) = h_i}$.
  \end{itemize}
  The multi-scale property of Equation \ref{eq:gammamultiscale} translates into $D_F$ as follows. Consider any ${\vec{z_0}\in\Z^d}$ and any ${t_0<\alpha}$.
  \newcommand\tsfo{\tau_{\vec{z_0},t_0}}
  If we denote by ${\tsfo:\Z^d\times\N\rightarrow\Z^d\times\N}$ the transformation such that ${\tsfo(\vec{z},t)=(\alpha\vec{z}+\vec{z_0},\alpha t+t_0)}$, then we have:
  \begin{itemize}
  \item ${\chi_i(\tsfo(\vec{z},t)) = \tsfo(\chi_i(\vec{z},t))}$ when ${(\vec{z},t)}$ is not a leaf;
  \item if ${e=\bigl((\vec{z},t),\chi_i(\vec{z},t)\bigr)}$ and ${e'=\bigl(\tsfo(\vec{z},t),\tsfo\chi_i(\vec{z},t)\bigr)}$ then ${\lambda(e)=\lambda(e')}$.
  \end{itemize}
  Indeed, as soon as ${t\geq\alpha M}$, if ${n=\max\{m:M\alpha^m\leq t\}}$, we have that\break ${n+1=\max\{m:M\alpha^m\leq\alpha t+t_0\}}$ and ${\alpha r + t_0= \alpha t + t_0 - M\alpha^{n+1}}$, where\break ${r=t-M\alpha^n}$.
  From this we deduce that to any path from ${(\vec{z},t)}$ to a leaf ${l\in\Z^d\times\N}$ corresponds a path from ${\tsfo(\vec{z},t)}$ to ${\tsfo(l)}$ with same labels. Conversely any path from ${\tsfo(\vec{z},t)}$ to some leaf admits as prefix a path from ${\tsfo(\vec{z},t)}$ to ${\tsfo(l)}$ where $l$ is a leaf. Formally, if ${P^{\vec{z},t}_{\vec{z'},t'}}$ denotes the set of path from ${(\vec{z},t)}$ to ${(\vec{z'},t')}$ in $D_F$ and $L$ the set of leaves we have 
  \[\bigcup_{l\in L}P^{\tsfo(\vec{z},t)}_l = \bigcup_{l\in L}P^{\tsfo(\vec{z},t)}_{\tsfo(l)}\cdot\bigcup_{l'\in L}P^{\tsfo(l)}_{l'}\] 
  where '$\cdot$' denotes the concatenation of paths.
  
  For any ${M\in\N}$, let ${Y_M = \{(\vec{z},t) : t<M\alpha\text{ and }\|\vec{z}\|_\infty\leq M\}}$. For any $\vec{z_0}$ with ${\|\vec{z_0}\|_\infty<\alpha}$ and ${t_0<\alpha}$ the (Euclidean) distance between ${(\vec{z},t)}$ and ${\tsfo(\vec{z},t)}$ goes to infinity as ${\|\vec{z}\|_\infty}$ grows. On the other hand, when ${t<M\alpha}$, the set of positions ${(\vec{z'},t')}$ reachable from ${(\vec{z},t)}$ in $D_F$ is finite and they are all at bounded (Euclidean) distance from ${(\vec{z},t)}$. This implies that for any large enough $M$, we have the following property: if there is a path in $D_F$ from some $\tsfo(l)$ to $Y_M$ with $l\in L$ and ${\|\vec{z_0}\|_\infty<\alpha}$ and ${t_0<\alpha}$, then ${l\in Y_M}$.
  Let now choose ${X = Y_M}$ with $M$ large enough to have the above property and also such that any ${(\vec{z},t)\in L}$ with  ${\Gamma_{\vec{z}}^t\neq 0}$ (\textit{i.e.} is not the constant map equal to $0$) belongs to $Y_M$. From Equation~\ref{eq:gammamultiscale} and by definition of $D_F$ and $X$, we have for any ${\vec{z}\in\Z^d}$ and any ${t\in\N}$:
  \begin{equation}
    \label{eq:gammadag}
    \Gamma_{\vec{z}}^t = \sum_{(\vec{z'},t')\in X}\sum_{\substack{\rho\in P^{\vec{z},t}_{\vec{z'},t'}\\\rho=e_1,\cdots,e_m}}\lambda(e_1)\circ\cdots\circ\lambda(e_m)\circ\Gamma_{\vec{z'}}^{t'}.
  \end{equation}
  
  We define ${E = \bigl(Q^Q\bigr)^X}$ and ${e : \Z^d\times\N \rightarrow E}$ by 
  \[e(\vec{z},t) =
    \begin{cases}
      (\vec{z'},t') \mapsto
      \begin{cases}
        id & \text{ if }(\vec{z'},t')=(\vec{z},t),\\
        0 & \text{ else.}
      \end{cases} & \text{ if $(\vec{z},t)$ is a leaf,}\\
      (\vec{z'},t') \mapsto \sum_{\substack{\rho\in P^{\vec{z},t}_{\vec{z'},t'}\\\rho=e_1,\cdots,e_m}}\lambda(e_1)\circ\cdots\circ\lambda(e_m)&\text{ else.}
    \end{cases}
  \]
  Then the map ${\Psi : E\rightarrow (Q\rightarrow Q)}$ defined for any ${f\in \bigl(Q^Q\bigr)^X}$ by 
  \[\Psi(f) = \sum_{(\vec{z},t)\in X}f(\vec{z},t)\circ\Gamma_{\vec{z}}^t\]
  is such that ${\Psi\bigl(e(\vec{z},t)\bigr) = \Gamma_{\vec{z}}^t}$ by definition of $e$ and Equation~\ref{eq:gammadag}.

  Finally, for any ${(\vec{z},t)\not\in L}$ and any $\vec{z_0}$ with ${\|\vec{z_0}\|_\infty<\alpha}$ and ${t_0<\alpha}$, we have: 
  \begin{equation}
    e(\tsfo(\vec{z},t)) = (\vec{z'},t') \mapsto \sum_{\substack{\rho\in P^{\tsfo(\vec{z},t)}_{\vec{z'},t'}\\\rho=e_1,\cdots,e_m}}\lambda(e_1)\circ\cdots\circ\lambda(e_m).\label{eq:ealpha}
  \end{equation}
  But, as said above, any path ${\rho\in P^{\tsfo(\vec{z},t)}_{\vec{z'},t'}}$ decomposes as a concatenation of a path ${\rho_1\in P^{\tsfo(\vec{z},t)}_{\tsfo(l)}}$ followed by a path ${\rho_2\in P^{\tsfo(l)}_{\vec{z'},t'}}$. Since ${(\vec{z'},t')\in X}$ and by choice of $X$ then ${l\in X}$. So ${\rho_1}$ is just the transformation under $\tsfo$ of a path from ${(\vec{z},t)}$ to ${l\in X}$ and this transformation doesn't change the labels $\lambda$. We can then rewrite Equation~\ref{eq:ealpha} as:
   \[e(\tsfo(\vec{z},t)) = (\vec{z'},t') \mapsto \sum_{y\in Y}\sum_{\rho\in P^{\vec{z},t}_y}\sum_{\rho'\in P^{\tsfo(y)}_{\vec{z'},t'}}\lambda(\rho)\cdot\lambda(\rho')\]
   or equivalently
   \[e(\tsfo(\vec{z},t)) = (\vec{z'},t') \mapsto \sum_{y\in Y}e(\vec{z},t)(y)\sum_{\rho'\in P^{\tsfo(y)}_{\vec{z'},t'}}\lambda(\rho').\]
   This shows that there is a map ${\sub{\vec{z_0}}{t_0}}$ not depending on ${\vec{z}}$ or $t$ which satisfies ${e(\tsfo(\vec{z},t)) = \sub{\vec{z_0}}{t_0}(e(\vec{z},t))}$. The lemma follows by letting ${T=\alpha M}$.
\end{proof}

The existence of this substitution has strong consequences on the
structure of traces: the trace of a finite configuration is determined
by a prefix of linear size in the distance of the farthest non-zero cell.
Let us first define some notation.  The \emph{size of a configuration $c\asy\overline{0}$} is the smallest $n\in\N$ such that $c(z)\neq 0$ imply $\|z\|_\infty\le n$.

\begin{lemm}
  \label{lem:prefijotraza}
  Let $F$ be any $d$-dimensional Abelian CA having admitting multi-scale additive identity. Let $m>0$
  and denote by $T_m$ the trace function associated to $F$ and
  $m$. There exist a function $\lambda:\N\rightarrow\N$ with
  $\lambda\in O(n)$ and such that for any $n$ and for any pair of configurations $c_1,c_2$ with:
  \begin{itemize}
  \item the size of $c_i$ is less than or equal to $n$,
  \item $T_m(c_1)(t) = T_m(c_2)(t)$ for any $t\leq\lambda(n)$,
  \end{itemize}
  then $T_m(c_1) = T_m(c_2)$.
\end{lemm}
\begin{proof}
  First $F$ fulfills the hypothesis of Lemma~\ref{lem:substitucion}
  so we have the existence of the substitution and adopt the notations
  of the lemma. 

 Let's focus on the substitution given by the function $e$ and
  consider $k\geq 0$, $t\geq \alpha^kT$, and $\vec{z}\in\Z^d$ with
  $\|z\|_\infty\le \alpha^k$. 
  By ${k-1}$ applications of the
  substitution we get the following expression for $e(\vec{z},t)$:
  \[e(\vec{z},t) = \sub{\vec{z}\bmod \alpha}{t \bmod
    \alpha}\circ\cdots\circ\sub{\vec{z}/\alpha^{k-1}\bmod \alpha}{t/\alpha^{k-1}\bmod
    \alpha}\bigl(e(\rho(\vec{z}),t/\alpha^k)\bigr)\]
  where $\|\rho(\vec{z})\|_\infty\leq \alpha$, and where the
  division/modulus correspond to the standard Euclidean division on
  $\Z^d$.
  
  The sequence of superscripts in this expression only depends on
  $t\bmod \alpha^k$. The sequence of subscripts depends only on $\vec{z}$. Therefore we can write this functional dependency of
  $e(\vec{z},t)$ on $e(\rho(\vec{z}),t/\alpha^k)$ in the
  following way:
  \begin{equation}
  e(\vec{z},t)= \chi_{\vec{z}}^{t\bmod \alpha^k}\bigl(
  e(\rho(\vec{z}),t/\alpha^k)\bigr).
  \label{eq:prefix}
  \end{equation}
  
  Now consider a time $t_0$ sufficiently large to see before time $t_0$ any possible
  vector of the form $({e(\vec{z_0},t)})_{\|\vec{z_0}\|_\infty\leq \alpha}$ that occur after time $T$, precisely: 
  \[\forall t\geq T, \exists t',T\leq t'\leq t_0, \forall \vec{z_0},
  \|\vec{z_0}\|_\infty\leq \alpha : e(\vec{z_0},t) = e(\vec{z_0},t').\]
  Given an index set $I$, consider any tuple ${(\vec{z}_i)_{i\in I}}$, with ${\|\vec{z}_i\|_\infty\leq
    \alpha^k}$, and any $\mathcal{P}\subseteq E^I$. For any time $t$, we
  can define the property $\mathcal{P}_I(t)$ by:
  \[\mathcal{P}_I(t) \Leftrightarrow \bigl(e(\vec{z}_i,t)\bigr)_{i\in
    I}\in \mathcal{P}.\]
  \emph{Claim}: if $\mathcal{P}_I(t)$ holds for every
  ${t\leq (t_0+1)\cdot \alpha^k}$ then $\mathcal{P}_I(t)$ holds for every $t\in\N$.

  Indeed, take some time ${t> (t_0+1)\cdot \alpha^k}$. Then by choice of
  $t_0$ there exists ${t'\leq t_0}$ such that:
\[\forall \vec{z_0},  \|\vec{z_0}\|_\infty\leq \alpha : e(\vec{z_0},t/\alpha^k)
= e(\vec{z_0},t').\]
Now we can choose ${t''\leq (t_0+1)\cdot \alpha^k}$ with 
\begin{align*}
  t''/\alpha^k&=t'\text{ and}\\
  t''\bmod \alpha^k &= t\bmod \alpha^k
\end{align*}
and equation~\ref{eq:prefix} yields the equalities:
\[e(\vec{z_i},t) = \chi_{\vec{z}}^{t\bmod \alpha^k}\bigl(
  e(\rho(\vec{z_i}),t/\alpha^k)\bigr) = \chi_{\vec{z}}^{t''\bmod \alpha^k}\bigl(
  e(\rho(\vec{z_i}),t')\bigr) = e(\vec{z_i},t'').\]
It shows that ${\mathcal{P}_I(t)\Leftrightarrow\mathcal{P}_I(t'')}$
and the claim follows.

  Since the space-time dependency function is completely determined
  by the substitution $e$ (Lemma~\ref{lem:substitucion}), the fact
  that the trace of a finite configuration at time $t$ is null can be
  expressed by a property of the form $\mathcal{P}_I(t)$. 
  More precisely, for any configuration $c$ of size $\alpha^k$, we can define
 \begin{eqnarray*}
 D&=&\{z\in\Z^d\ :\ c(z)\neq 0\},\\
I&=&\displaystyle{\bigcup_{z\in D} B_m(z)},\\
\mathcal{P}&=&\{ f\in E^I\ :\ \forall x\in B_m(0), \displaystyle{\sum_{z\in D}\Psi(f_{z+x})(c(z))=0}\}.
\end{eqnarray*}
We then have that
\begin{eqnarray*}
\mathcal{P}_I(t)&\Leftrightarrow& \forall x\in B_m(0), \displaystyle{\sum_{z\in D}\Psi(e(z+x,t))(c(z))=0}\\
&\Leftrightarrow& \forall x\in B_m(0), \displaystyle{\sum_{z\in D}\Gamma_{z+x}^t(c(z))=0}\\
&\Leftrightarrow& \displaystyle{\sum_{z\in D}T_m(\sigma_{-z}(c^{c(z)}))_t=0}\\
&\Leftrightarrow& T_m(c)_t=0.
\end{eqnarray*}  
  We deduce that if $F^t(c)$ is null on $B_m(0)$ until time
  $(t_0+1)\cdot \alpha^k$ then it is null forever. By linearity of $F$,
  equality of two traces is equivalent to nullity of their difference.
  We have thus shown the lemma for $m\ge0$ by choosing $\lambda(n)
  = (t_0+1)\cdot \alpha^k$ for $k=\lceil log_\alpha(n)\rceil$.
\end{proof}

\section{Pre-expansivity}
\label{sec:prexp}

Pre-expansivity is the property of positive expansivity restricted to
asymptotic pairs of configurations.

\begin{defi}
  Let $F$ be a cellular automaton
  over $Q^\G$.
  $F$ is \emph{pre-expansive} if:
  \[\exists \delta>0: \forall c,d\in Q^\G, c\not=d\text{ and }c\asy d
  \Rightarrow \exists t\in\N,\dist(F^t(c),F^t(d))>\delta.\]
The value $\delta$ is the \emph{pre-expansivity constant}.
\end{defi}

\begin{rem}\
  
  \begin{itemize}
  \item Cellular automata can be seen as examples of two continuous commuting actions on a metric space: the spatial action $S$ (the $\G$-shift) and the temporal action $F$ (the cellular automaton itself). The definition of pre-expansivity can be adapted to this general settings by requiring that $F$ be expansive on pairs of $S$-asymptotic configurations. This goes far beyond the scope of the present paper which focuses on cellular automata. 
  \item Pre-expansivity is a conjugacy invariant: if a cellular automata $F$ and $F'$ on $Q^\G$ are conjugated via $\phi$ (i.e. $\phi$ is an automorphism of the full-shift $Q^\G$ with ${F\circ\phi = \phi\circ F'}$), then $F$ is pre-expansive if and only if $F'$ is. Indeed, ${c\asy d}$ is equivalent to ${\phi(c)\asy\phi(d)}$ and for all $\epsilon$ there is $\delta$ such that ${\dist(\phi(c),\phi(d))>\epsilon}$ implies ${\dist(c,d)>\delta}$. 
  \end{itemize}
\end{rem}

The notion of pre-expansivity can be further refined by considering only pairs of
configurations with a fixed finite number of differences.
\newcommand\di[1]{\,{\neq}_{#1}\,}
Given ${c,d\in Q^\G}$, we denote ${c\di{k} d}$ if
${\#\bigl\{z\in\G: c(z)\not=d(z)\bigr\}=k}$, \textit{i.e.} if $c$ and $d$ differ in exactly $k$ positions.

\begin{defi}
  Let $F$ be a cellular automaton
  over $Q^\G$ and let $k>0$.
  $F$ is \emph{$k$-expansive} if:
  \[\exists \delta>0: \forall c,d\in Q^\G, c\di{k}d
  \Rightarrow \exists t\in\N,\Delta(F^t(c),F^t(d))>\delta.\]
\end{defi}

\begin{prop}
  \label{prop:general} %amenalble solo al final
  Let $F$ be any CA over $Q^\G$, it holds:
  \begin{enumerate}
  \item $F$ is pre-expansive $\Rightarrow$ $\forall k>0$ $F$ is $k$-expansive,
  \item $F$ is $k$-expansive $\Rightarrow$ $F$ is sensitive to initial configurations,
  \item $F$ is pre-expansive $\Leftrightarrow$ $T_m$ is pre-injective
    for some $m$.
  \item If $L$ is a CA over $\tilde{Q}^\G$, then $F\times L$ is $k'$-expansive for every $k'\le k$ if and only if $L$ and $F$ are $k'$-expansive for every $k'\le k$.
  \item ${F \text{positively expansive} \Rightarrow F\text{
        pre-expansive}\Rightarrow F\text{ pre-injective}}$, moreover if $\G$ is amenable then ${F\text{ pre-expansive}\Rightarrow
      F\text{ surjective}}$.
  \end{enumerate}
\end{prop}
\begin{proof}~
  \begin{enumerate}
  \item It follows directly from definitions.
  \item It is enough to note that for any
    configuration $c$, any ${\delta>0}$ and any ${k\geq 1}$ there
    always exist a configuration $c'$ with ${c\di{k}c'}$ and
    ${\dist(c,c')\leq\delta}$.
  \item For the third item, it is sufficient to note that the existence
    of some time $t$ such that
    ${\dist\bigl(F^t(c),F^t(c')\bigr)>\delta}$ is equivalent to
    ${T_m(c)\not= T_m(c')}$ for a suitable choice of $m$.
  \item If $F\times L$ is $k'$-expansive, it is enough to take two configurations with their differences in only one of their components, since both automata act independently, $k'$-expansivity of $F\times L$ imply that the perturbations will arrive to the center at the same component, proving the $k'$-expansivity of the corresponding automaton.\\
  If now $L$ and $F$ are $k'$-expansive for every $k'\le k$, we take two configurations with $k'$ differences. They may lay in one or both of their components, in any case there will be $0< k''\le k'$ differences in one of the components of $F\times L$. By the $k''$-expansivity of the corresponding automaton, we show the $k'$-expansivity of $F\times L$.
  \item It is clear that positive expansivity implies
    pre-expansivity (restriction of the universal
    quantification). Then pre-expansivity implies pre-injectivity
    because if there is a pair of configurations $c,c'$ with ${c\asy
      c'}$ and ${F(c)=F(c')}$ then, eventually applying a translation,
    we can also suppose them such that ${\dist(c,c')}$ is arbitrarily
    small. Finally, if $\G$ is amenable we also have that pre-injectivity
    implies surjectivity by Garden of Eden Theorem \cite{cagroups}.
  \end{enumerate}
\end{proof}

Note however that $k$-expansivity does not generally imply
pre-injectivity or surjectivity as shown by the following example.

\begin{prop}
  \label{prop:onexpnonsurj}
  For any $k\geq1$ there exists a one-dimensional CA which is not surjective but
  $k'$-expansive for any ${k'\leq k}$.
\end{prop}
\begin{proof}
  Consider any pre-expansive one-dimensional CA $F$ of radius ${1}$
  over state set $Q=\{0,1\}$ (for instance a bi-permutative CA), and
  define a CA $\Psi$ over state set $Q^{k+1}$ as follows. It has $k+1$
  ``layers'' and to any configuration $c$ we associate its projection
  $\pi_i(c)$ on the $i$th layer. Intuitively it behaves on the $k$
  first layers as $k$ independent copies of $F$, except that the
  $(k+1)$th layer induce a state flip in the image in the following
  way: if it has a $1$ at position $z$ then, in the image, layer $i$
  is flipped at position $z+3i$. Moreover, $(k+1)$th layer is
  uniformly reset to $0$ after one step. Formally, $\Psi$ is defined by:
  \[\Psi(c)_z = \bigl(F(\pi_1(c))_z + \pi_{k+1}(c)_{z-3} \bmod 2, \ldots,
  F(\pi_k(c))_z + \pi_{k+1}(c)_{z-3k} \bmod 2, 0\bigr)\]

  First it is clear from the definition that it is not surjective
  since the image of any configuration is always $0$ on layer $k+1$.
  Note also that, when reduced to state set ${\{0,1\}^k\times\{0\}}$,
  $\Psi$ is isomorphic to $F^k$ which is pre-expansive. Therefore, to
  show that $\Psi$ is $k'$-expansive for any ${1\leq k'\leq k}$, it is sufficient to show that for any
  pair of configurations $c$ and $d$ with ${c\di{k'}d}$ we have
  ${\Psi(c)\neq \Psi(d)}$.

  So consider such a pair $(c,d)$. $\Psi$ was defined such that, if
  $c$ and $d$ differ on the $(k+1)$th layer at position $z$, then, on
  the $i$th layer, $F(c)$ and $F(d)$ will differ at position ${z+3i}$
  as soon as $c$ and $d$ are the same on the $i$th layer at positions
  ${z+3i-1}$, ${z+3i}$ and ${z+3i+1}$. Therefore, supposing that $c$
  and $d$ indeed differ on the $(k+1)$th layer at position $z$, it
  implies that $F(c)$ and $F(d)$ differ because $c$ and $d$ having
  only ${k'\leq k}$ differences, they can not differ at $z$ and at one
  of the positions ${z+3i-1}$, ${z+3i}$ or ${z+3i+1}$ for each ${1\leq
    i\leq k}$.

  Finally, suppose that $c$ and $d$ are equal on the $(k+1)$th
  layer. Then they must differ on some layer $i$ with ${1\leq i\leq
    k}$. Therefore, we must have ${F(\pi_i(c))\neq F(\pi_i(d))}$. We
  deduce that ${\Psi(c)\neq \Psi(d)}$ because their respective $i$th layers
  are $F(\pi_i(c))$ and $F(\pi_i(d))$ up to some modification by the
  $(k+1)$th layer which are identical in $c$ and $d$.
\end{proof}

%Let us remark that the CA defined in the last proof is $l$-expansive for every $l\le k$.

The next lemma talks about \emph{linear} CA.
When $F$ is supposed to be linear (for law $\gplus$), then $T_m$ is
also linear, \textit{i.e.} $T_m(c\gpconf d) = T_m(c)\gpconf T_m(d)$
where $\gpconf$ denotes the component-wise application of $\gplus$
either on $Q^\G$ or on $\bigl(Q^{B_m}\bigr)^\N$.

\begin{prop}
  \label{prop:linear}
  Let $F$ be a linear CA for law $\gplus$ and neutral element $0$. Let
  $I$ be the set 
  \[I = \left\{k\in \N : F\text{ is not $k$-expansive}\right\}\]
  \begin{itemize}
  \item\label{item:ideal} if ${k_1,k_2\in I}$ then ${k_1+k_2\in I}$,
  \item $F$ is pre-expansive if and only if for some $m>0$ there is no
    finite sequence $(g_1,\ldots, g_n)$ of different cells in $\G$ and states $(q_1,\ldots,q_n)$  in $Q$  such that
    \[T_m(\sigma_{g_1}(c^{q_1}))\gpconf\cdots\gpconf T_m(\sigma_{g_n}(c^{q_n})) = \overline{0}.\]
  \end{itemize}
\end{prop}
\begin{proof}
  First, by linearity of the trace functions $T_m$, we have that
  ${T_m(c)=T_m(c')}$ if and only if
  ${T_m\bigl(c'\gpconf(-c)\bigr)=T_m(\overline{0})}$ where $-c$ is the
  configuration such that ${c\gpconf (-c)=\overline{0}}$. Moreover we
  also have ${c\di{k}c'}$ if and only if ${c'\gpconf
    (-c)\di{k}\overline{0}}$. Hence $F$ is pre-expansive
  (resp. $k$-expansive) if and only if there is $m$ such that no
  ${c\not=\overline{0}}$ with ${c\asy\overline{0}}$
  (resp. ${c\di{k}\overline{0}}$) can verify ${T_m(c)=T_m(\overline{0})}$.

  From this we deduce the second item of the proposition.

  For the first item, consider $k_1$ and $k_2$ in $I$. From what we
  said above, for any
  $m_1$ there is ${c_1}$ such that
  ${c_1\di{k_1}\overline{0}}$ and
  ${T_{m_1}(c_1)=T_{m_1}(\overline{0})}$. Now choose $m_2$ large
  enough so that any non-zero state of $c_1$ appears at distance at
  most $m_2$ from the center. 
  Let us remark that the differences between $c_1$ and $\overline{0}$ are outside $B_{m_1}$, otherwise $T_{m_1}(c_1)\not = T_{m_1}(\overline{0})$, thus $m_2>m_1$.
  Since ${k_2\in I}$ we deduce from what
  we said earlier that there is ${c_2}$ such that
  ${c_2\di{k_2}\overline{0}}$ and
  ${T_{m_2}(c_2)=T_{m_2}(\overline{0})}$.
  By our choice of $m_2$, this
  implies that ${T_{m_1}(c_1\gpconf
    c_2)=\T_{m_1}(\overline{0})}$. Moreover ${c_1\gpconf
    c_2\di{k_1+k_2}\overline{0}}$. Since $m_1$ was arbitrary, we
  deduce that $F$ is not ${(k_1+k_2)}$-expansive.
\end{proof}

\section{1-dimensional Cellular Automata}
\label{sec:1D}

The 1-dimensional setting is particular in our study since it allows examples of all properties considered in this paper, and gives additional structure to analyze them.

The first goal of this section is to show that the notion of
$k$-expansivity, pre-expansivity and positive expansivity all differ
and interact differently with properties of bijectivity and
surjectivity. More precisely we will show the following existential
result.

\newcommand\bij{\mathbb{B}}
\newcommand\surj{\mathbb{S}}
\newcommand\onex{\mathbb{X}_1}
\newcommand\prex{\mathbb{X}_{pre}}
\newcommand\posx{\mathbb{X}_{pos}}

% Puto Latex! Porque tengo que preocuparme de eso, es tu pega!
\newcommand\Ts{\rule{0pt}{2.6ex}}       % Top strut
\newcommand\Bs{\rule[-1.2ex]{0pt}{0pt}} % Bottom strut

\begin{theo}\label{theo:examples}  
  Let $\bij$ and $\surj$ denote the set of bijective and surjective CA respectively. Let $\onex$, $\prex$ and $\posx$ denote the set of $1$-expansive, pre-expansive and positively expansive CA respectively. It holds:
  \begin{itemize}
  \item ${\prex\cap(\bij\setminus\posx)\neq\emptyset}$,
  \item ${\prex\setminus(\bij\cup\posx)\neq\emptyset}$,
  \item ${\onex\cap(\bij\setminus\prex)\neq\emptyset}$,
  \item ${\onex\setminus(\prex\cup\bij)\neq\emptyset}$,
  \item ${\onex\setminus\surj\neq\emptyset}$.
  \end{itemize}
  We therefore have the following situation:
  \begin{center}
    \begin{tabular}{r|c|c|c}
     $\cap$ &$\bij$ & $\surj\setminus\bij$ & $\surj^c$ \Bs\\
      \hline
      $\onex\setminus\prex$ & $\exists$ & $\exists$ & $\exists$ \Ts\\
      $\prex\setminus\posx$ & $\exists$ & $\exists$ & $\emptyset$\Ts\\
      $\posx$ & $\emptyset$ & $\exists$ & $\emptyset$\Ts
    \end{tabular}
  \end{center}
\end{theo}

For this purpose it will be sufficient to focus on linear cellular
automata. At the end of the section, we consider a well-known class of non-linear
bijective CA where pre-expansivity is a relevant property. 
But before the study of examples, we give some additional results which
hold in dimension 1.

\subsection{Left/right propagation and directional dynamics}

The pre-expansivity constant can be fixed canonically as we will prove in Lemma~\ref{lem:lr}.
The next lemma is direct and it expresses the locality of CAs.

\begin{lemm}\label{th:locality}
Let $F$ be a CA in $\Z$ with neighborhood $[-l,r]$ the next assertions hold.
\begin{itemize}
\item If $c_{]-\infty,n]}=d_{]-\infty,n]}$ and there exists an iteration $t$ such that $F^t(c)_{]-\infty,n]}\not = F^t(d)_{]-\infty,n]}$, then there is an iteration $t'\le t$ such that $F^{t'}(c)_{[n-r,n]}\not = F^{t'}(d)_{[n-r,n]}$.
\item If $c_{[n,\infty[}=d_{[n,\infty[}$ and there exists an iteration $t$ such that $F^t(c)_{[n,\infty[}\not = F^t(d)_{[n,\infty[}$, then there is an iteration $t'\le t$ such that $F^{t'}(c)_{[n,n+l]}\not = F^{t'}(d)_{[n,n+l]}$.
\end{itemize}
\end{lemm}
\begin{proof}
We will only prove the first assertion, the second one is completely analogous.
Let $t'$ be the first time such that $F^{t'}(c)_{]-\infty,n]}\not = F^{t'}(d)_{]-\infty,n]}$, and let $i\in ]-\infty,n]$ be a position such that  $F^{t'}(c)_{i}\not = F^{t'}(d)_{i}$.
Since $F^{t'-1}(c)_{]-\infty,n]} = F^{t'-1}(d)_{]-\infty,n]}$, and only the cells in $]n-r,n]$ depend on cells in $]n,\infty[$, $i\ge n-r$.
\end{proof}

This lemma shows a particularity of dimension 1: expansivity properties can be understood through left/right propagation of information.
Let us precise this notion.

\begin{defi}
Given two configurations $c\not = d$, and a CA $F$, we define the left and right propagation sequences as follows.
\begin{align*}
  l^d_t(c) &= \inf \{z\in\Z : \bigl(F^t(c)\bigr)(z)\neq \bigl(F^t(d)\bigr)(z)\}\\
  r^d_t(c) &= \sup \{z\in\Z : \bigl(F^t(c)\bigr)(z)\neq \bigl(F^t(d)\bigr)(z)\}
\end{align*}
\end{defi}

Note that if ${c\asy d}$ then ${l^d_t(c)}$ and ${r^d_t(c)}$ are always finite integers.

\begin{lemm}\label{lem:lr}
Given a CA $F$ of neighborhood $[-l,r]$ and $k\in\N$,  the next assertions hold.
\begin{enumerate}
\item If $F$ is $k$-expansive, then $\forall c\di{k} d, (l_t^d(c))_{t\in\N}$ is not lower bounded and $(r_t^d(c))_{t\in\N}$ is not upper bounded.
\item If $\forall k'\le k, \forall c\di{k'} d, (l_t^d(c))_{t\in\N}$ is not lower bounded and $(r_t^d(c))_{t\in\N}$ is not upper bounded, then $F$ is $k$-expansive with pre-expansivity constant $2^{-max\{l,r\}}$.
\end{enumerate}
\end{lemm}
\begin{proof}
\begin{enumerate}
\item Let us suppose that $F$ is $k$-expansive with pre-expansivity constant $m$.
Let $c\di{k} d$ be two configurations, and let us assume that $l_t^d(c)> n$ for some $n\in \Z$;
this means that $F^t(c)_{]-\infty,n]}=F^t(d)_{]-\infty,n]}$ for every $t\in\N$.
Thus $T_m(\sigma_{n-m}(c))=T_m(\sigma_{n-m}(d))$, which is a contradiction.
The analogous happens if $(r_t^d(c))_{t\in\N}$ is upper bounded.
\item We need to prove that $F$ is $k$-expansive with pre-expansivity constant $2^{-m}=2^{-max\{l,r\}}$.
Let $c\di{k}d$ be two configurations, and let us define the next two additional configurations.

\begin{align*}
c^l(i)&=\left\{\begin{array}{rr} c(i)&\textrm{if } i<0\\ d(i)&\textrm{if } i\ge 0\end{array}\right.\\
c^r(i)&=\left\{\begin{array}{rr} d(i)&\textrm{if } i<0\\ c(i)&\textrm{if } i\ge 0\end{array}\right.
\end{align*}

If $c_{[-m,m]}\neq d_{[-m,m]}$, $T_m(c)\neq T_m(d)$ and we are done, so let us suppose that $c_{[-m,m]}=d_{[-m,m]}$.
Let $k'$ and $k''$ be such that $c^r\di{k'}d$, $c^l\di{k''}d$ and $k'+k''=k$.

If $k'\neq 0$, $(l_t^d(c^r))_{t\in\N}$ is not lower bounded, thus by Lemma~\ref{th:locality} and the fact that $c^r$ is equal to $d$ below position $m$, there is a minimal iteration $t_r$ such that $F^{t_r}(c^r)_{[m-r,m]}\neq F^{t_r}(d)_{[m-r,m]}$.
Analogously, if $k''\neq 0$ there is a minimal iteration $t_l$ such that $F^{t_l}(c^l)_{[-m,-m+l]}\neq F^{t_l}(d)_{[-m,-m+l]}$.

Let us take $\overline{t}=min\{t_r,t_l\}$, by the choice of $m$ we have that $F^{\overline{t}}(c)_{[0,m]}=F^{\overline{t}}(c^r)_{[0,m]}$ and $F^{\overline{t}}(c)_{[-m,0]}=F^{\overline{t}}(c^l)_{[-m,0]}$, and at least one of them is different from $F^{\overline{t}}(d)$ between $[-m,m]$, thus $T_m(c)\neq T_m(d)$.
\end{enumerate}
\end{proof}

The last lemma establishes that the pre-expansivity constant is uniform in dimension one (it does not depends on $k$), thus we can conclude the next corollary.

\begin{coro}\label{th:e-constant}
If $F$ is $k$-expansive for every $k\in\N$ then it is pre-expansive.
\end{coro}

\begin{lemm} 
$F$ is pre-expansive if and only if for all $c\asy d$, $(l_t^d(c))_{t\in\N}$ is not lower bounded and $(r_t^d(c))_{t\in\N}$ is not upper bounded.
\label{lem:leftright}
\end{lemm}
\begin{proof}
\begin{itemize}
\item[($\Rightarrow$)] Let $c\di{k} d$ be two asymptotic configurations. Since $F$ is pre-expansive, it is also $k$-expansive, thus by Lemma~\ref{lem:lr}, $(l_t^d(c))_{t\in\N}$ is not lower bounded and $(r_t^d(c))_{t\in\N}$ is not upper bounded.
\item[($\Leftarrow$)] If for every $k\in\N$ and every pair $c\di{k} d$ we have that $(l_t^d(c))_{t\in\N}$ is not lower bounded and $(r_t^d(c))_{t\in\N}$ is not upper bounded, then by Lemma~\ref{lem:lr}, $F$ is $k$-expansive, and by Corollary~\ref{th:e-constant}, we conclude that $F$ is pre-expansive.
\end{itemize}
\end{proof}

Left and right propagation determines also expansivity. The next lemma can be proven by using the techniques from the last lemmas.

\begin{lemm}\label{th:exp}
$F$ is positively expansive if and only if for any pair of different configurations $c,d$, if $c_{]-\infty,0]}=d_{]-\infty,0]}$ then $(l_t^d(c))_{t\in\N}$ is not lower bounded, and if $c_{[0,\infty[}=d_{[0,\infty[}$ then $(r_t^d(c))_{t\in\N}$ is not upper bounded.
\end{lemm}

The last two lemmas show a similarity in information propagation between pre-expansivity and positive expansivity: the differences between two configurations are spread both to the left and to the right and thus there is a sensitivity to initial conditions in all directions. 
 We can formalize this following the directional dynamics setting of~\cite{Sablik08} (we could also use the more general viewpoint of~\cite{DelacourtPST11} but we prefer to keep a lighter setting for clarity of exposition). Let ${\alpha\in\R}$. A CA $F$ is said to be pre-expansive in \emph{direction $\alpha$} if there is some ${\delta>0}$ such that
\[\forall c,d\in Q^\Z, c\not=d\text{ and }c\asy d
  \Rightarrow \exists t\in\N,\dist(\sigma_{\lceil\alpha t\rceil}\circ F^t(c), \sigma_{\lceil\alpha t\rceil}\circ F^t(d))>\delta.\]
In particular, pre-expansivity in direction $0$ is the same as pre-expansivity. Similarly $F$ is said to be sensitive to initial conditions in direction $\alpha$ if there is some ${\delta>0}$ such that 
\[\forall c\in Q^\Z, \forall\epsilon>0, \exists d\in Q^\Z\exists t\in\N : \dist(c,d)<\epsilon \wedge\dist(\sigma_{\lceil\alpha t\rceil}\circ F^t(c), \sigma_{\lceil\alpha t\rceil}\circ F^t(d))>\delta.\]
% exists d a distancia epsilon

\begin{rem}\label{rem}
  Lemma~\ref{lem:leftright} generalizes to the directional dynamics setting: \\
  $F$ is pre-expansive in direction $\beta$ if and only if for all $c\asy d$, ${(l_t^d(c)-\lceil\beta t\rceil)_{t\in\N}}$ is not lower bounded and ${(r_t^d(c)-\lceil\beta t\rceil)_{t\in\N}}$ is not upper bounded.\\
  In particular, if the CA neighborhood is $N\subseteq\{l,\ldots,r\}$, then its set of pre-expansivity directions is included in $]-r,-l[$.
\end{rem}
%Remark!!

\begin{prop}
  \label{prop:directional}
  Let $F$ be any CA which is pre-expansive in some direction $\alpha$. Then the following holds:
  \begin{enumerate}
  \item $F$ is sensitive to initial conditions in any direction;
  \item if $F$ is also pre-expansive in direction ${\alpha'>\alpha}$ then it is pre-expansive in any direction $\beta$ with ${\alpha\leq\beta\leq\alpha'}$.
  \end{enumerate}
\end{prop}
\begin{proof}
  We first deduce that if ${\alpha\leq\beta\leq\alpha'}$ and $\alpha'$ is a direction of pre-expansivity, then for any ${c\neq d}$ with ${c\asy d}$:
  \[l_t^d(c)-\lceil\beta t\rceil\leq l_t^d(c)-\lceil\alpha t\rceil\]
  is not lower bounded and 
  \[r_t^d(c)-\lceil\beta t\rceil\geq r_t^d(c)-\lceil\alpha' t\rceil\]
  is not upper bounded so $\beta$ is also a direction of pre-expansivity.

  Finally, consider a direction ${\beta\geq\alpha}$ (the case ${\beta\leq\alpha}$ is symmetric), some configuration $c$ and some ${\epsilon>0}$. Take any ${d\neq c}$ with ${c\asy d}$ and ${l_0^d(c)}$ large enough so that ${\dist(c,d)\leq\epsilon}$. Since ${(l_t^d(c)-\lceil\alpha t\rceil)_{t\in\N}}$ is not lower bounded and ${\beta\geq\alpha}$ then there is some $t$ with ${l_t^d(c)-\lceil\beta t\rceil\leq 0}$. Let $t_0$ be the smallest such $t$. Since ${|l_{t+1}^d(c) -l_t^d(c)|}$ is bounded by the radius of $F$ we deduce that there exists a constant $M$, depending only on $\beta$ and the radius of $F$, such that ${M\leq l_{t_0}^d(c)-\lceil\beta t\rceil\leq 0}$. Said differently ${\dist(\sigma_{\lceil\beta t_0\rceil}\circ F^{t_0}(c), \sigma_{\lceil\beta t_0\rceil}\circ F^{t_0}(d))\geq 2^{M}}$. We have thus shown that $F$ is sensitive in direction $\beta$.
\end{proof}%redacción constante M

\subsection{Hierarchy of expansive-like properties}
 
The following proposition shows that the simplest form of linearity
is not sufficient to achieve the separation between positive expansivity and
pre-expansivity. 
Let us first introduce some notation.
%If 0 is the neutral element of $Q$, we denote by $\overline{0}$ the configuration identically equal to 0.
%\ch{defini estas confs, ojo elegi $c^a$ en vez de $c_a$ para no confundirlo con ``$c$ evaluado en $a$'', pero estoy abierta a cambiar}

\begin{prop}\label{prop:cyclic1D}
  Let $\Z_n$ be the group of integers modulo $n$ with addition,
  and let $F$ be a one-dimensional linear CA over $\Z_n$. Then $F$ is
  $1$-expansive if and only if it is positively expansive.
\end{prop}
\begin{proof}
  First by Proposition~\ref{prop:general} if $F$ is positively expansive it is
  in particular $1$-expansive.

  For the other direction, it is sufficient to consider the case
  ${n=p^k}$ with $p$ a prime number by Lemma~\ref{lem:groupdecomp}
  because if some ${F_1\times F_2}$ is $1$-expansive then both $F_1$
  and $F_2$ must be $1$-expansive.
  
%  Given $a\in \Z_{p^k}$, let us call $c_a$ the configuration that is 0 everywhere except at cell 0 where its state is $a$.
  By commutation with shifts we have $l^{\overline{0}}_t(\sigma_n(c^a)) = -n+l^{\overline{0}}_t(c^a)$ and the analogous for $r_t$. 
  Moreover ${c^a = a\cdot c^1}$ because we are on a cyclic group.

  Let us define $l_t^U(c^a)=min\{i\in\Z\ |\gcd(F^t(c^a)_i,p)=1\}$.
  The next properties hold.
  \begin{enumerate}
  \item \emph{If $\gcd(a,p)\not =1$, then $l_t^U(c^a)=\infty$ (respectively $r_t^U(c^a)=-\infty$)}. 
In fact, in  this case the entire evolution of $F$ over $c^a$ is composed by multiples of $a$, which are multiples of $p$ too.
  \item \emph{If $\gcd(a,p)=1$, then $l_t^U(c^a)=l_t^U(c^1)$ (respectively $r_t^U(c^a)=r_t^U(c^1)$)}. 
In fact, in this case $F^t(c^a)_i=aF^t(c^1)_i$ is coprime with $p$ if and only if $F^t(c^1)_i$ is.
  \item \emph{$l_t^{\overline{0}}(c^a)\le l_t^U(c^1)$ (respectively $r_t^{\overline{0}}(c^a)\ge r_t^U(c^1)$)}. 
In fact, if $F^t(c^1)_i$ is coprime with $p$, then $F^t(c^a)_i=aF^t(c^1)_i$ is not null.
  \item \emph{$l_t^{\overline{0}}(c^{p^{k-1}})=l^U_t(c^1)$ (respectively $r_t^{\overline{0}}(c^{p^{k-1}})=r^U_t(c^1)$)}. 
In fact, $F^t(c^{p^{k-1}})_i=p^{k-1}F^t(c^1)_i=0\Leftrightarrow \gcd(F^t(c^1)_i,p)\not=1$.
  \end{enumerate}
  From the last assertion and Lemma~\ref{lem:lr} (first item) we have that $l^U_t(c^1)$ is not lower bounded and $r^U_t(c^1)$ is not upper bounded.

  Now let us take a configuration $v\in(\Z_{p^k})^\Z$, such that $v_{i}=0$ for every $i< 0$ and $v_0\neq 0$.
  Let us define $j=max\{i\in\{0,...,k\}\ |\ \forall x\in\Z, p^i | v(x)\}$, and let us consider $u=v/p^j$.
  In this way, there is $y\geq 0$ with ${u(y)\neq 0}$ and
  ${\gcd(p,u(y))=1}$. Let $y$ be the smallest integer with this property.
  
  $$F^t(u)_{l^U_t(c^1)+y}=\sum_{x=0}^{l^U_t(c^1)+y+rt}F^t(c^{u_x})_{l^U_t(c^1)+y-x}$$

  But $u(x)$ is a multiple of $p$ when $x< y$, thus ${F^t(c^{u_x})_i =
    0 \bmod p}$ for any $i$ and any ${x<y}$.
  For $x>y$, $F^t(c^{u_x})_{l^U_t(c^1)+y-x}$ is also a multiple of $p$ because the smallest index for which $F^t(c^{u_x})$ is coprime with $p$ is $l^U_t(c^{u_x})$ which is greater than or equal to $l^U_t(c^1)$.
  In this way we conclude that \[F^t(u)_{l^U_t(c^1)+y}=F^t(c^{u_y})_{l^U_t(c^1)}\bmod p\] is coprime with $p$ and is non null. 
  Therefore $F^t(v)_{l^U_t(c^1)+y}$ is non null as well.
  This implies that $l_t^{\overline{0}}(v)$ is not lower bounded.
  Symmetrically, $r_t^{\overline{0}}(w)$ is not upper bounded when $w$ is any configuration equal to zero on positive coordinates. 
  Lemma~\ref{th:exp} concludes that $F$ is positively expansive.
\end{proof}

To establish a separation between positive expansivity and pre-expansivity, we
will focus on linear CA obtained by what is often called ``second
order method'' in the literature \cite{Margolus1984}.  The idea is to
turn any CA into a reversible one by memorizing one step of history
and combining, in a reversible way, the memorized past step into the
produced future step. The interest of this construction for our
purpose is that positive expansivity is excluded from the beginning because no
non-trivial CA can be positively expansive and reversible at the same time
\cite{posexponetoone}.

\newcommand\SO[2]{\mathcal{SO}\bigl(#1,#2\bigr)}
Let $Q=\{0,\ldots,n-1\}$ be equipped with some group law $\gplus$ and
consider some CA $F$ over state set $Q$. %necessaire?
The second-order CA associated to
$F$ and $\gplus$, denoted $\SO{F}{\gplus}$ is the CA over state set
${Q\times Q}$, which is conjugated through the natural bijection ${Q^\Z\times
  Q^\Z\rightarrow (Q\times Q)^\Z}$ to the map:
\[(c,d) \mapsto \bigl(d,F(d)\gpconf c\bigr)\]

The following proposition shows that second order construction is
useful to separate positive expansivity from $1$-expansivity.
Some of the results in the next proposition can be deduced from more general results in~\cite{phdLukkarila,phdDiLena,kitchens}, but we develop a new specific proof here.

\begin{prop}
  \label{prop:secondorder}
  Let $\gplus$ be a group law over $Q$ with neutral element $0$ and $F$ be a CA over $Q$ which
  is linear for $\gplus$. It
  holds:
  \begin{enumerate}
  \item $\SO{F}{\gplus}$ is bijective and linear for the law $\gplus\times\gplus$;
  \item if $F$ is LR-permutive then $\SO{F}{\gplus}$ is $\Z$-expansive
    and $1$-expansive;
  \item if $F$ is LR-permutive then for any $m>0$ the subshift of
    traces ${T_m\bigl((Q\times Q)^\Z\bigr)}$ is a vertex SFT.
  \end{enumerate}
\end{prop}
\begin{proof}
  \begin{enumerate}
  \item It is sufficient to check that the CA over the state
  set ${Q\times Q}$ is conjugated to the following map:
  \[(c,d)\mapsto \bigl(\overline{\iota}(F(c))\gpconf d,c\bigr)\]
  the inverse of $\SO{F}{\gplus}$, where $\iota$ denotes the
  inverse function for the group law $\gplus$. Moreover
  $\SO{F}{\gplus}$ is linear for $\gplus\times\gplus$ because it is
  component-wise linear for $\gplus$.

  \item Let us suppose that $F$ is LR-permutive with neighborhood $\{-l,...,r\}$ and denote
  $\Psi=\SO{F}{\gplus}$. %\Psi->H ?
  To prove $\Z$-expansivity of $\Psi$ it is sufficient to notice that $\Psi$ propagates to left and right when the second
  $Q$-component is non-null and $\Psi^{-1}$ propagates to left and right when the first $Q$-component is non-null.
  In fact, let us consider a configuration $c\in(Q\times Q)^\Z$
  equal to $(0,0)$ on negative coordinates but such that $c(0)\neq
  (0,0)$. 
  If the second $Q$-component of $c(0)$ is non-null then, and since $F(0,...,0)=0$, the
  leftmost non-null cell of $\Psi(c)$ is at position $-r$ and it is its second $Q$-component which is non-null, i. e., if the leftmost difference from (0,0) is in the second $Q$-component, this will be always like this and the difference will propagate to the left.
  The same holds symmetrically for propagation to the right.
  In the same way, and given the form of $\Psi^{-1}$, differences in the first $Q$-component will propagate to the left and right through $\Psi^{-1}$, thus by lemma~\ref{th:exp}, $\Psi$ is $\Z$-expansive.\\
  Now let us take $c^{(a,b)}$ (recall it is the configuration equal to $(a,b)$ at 0 and $(0,0)$ everywhere else).
  By the previous arguments, if $b\neq 0$, $(l_t^{\overline{(0,0)}}(c^{(a,b)}))_{t\in\N}$ is not lower bounded and $(r_t^{\overline{(0,0)}}(c^{(a,b)}))_{t\in\N}$ is not upper bounded.
  But if $b=0$ and $a\neq 0$, then $(\Psi(c))(0)=(0,a)$ and null everywhere else, then again $(l_t^{\overline{(0,0)}}(c^{(a,b)}))_{t\in\N}$ is not lower bounded and $(r_t^{\overline{(0,0)}}(c^{(a,b)}))_{t\in\N}$ is not upper bounded.
  Therefore, $\Psi$ (and $\Psi^{-1}$) is $1$-expansive.
  
  \item The proof of the third item will be performed in two steps. To
  simplify notations, for any pair of words ${u,v\in Q^\ast}$ of the
  same length, we will denote by $u\choose v$ the word over alphabet
  ${Q\times Q}$ whose projection on the first (resp. second) component
  is $u$ (resp. $v$).
  \begin{description}
  \item{{\bf Assertion 1L:}} \emph{For every word $u\in Q^r$ there exists a configuration $c$ such that $\Psi^t(c)|_{[0,r-1]}={ u\choose 0^r}$ and $\Psi^k(c)_i=(0,0)$ for every $0\le k\le t$ and $i<(t-k)r$.}
  
  \emph{Proof of Assertion 1L.} By induction on $t$. If $t=0$ it is obvious, we just take $c$ equal to ${ u\choose 0^r}$ at $[0,r-1]$ and $(0,0)$ everywhere else.
  Now, since $F$ is LR-permutive, given a word $w\in Q^{l+r}$, let us define the permutation $\tau_w(a)=f(wa)$ for every $a\in Q$.
  Given a word $u\in Q^r$, we inductively define another word $v\in Q^r$ as follows: $v_0=\tau^{-1}_{0^{l+r}}(u_0)$, $v_{i+1}=\tau^{-1}_{0^{l+r-i}v_{[0,i]}}(u_{i+1})$.
  In this way, $f(0^{l+r}v)=u$.
  By induction hypothesis, there exists a configuration $c$ such that $\Psi^t(c)|_{[0,r-1]}={ v\choose 0^r}$ and $\Psi^k(c)_i=(0,0)$ for every $0\le k\le t$ and $i<(t-k)r$.
  We take $d=\sigma_{-r}(c)$, then $\Psi^t(d)|_{[r,2r-1]}={ v\choose 0^r}$, and $(0,0)$ to the left of $r$; and $\Psi^{t+1}(d)|_{[0,r-1]}={ u\choose 0^r}$, and $(0,0)$ to the left of $0$.
  Moreover, $\Psi^k(d)_i=(0,0)$ for every $0\le k\le t+1$ and $i<(t+1-k)r$.
  \item{{\bf Assertion 1R:}} \emph{For every word $u\in Q^l$ there exists a configuration $c$ such that $\Psi^t(c)|_{[-l+1,0]}={ u\choose 0^l}$ and $\Psi^k(c)_i=(0,0)$ for every $0\le k\le t$ and $i>(k-t)l$.}

  \emph{The proof of Assertion 1R is analogous to the proof of Assertion 1L.}
  
  \item{{\bf Assertion 2:}} \emph{A sequence $\bigl(w_t\bigr)_{t=0}^n$, with $w_t\in (Q\times Q)^{2m+1}$, is a finite subsequence of a trace in $T_m\bigl((Q\times Q)^\Z\bigr)$ if and only if for any $t$, there are extensions $w_R\in (Q\times Q)^r$ and $w_L\in (Q\times Q)^l$ verifying}
  \[\psi(w_L\cdot w_t\cdot w_R)= w_{t+1}\]
  where $\psi$ denotes the action of $\Psi$ over finite words.
  
  \emph{Proof of Assertion 2.} In one direction, it is clear, so let $\bigl(w_t\bigr)_{t=0}^n$ be a sequence such that for any $t$, there are extensions $w_R\in (Q\times Q)^r$ and $w_L\in (Q\times Q)^l$ verifying $\psi(w_L\cdot w_t\cdot w_R)= w_{t+1}$, and let us prove that it is a subsequence of a trace of $\Psi$.
We perform the proof by induction on $n$. If $n=0$ there is nothing to prove, of course any sequence of length 1 can be part of a trace.
Now let $c$ be a configuration such that $\Psi^k(c)|_{[-m,m]}=w_k$, for every $k\in\{0,...,n-1\}$.
By locality, only the values of $c$ between $-m-nl$ and $m+nr$ are relevant to this hypothesis, and we take $c(i)=(0,0)$ outside these limits.
%Let $w_R={v_R\choose 0^r}\in (Q\times Q)^r$ and $w_L={v_L\choose 0^r}\in (Q\times Q)^l$ be such that $g(w_L\cdot w_{n-1}\cdot w_R)= w_n$, and let $u_R=v_R-\pi_2(\Psi^{n-1}(c)_{[m+1,m+r]})$ and $u_L=v_L-\pi_2(\Psi^{n-1}(c)_{[-m-l,-m-1]})$.
Let $w_R$ and $w_L$ be such that $\psi(w_L\cdot w_{n-1}\cdot w_R)= w_n$.
Let us remark that the first $Q$-component of $w_L$ and $w_R$ can be chosen arbitrarily, given the form of $\Psi$.
We will suppose that $\pi_1(w_L)=\pi_1(\Psi^{n-1}(c)_{[-m-l,-m-1]})$ and $\pi_1(w_R)=\pi_1(\Psi^{n-1}(c)_{[m+1,m+r]})$, thus we can take ${u_L\choose 0^l}=w_L-\Psi^{n-1}(c)|_{[-m-l,-m-1]}$ and ${u_R\choose 0^r}=w_R-\Psi^{n-1}(c)|_{[m+1,m+r]}$.
We take from Assertion 1L a configuration $c^R$ that produces the word ${u_R\choose 0^r}$ at time $n-1$ at position ${[m+1,m+r]}$ and $(0,0)$ to the left of the light cone that starts at $m+nr$ with slope $-1/r$.
From Assertion 1R we take a configuration $c^L$ that produces the word ${u_L\choose 0^r}$ at time $n-1$ at position ${[-m-l,-m-1]}$ and $(0,0)$ to the right of the light cone that starts at $-m-nl$ with slope $1/l$.
By linearity, $\Psi^{n-1}(c^L\gpconf c^R\gpconf c)|_{[-m-l,m+r]}=w_Lw_{n-1}w_R$, and then $\Psi^n(c_L\gpconf c_R\gpconf c)|_{[-m,m]}=w_n$, moreover $\Psi^k(c_L\gpconf c_R\gpconf c)|_{[-m,m]}=w_k$, for every $0\le k<n$, which completes the proof.
  \end{description}
  \end{enumerate}
  \end{proof}
%enumerate

We will now give an example of pre-expansive CA which is not
positively expansive. 

\begin{exam}[$\Psi$]
  Let ${Q=\{0,1,2\}}$, $+$ be the addition modulo $3$, and $F_3$ be the
  CA defined over $Q^\Z$ by ${F_3 = \sigma \overline{+} \sigma_{-1}}$. We define
  $\Psi$ as the second order construction applied to $F_3$. 
  \[\Psi=\SO{F_3}{+}\]
  Figure~\ref{fig:stdiag} shows a space-time diagram of $\Psi$.
\end{exam}

\begin{figure}[hbp!]
  \centering
  \includegraphics[width=.9\textwidth]{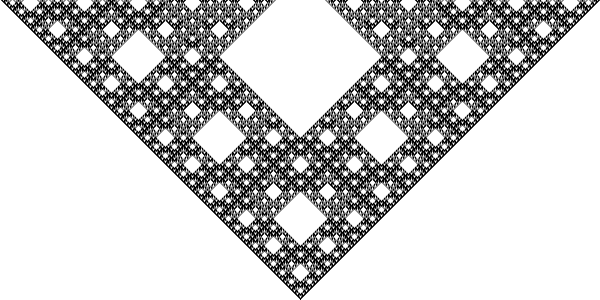}
  \caption{Space-time diagram of $\Psi$ starting from a configuration with a single non-zero cell.}
  \label{fig:stdiag}
\end{figure}

\newcommand\stateZ[2]{}
\newcommand\stateO[2]{\draw[fill=white!30!black] (#1,#2) rectangle  +(1,1); \draw (#1,#2)+(.5,.5) node {$1,0$};}
\newcommand\stateZO[2]{\draw[fill=white!40!black] (#1,#2) rectangle +(1,1);\draw (#1,#2)+(.5,.5) node {$2,0$};}
\newcommand\stateOO[2]{\draw[fill=white!50!black] (#1,#2) rectangle +(1,1);\draw (#1,#2)+(.5,.5) node {$0,1$};}
\newcommand\stateZZO[2]{\draw[fill=white!60!black] (#1,#2) rectangle +(1,1);\draw (#1,#2)+(.5,.5) node {$1,1$};}
\newcommand\stateOZO[2]{\draw[fill=white!70!black] (#1,#2) rectangle +(1,1);\draw (#1,#2)+(.5,.5) node {$2,1$};}
\newcommand\stateZOO[2]{\draw[fill=white!80!black] (#1,#2) rectangle +(1,1);\draw (#1,#2)+(.5,.5) node {$0,2$};}
\newcommand\stateOOO[2]{\draw[fill=white!90!black] (#1,#2) rectangle +(1,1);\draw (#1,#2)+(.5,.5) node {$1,2$};}
\newcommand\stateZZZO[2]{\draw (#1,#2) rectangle +(1,1);\draw (#1,#2)+(.5,.5) node {$2,2$};}

\begin{figure}
  \begin{center}
    \tiny \hbox{
      \begin{tikzpicture}[scale=.4]
        \draw[thick,->] (2,-.5) -- (2,2) node[midway,sloped,above] {time} ;
        \draw[dotted] (1.5,0) -- (13,0);

\stateZ{0}{0}\stateZ{0}{1}\stateZ{0}{2}\stateZ{0}{3}\stateZ{0}{4}\stateZ{0}{5}
\stateZ{1}{0}\stateZ{1}{1}\stateZ{1}{2}\stateZ{1}{3}\stateZ{1}{4}\stateZ{1}{5}
\stateZ{2}{0}\stateZ{2}{1}\stateZ{2}{2}\stateZ{2}{3}\stateZ{2}{4}\stateZ{2}{5}\stateZOO{2}{6}
\stateZ{3}{0}\stateZ{3}{1}\stateZ{3}{2}\stateZ{3}{3}\stateZ{3}{4}\stateZOO{3}{5}\stateZO{3}{6}
\stateZ{4}{0}\stateZ{4}{1}\stateZ{4}{2}\stateZ{4}{3}\stateZOO{4}{4}\stateZO{4}{5}\stateZ{4}{6}
\stateZ{5}{0}\stateZ{5}{1}\stateZ{5}{2}\stateZOO{5}{3}\stateZO{5}{4}\stateZOO{5}{5}\stateZO{5}{6}
\stateZ{6}{0}\stateZ{6}{1}\stateZOO{6}{2}\stateZO{6}{3}\stateOO{6}{4}\stateO{6}{5}\stateZOO{6}{6}
\stateZO{7}{0}\stateZOO{7}{1}\stateZO{7}{2}\stateZ{7}{3}\stateZ{7}{4}\stateZOO{7}{5}\stateZO{7}{6}
\stateZ{8}{0}\stateZ{8}{1}\stateZOO{8}{2}\stateZO{8}{3}\stateOO{8}{4}\stateO{8}{5}\stateZOO{8}{6}
\stateZ{9}{0}\stateZ{9}{1}\stateZ{9}{2}\stateZOO{9}{3}\stateZO{9}{4}\stateZOO{9}{5}\stateZO{9}{6}
\stateZ{10}{0}\stateZ{10}{1}\stateZ{10}{2}\stateZ{10}{3}\stateZOO{10}{4}\stateZO{10}{5}\stateZ{10}{6}
\stateZ{11}{0}\stateZ{11}{1}\stateZ{11}{2}\stateZ{11}{3}\stateZ{11}{4}\stateZOO{11}{5}\stateZO{11}{6}
\stateZ{12}{0}\stateZ{12}{1}\stateZ{12}{2}\stateZ{12}{3}\stateZ{12}{4}\stateZ{12}{5}\stateZOO{12}{6}
\stateZ{13}{0}\stateZ{13}{1}\stateZ{13}{2}\stateZ{13}{3}\stateZ{13}{4}\stateZ{13}{5}
\stateZ{14}{0}\stateZ{14}{1}\stateZ{14}{2}\stateZ{14}{3}\stateZ{14}{4}\stateZ{14}{5}

      \end{tikzpicture} \hskip 1cm
      \begin{tikzpicture}[scale=.4]
        \draw[thick,->] (1,-.5) -- (1,2) node[midway,sloped,above] {time} ;
        \draw[dotted] (.5,0) -- (13,0);
\stateZ{0}{0}\stateZ{0}{1}\stateZ{0}{2}\stateZ{0}{3}\stateZ{0}{4}\stateZ{0}{5}
\stateZ{1}{0}\stateZ{1}{1}\stateZ{1}{2}\stateZ{1}{3}\stateZ{1}{4}\stateZ{1}{5}\stateZOO{1}{6}
\stateZ{2}{0}\stateZ{2}{1}\stateZ{2}{2}\stateZ{2}{3}\stateZ{2}{4}\stateZOO{2}{5}\stateZO{2}{6}
\stateZ{3}{0}\stateZ{3}{1}\stateZ{3}{2}\stateZ{3}{3}\stateZOO{3}{4}\stateZO{3}{5}\stateOO{3}{6}
\stateZ{4}{0}\stateZ{4}{1}\stateZ{4}{2}\stateZOO{4}{3}\stateZO{4}{4}\stateZ{4}{5}\stateZ{4}{6}
\stateZ{5}{0}\stateZ{5}{1}\stateZOO{5}{2}\stateZO{5}{3}\stateZOO{5}{4}\stateZO{5}{5}\stateOO{5}{6}
\stateZ{6}{0}\stateZOO{6}{1}\stateZO{6}{2}\stateOO{6}{3}\stateO{6}{4}\stateZOO{6}{5}\stateZO{6}{6}
\stateZOO{7}{0}\stateZO{7}{1}\stateZ{7}{2}\stateZ{7}{3}\stateZOO{7}{4}\stateZO{7}{5}\stateZ{7}{6}
\stateZ{8}{0}\stateZOO{8}{1}\stateZO{8}{2}\stateOO{8}{3}\stateO{8}{4}\stateZOO{8}{5}\stateZO{8}{6}
\stateZ{9}{0}\stateZ{9}{1}\stateZOO{9}{2}\stateZO{9}{3}\stateZOO{9}{4}\stateZO{9}{5}\stateOO{9}{6}
\stateZ{10}{0}\stateZ{10}{1}\stateZ{10}{2}\stateZOO{10}{3}\stateZO{10}{4}\stateZ{10}{5}\stateZ{10}{6}
\stateZ{11}{0}\stateZ{11}{1}\stateZ{11}{2}\stateZ{11}{3}\stateZOO{11}{4}\stateZO{11}{5}\stateOO{11}{6}
\stateZ{12}{0}\stateZ{12}{1}\stateZ{12}{2}\stateZ{12}{3}\stateZ{12}{4}\stateZOO{12}{5}\stateZO{12}{6}
\stateZ{13}{0}\stateZ{13}{1}\stateZ{13}{2}\stateZ{13}{3}\stateZ{13}{4}\stateZ{13}{5}\stateZOO{13}{6}
\stateZ{14}{0}\stateZ{14}{1}\stateZ{14}{2}\stateZ{14}{3}\stateZ{14}{4}\stateZ{14}{5}
      \end{tikzpicture}}   
    \vskip 1cm
           \includegraphics[width=.8\textwidth]{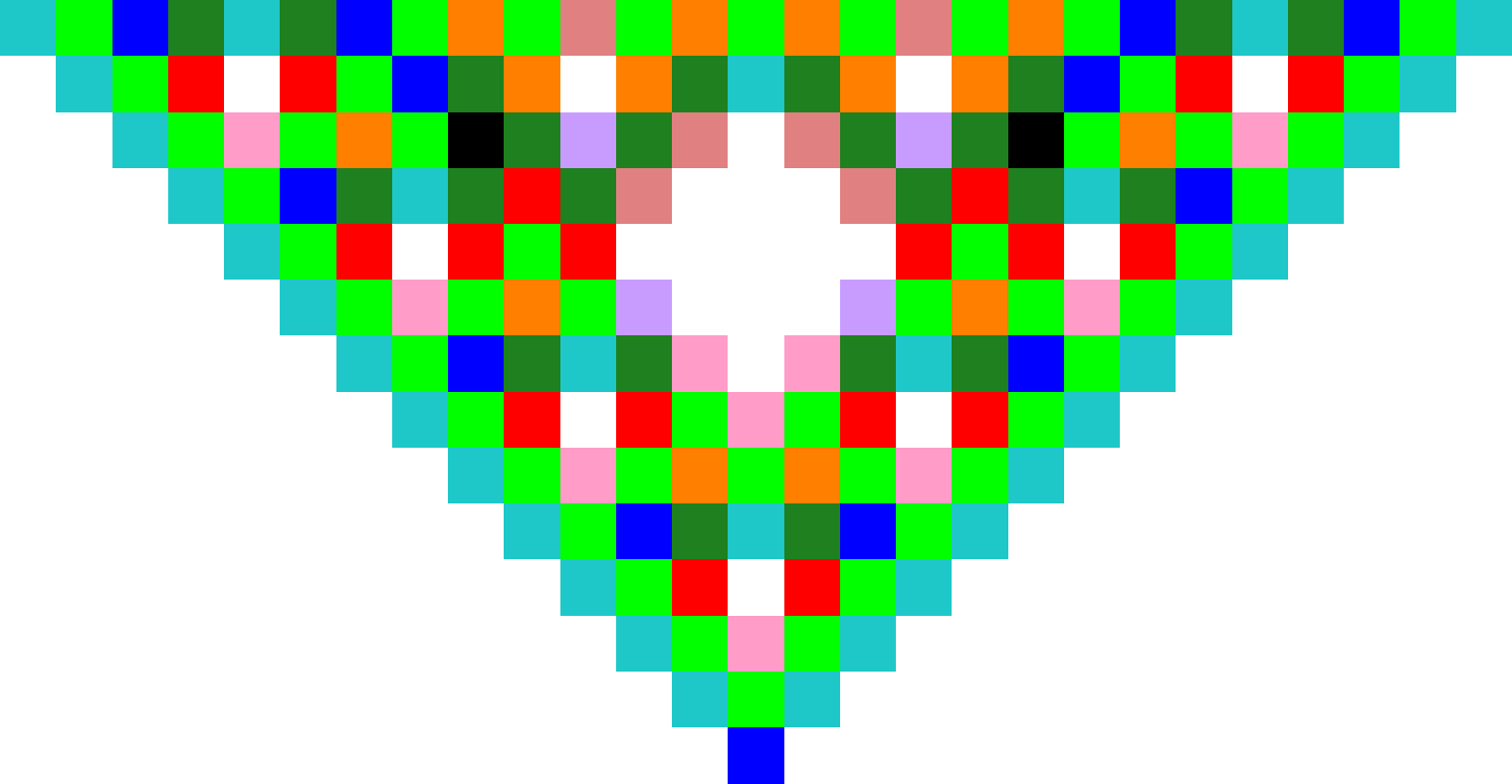}
  \end{center}
  \caption{In the first row some space-time diagrams of $\Psi$ are shown (state $(0,0)$ is
    represented by empty space).
    The third diagram shows the morphisms representing the dependence of $\Psi^t(c^{(a,b)})_i$ on $(a,b)$, for each space-time coordinate $(i,t)$: in particular, white means no dependence and red means bijective dependence.}
\label{fig:stdg}
\end{figure}
% $G$ is defined over state set ${Q=\{0,1,2\}\times\{0,1,2\}}$ and has
% radius $1$. To simplify notations we define $G$ as a map acting
% on ${\{0,1,2\}^\Z\times\{0,1,2\}^\Z}$:
% \[G(c,d) = \bigl(d,\sigma_{-1}(d)\,\overline{+}\,c\,\overline{+}\,\sigma_1(d)\bigr)\]
% where $\overline{+}$ is the sum modulo $3$ extended to configurations. In the
% sequel we refer to $G$ as the map over $Q^\Z$.

% \begin{lemm}
%   $G$ is linear for the component-wise addition modulo $3$. $G$ is
%   also reversible and hence not expansive.
% \end{lemm}
% \begin{proof}
%   \TODO{...}
% \end{proof}

To establish the pre-expansivity of $\Psi$ we will study its dependency
structure, \textit{i.e.} how the value of the cell at position $z$ and
time $t$ depends on value of cells at other positions and earlier times.
To express these space-time relations we denote by $\Psi_z^t$ the map
$\sigma_z\circ \Psi^t$ and by $\oplus$ the component-wise addition modulo
$3$ over ${\{0,1,2\}\times\{0,1,2\}}$ naturally extended to configurations of $(Q\times Q)^\Z$ and then
naturally extended to functions on such configurations.

%%%%%%%%%%%%%%%%%%%%%%%%%%%%%%%%%%%%%%%
%%%%%%%%%%%%%ANAHI%%%%%%%%%%%%%%%%%%%%%%%
Following section~\ref{sec:abelian}, and lemma~\ref{lem:multiscaleidentities}, the matrix that represents $\Psi$ is:

$$\left(\begin{array}{rr} 0 & 1\\ 1 & x^{-1}+x\end{array}\right).$$

Its characteristic polynomial is $p(\lambda)=\lambda^2-\lambda(x^{-1}+x)-1$, which gives the relation $\Psi^2=\Psi\circ\sigma_{-1} \oplus id \oplus \Psi\circ\sigma$. 
Lemma~\ref{lem:multiscaleidentities} proves that this imply the existence of a multi-scale additive identity, next lemma gives the precise shape of this identity in the special case of $\Psi$.
%%%%%%%%%%%%%ANAHI%%%%%%%%%%%%%%%%%%%%%%%

\begin{lemm}
  \label{lem:relation}
  Let $c$ be a configuration in $(Q\times Q)^\Z$.%such that $c(z)=(0,0)$ for any $z\neq 0$.
  Then, for any $k\geq 0$, any $t\geq 0$ and any $z\in\Z$ we have:
  \[\Psi_z^{2\cdot 3^k+t} (c) = \Psi_z^t \oplus
  \Psi_{z-3^k}^{3^k+t}\oplus \Psi_{z+3^k}^{3^k+t} (c).\]
\end{lemm}

\begin{proof}
  First it is straightforward to check that
  \[\Psi^2_0 (c) =  \Psi_{-1}^1 \oplus Id \oplus\Psi_1^1 (c).\]
  Then, by property of the Frobenius endomorphism, we have ${(a+b)^{3^k} = a^{3^k} + b^{3^k}}$
  when doing the arithmetics modulo 3. This identity naturally extends to
  $\oplus$ and therefore we have:
  \[\Psi^{2\cdot 3^k}_0 (c) =  \Psi_{-3^k}^{3^k} \oplus
  Id \oplus\Psi_{3^k}^{3^k} (c).\]
  Finally, by linearity of both $\sigma$ and $\Psi$ with respect to
  $\oplus$ we can compose both sides of the above equality by $\Psi_z^t$
  and the lemma follows.
\end{proof}

%%%%%%%%%%%%%%%%%%%%%%%%%%%%%%%%%%%%%%%
%%%%%%%%%%%%%ANAHI%%%%%%%%%%%%%%%%%%%%%%%
Using the above lemma, and arithmetics modulo 3, we can show that $\Psi$ has a simple dependency
structure at some space-time locations.
%%%%%%%%%%%%%ANAHI%%%%%%%%%%%%%%%%%%%%%%%
%%%%%%%%%%%%%%%%%%%%%%%%%%%%%%%%%%%%%%%

\begin{lemm}
  \label{lem:gponctual}
  Let us consider a configuration $c=c^{(a,b)}$ for some pair $(a,b)\in Q\times Q$. 
  For any integers $k\geq 0$ and ${M\geq 1}$, let $d_{M,k}=\Psi^{M\cdot3^{k+1}}(c)$. Then we have:
  \begin{itemize}
  \item $d_{M,k}(-M\cdot3^{k+1}+2\cdot3^k) = d_{M,k}(M\cdot3^{k+1}-2\cdot3^k) = \phi(a,b)$
  \item $d_{M,k}(i) = (0,0)$ for ${(M-1)\cdot 3^{k+1}\leq |i| < (M-1)\cdot 3^{k+1} + 3^k}$
  \end{itemize}
  where $\phi$ is an automorphism of $Q\times Q$ which does not depend on $k$ nor on $M$. %%%ANAHI%%%
\end{lemm}
\begin{proof}
  Let's first show the two items for ${M=1}$ and let ${d_k=d_{1,k}}$.
  Denote by $c_z^t$ the state ${\bigl(\Psi^t(c)\bigr)(z)}$. First, by a
  simple recurrence we can show that ${c_{-n}^n=c_n^n=\pi(a,b)}$
  where $\pi$ is the projection ${\pi(a,b)=(0,b)}$.
  Then, applying Lemma~\ref{lem:relation} on $\Psi_{-3^k}^{2\cdot 3^k +
    3^k}$, we obtain:
  \[d_k(-3^k) = c_{-2\cdot 3^k}^{2\cdot 3^k} + c_{-3^k}^{3^k} +
  c_0^{2\cdot 3^k}.\]
  Applying Lemma~\ref{lem:relation} to $\Psi_0^{2\cdot 3^k+0}$ we then have:
  \begin{eqnarray*}
d_k(-3^k)& =& c_{-2\cdot 3^k}^{2\cdot 3^k} + 2\cdot c_{-3^k}^{3^k} +
  c_{3^k}^{3^k} + c_0^0\\
&=& (0,b)+2(0,b)+(0,b)+(a,b)\\
&=& (a,2b)
  \end{eqnarray*} 
  where ${\phi(a,b) = (a,2b)}$ is an automorphism.
  The same equality holds for $d_k(3^k)$ by
  symmetry and the first item of the lemma is shown.

  For the second item, first note that ${c_z^t=(0,0)}$ whenever
  ${z<-t}$ or ${z>t}$ because $\Psi$ has radius $1$. Consider any $i$
  with ${-3^k< i < 3^k}$. Applying Lemma~\ref{lem:relation} on $\Psi_i^{2\cdot 3^k +
    3^k}$, we obtain:
  \[d_k(i) = c_{i-3^k}^{2\cdot 3^k} + c_{i}^{3^k} +
  c_{i+3^k}^{2\cdot 3^k}.\]
  Applying Lemma~\ref{lem:relation} to $\Psi_{i-3^k}^{2\cdot 3^k+0}$ and
  $\Psi_{i+3^k}^{2\cdot 3^k+0}$ we further get:
  \[d_k(i) = c_{i-2\cdot 3^k}^{3^k} + c_{i-3^k}^0 + 3\cdot c_{i}^{3^k}
  + c_{i+3^k}^0 +  c_{i+2\cdot 3^k}^{3^k}.\]
  From what we said before and doing the arithmetics modulo $3$ we
  deduce ${d_k(i)=(0,0)}$ and the lemma follows.

  Now, proceeding by induction on $M$, suppose we have:
  \begin{itemize}
  \item $d_{M,k}(-M\cdot3^{k+1}+2\cdot3^k) = d_{M,k}(M\cdot3^{k+1}-2\cdot3^k) = \phi(a,b)$
  \item $d_{M,k}(i) = (0,0)$ for ${(M-1)\cdot 3^{k+1}\leq |i| < M\cdot 3^{k+1} -2\cdot3^k}$
  \end{itemize}
  Writing ${(M+1)\cdot3^{k+1} = 2\cdot 3^{k+1} + t}$ with ${t=(M-1)\cdot3^{k+1}}$ and applying\break Lemma~\ref{lem:relation} to $\Psi_z^{2\cdot 3^{k+1} + t}$ with ${z=(M+1)\cdot3^{k+1}-2\cdot3^k}$ we get:
  \begin{eqnarray*}
  d_{M+1,k}((M+1)3^{k+1}-2\cdot3^k) &=& d_{M,k}(M\cdot3^{k+1}-2\cdot3^k) + c_{(M+1)3^{k+1}-2\cdot3^k}^t \\& &+ d_{M,k}((M+2)3^{k+1}-2\cdot3^k)\\
    &= &d_{M,k}(M\cdot3^{k+1}-2\cdot3^k)\\
    &=& \phi(a,b) =(a,2b)
  \end{eqnarray*}
  Applying again Lemma~\ref{lem:relation} but with ${z=(M+1)\cdot3^{k+1}-2\cdot3^k - j}$ where\break ${-3^k\leq j < 0}$ we deduce $d_{M+1,k}(i) = (0,0)$ for ${M\cdot3^{k+1}\leq i < (M+1)\cdot 3^{k+1} -2\cdot3^k}$. By symmetry ${z\mapsto -z}$ we obtain the corresponding equalities and we finally have:
  \begin{itemize}
  \item $d_{M+1,k}(-(M+1)\cdot3^{k+1}+2\cdot3^k) = d_{M+1,k}((M+1)\cdot3^{k+1}-2\cdot3^k) = \phi(a,b)$
  \item $d_{M+1,k}(i) = (0,0)$ for ${M\cdot 3^{k+1}\leq |i| < (M+1)\cdot 3^{k+1} -2\cdot3^k}$
  \end{itemize}
  which completes the induction step. The lemma follows.
\end{proof}

%%%%%%%%%%%%%%%%%%%%%%%%%%%%%%%%%%%%%%%
%%%%%%%%%%%%%ANAHI%%%%%%%%%%%%%%%%%%%%%%%
The last lemma assures that finite configurations will always produce at least one non 0 cell at arbitrary large positions in the line.
%%%%%%%%%%%%%ANAHI%%%%%%%%%%%%%%%%%%%%%%%
%%%%%%%%%%%%%%%%%%%%%%%%%%%%%%%%%%%%%%%

\begin{prop}\label{prop:prexpexample}
  $\Psi$ is pre-expansive in direction $\alpha$ if and only if ${\alpha\in]-1,1[}$.
  \label{prop:gprexp}
\end{prop}
\begin{proof}
  Let $c$ be any configuration with $c\asy\overline{(0,0)}$.
  Denote
  $l=l_0^{\overline{(0,0)}}(c)$ and $r=r_0^{\overline{(0,0)}}(c)$ and consider any $k$ such that
  ${3^k>\max(|l|,|r|)}$ and ${M}$ such that ${|\alpha|<1-\frac{2}{3M}}$. By a finite number of applications of
  Lemma~\ref{lem:gponctual}, and by linearity and translation
  invariance of $\Psi$, we have:
  \begin{align*}
    \bigl(\Psi^{M\cdot3^{k+1}}(c)\bigr)(r-M\cdot3^{k+1}+2\cdot3^k) &= \phi\bigl(c(r)\bigr)\\
    \bigl(\Psi^{M\cdot3^{k+1}}(c)\bigr)(l+M\cdot3^{k+1}-2\cdot3^k) &= \phi\bigl(c(l)\bigr).\\
  \end{align*}
  Since $\phi$ is a permutation of $Q$ sending $(0,0)$ to itself and
  since $c(r)$ and $c(l)$ are both different from $(0,0)$ we deduce
  that ${l_{M\cdot3^{k+1}}^{\overline{(0,0)}}(c)<r-M\cdot3^{k+1}+2\cdot3^k}$ and ${r_{3^{k+1}}^{\overline{(0,0)}}(c)>l+M\cdot3^{k+1}-2\cdot3^k}$ for any
  $k$ large enough. This shows that ${\bigl(l_t(c)-\lceil\alpha t\rceil\bigr)_{t\in\N}}$ is not lower-bounded and
  ${\bigl(r_t(c)-\lceil\alpha t\rceil\bigr)_{t\in\N}}$ is not upper-bounded and concludes the proof by
  a directional version of Lemma~\ref{lem:leftright}.
\end{proof}

We now give an example of CA that is $1$-expansive but not pre-expansive.

\begin{exam}[$\Upsilon$]
  Let ${Q=\{0,1\}}$, $+$ be the addition modulo $2$, and $F_2$ be the
  CA defined over $Q^\Z$ by ${F_2 = \sigma \overline{+} \sigma_{-1}}$. We define
  $\Upsilon$ as the second order construction applied to $F_2$:
  \[\Upsilon=\SO{F_2}{+}.\]
\end{exam}

%%%%%%%%%%%%%%%%%%%%%%%%%%%%%%%%%%%%%%%
%%%%%%%%%%%%%ANAHI%%%%%%%%%%%%%%%%%%%%%%%
$\Upsilon$ has the same matrix as $\Psi$, the same characteristic polynomial, and a similar multi-scale identity, but arithmetics modulo 2 gives a different shape to the time space diagram as we can see in figure~\ref{fig:stdh}.
%%%%%%%%%%%%%ANAHI%%%%%%%%%%%%%%%%%%%%%%%
%%%%%%%%%%%%%%%%%%%%%%%%%%%%%%%%%%%%%%%

\renewcommand\stateZ[2]{}
\renewcommand\stateO[2]{\draw (#1,#2) rectangle +(1,1); \draw (#1,#2)+(.5,.5) node {$1,0$};}
\renewcommand\stateZO[2]{\draw[fill=white!80!black] (#1,#2) rectangle +(1,1); \draw (#1,#2)+(.5,.5) node {$0,1$};}
\renewcommand\stateOO[2]{\draw[fill=white!90!black] (#1,#2) rectangle +(1,1); \draw (#1,#2)+(.5,.5) node {$1,1$};}

\newcommand\stateB[2]{\draw (#1,#2) rectangle +(1,1); \draw (#1,#2)+(.5,.5) node {$id$};}
\newcommand\stateAB[2]{\draw[fill=white!80!black] (#1,#2) rectangle +(1,1); \draw (#1,#2)+(.5,.5) node {$\pi$};}
\newcommand\stateBB[2]{\draw[fill=white!90!black] (#1,#2) rectangle +(1,1); \draw (#1,#2)+(.5,.5) node {$e$};}

\begin{figure}
  \begin{center}
    \tiny \hbox{
      \begin{tikzpicture}[scale=.4]
        \draw[thick,->] (2,-.5) -- (2,2) node[midway,sloped,above] {time} ;
        \draw[dotted] (1.5,0) -- (13,0);
        \stateZ{0}{0}\stateZ{0}{1}\stateZ{0}{2}\stateZ{0}{3}\stateZ{0}{4}\stateZ{0}{5}
        \stateZ{1}{0}\stateZ{1}{1}\stateZ{1}{2}\stateZ{1}{3}\stateZ{1}{4}\stateZ{1}{5}
        \stateZ{2}{0}\stateZ{2}{1}\stateZ{2}{2}\stateZ{2}{3}\stateZ{2}{4}\stateZO{2}{5}
        \stateZ{3}{0}\stateZ{3}{1}\stateZ{3}{2}\stateZ{3}{3}\stateZO{3}{4}\stateOO{3}{5}
        \stateZ{4}{0}\stateZ{4}{1}\stateZ{4}{2}\stateZO{4}{3}\stateOO{4}{4}\stateOO{4}{5}
        \stateZ{5}{0}\stateZ{5}{1}\stateZO{5}{2}\stateOO{5}{3}\stateOO{5}{4}\stateOO{5}{5}
        \stateZ{6}{0}\stateZO{6}{1}\stateOO{6}{2}\stateOO{6}{3}\stateOO{6}{4}\stateOO{6}{5}
        \stateOO{7}{0}\stateOO{7}{1}\stateOO{7}{2}\stateOO{7}{3}\stateOO{7}{4}\stateOO{7}{5}
        \stateZ{8}{0}\stateZO{8}{1}\stateOO{8}{2}\stateOO{8}{3}\stateOO{8}{4}\stateOO{8}{5}
        \stateZ{9}{0}\stateZ{9}{1}\stateZO{9}{2}\stateOO{9}{3}\stateOO{9}{4}\stateOO{9}{5}
        \stateZ{10}{0}\stateZ{10}{1}\stateZ{10}{2}\stateZO{10}{3}\stateOO{10}{4}\stateOO{10}{5}
        \stateZ{11}{0}\stateZ{11}{1}\stateZ{11}{2}\stateZ{11}{3}\stateZO{11}{4}\stateOO{11}{5}
        \stateZ{12}{0}\stateZ{12}{1}\stateZ{12}{2}\stateZ{12}{3}\stateZ{12}{4}\stateZO{12}{5}
        \stateZ{13}{0}\stateZ{13}{1}\stateZ{13}{2}\stateZ{13}{3}\stateZ{13}{4}\stateZ{13}{5}
        \stateZ{14}{0}\stateZ{14}{1}\stateZ{14}{2}\stateZ{14}{3}\stateZ{14}{4}\stateZ{14}{5}
      \end{tikzpicture}\hskip 1cm
      \begin{tikzpicture}[scale=.4]
        \draw[thick,->] (2,-.5) -- (2,2) node[midway,sloped,above] {time} ;
        \draw[dotted] (1.5,0) -- (13,0);
        \stateZ{0}{0}\stateZ{0}{1}\stateZ{0}{2}\stateZ{0}{3}\stateZ{0}{4}\stateZ{0}{5}
        \stateZ{1}{0}\stateZ{1}{1}\stateZ{1}{2}\stateZ{1}{3}\stateZ{1}{4}\stateZ{1}{5}
        \stateZ{2}{0}\stateZ{2}{1}\stateZ{2}{2}\stateZ{2}{3}\stateZ{2}{4}\stateZ{2}{5}
        \stateZ{3}{0}\stateZ{3}{1}\stateZ{3}{2}\stateZ{3}{3}\stateZ{3}{4}\stateZO{3}{5}
        \stateZ{4}{0}\stateZ{4}{1}\stateZ{4}{2}\stateZ{4}{3}\stateZO{4}{4}\stateO{4}{5}
        \stateZ{5}{0}\stateZ{5}{1}\stateZ{5}{2}\stateZO{5}{3}\stateO{5}{4}\stateZO{5}{5}
        \stateZ{6}{0}\stateZ{6}{1}\stateZO{6}{2}\stateO{6}{3}\stateZO{6}{4}\stateO{6}{5}
        \stateO{7}{0}\stateZO{7}{1}\stateO{7}{2}\stateZO{7}{3}\stateO{7}{4}\stateZO{7}{5}
        \stateZ{8}{0}\stateZ{8}{1}\stateZO{8}{2}\stateO{8}{3}\stateZO{8}{4}\stateO{8}{5}
        \stateZ{9}{0}\stateZ{9}{1}\stateZ{9}{2}\stateZO{9}{3}\stateO{9}{4}\stateZO{9}{5}
        \stateZ{10}{0}\stateZ{10}{1}\stateZ{10}{2}\stateZ{10}{3}\stateZO{10}{4}\stateO{10}{5}
        \stateZ{11}{0}\stateZ{11}{1}\stateZ{11}{2}\stateZ{11}{3}\stateZ{11}{4}\stateZO{11}{5}
        \stateZ{12}{0}\stateZ{12}{1}\stateZ{12}{2}\stateZ{12}{3}\stateZ{12}{4}\stateZ{12}{5}
        \stateZ{13}{0}\stateZ{13}{1}\stateZ{13}{2}\stateZ{13}{3}\stateZ{13}{4}\stateZ{13}{5}
        \stateZ{14}{0}\stateZ{14}{1}\stateZ{14}{2}\stateZ{14}{3}\stateZ{14}{4}\stateZ{14}{5}
      \end{tikzpicture}}
    \vskip 1cm
    \hbox{\begin{tikzpicture}[scale=.4]
        \draw[thick,->] (2,-.5) -- (2,2) node[midway,sloped,above] {time} ;
        \draw[dotted] (1.5,0) -- (13,0);
    \stateZ{0}{0}\stateZ{0}{1}\stateZ{0}{2}\stateZ{0}{3}\stateZ{0}{4}\stateZ{0}{5}
      \stateZ{1}{0}\stateZ{1}{1}\stateZ{1}{2}\stateZ{1}{3}\stateZ{1}{4}\stateZ{1}{5}
      \stateZ{2}{0}\stateZ{2}{1}\stateZ{2}{2}\stateZ{2}{3}\stateZ{2}{4}\stateZO{2}{5}
      \stateZ{3}{0}\stateZ{3}{1}\stateZ{3}{2}\stateZ{3}{3}\stateZO{3}{4}\stateO{3}{5}
      \stateZ{4}{0}\stateZ{4}{1}\stateZ{4}{2}\stateZO{4}{3}\stateO{4}{4}\stateZO{4}{5}
      \stateZ{5}{0}\stateZ{5}{1}\stateZO{5}{2}\stateO{5}{3}\stateZO{5}{4}\stateO{5}{5}
      \stateZ{6}{0}\stateZO{6}{1}\stateO{6}{2}\stateZO{6}{3}\stateO{6}{4}\stateZO{6}{5}
      \stateZO{7}{0}\stateO{7}{1}\stateZO{7}{2}\stateO{7}{3}\stateZO{7}{4}\stateO{7}{5}
      \stateZ{8}{0}\stateZO{8}{1}\stateO{8}{2}\stateZO{8}{3}\stateO{8}{4}\stateZO{8}{5}
      \stateZ{9}{0}\stateZ{9}{1}\stateZO{9}{2}\stateO{9}{3}\stateZO{9}{4}\stateO{9}{5}
      \stateZ{10}{0}\stateZ{10}{1}\stateZ{10}{2}\stateZO{10}{3}\stateO{10}{4}\stateZO{10}{5}
      \stateZ{11}{0}\stateZ{11}{1}\stateZ{11}{2}\stateZ{11}{3}\stateZO{11}{4}\stateO{11}{5}
      \stateZ{12}{0}\stateZ{12}{1}\stateZ{12}{2}\stateZ{12}{3}\stateZ{12}{4}\stateZO{12}{5}
      \stateZ{13}{0}\stateZ{13}{1}\stateZ{13}{2}\stateZ{13}{3}\stateZ{13}{4}\stateZ{13}{5}
      \stateZ{14}{0}\stateZ{14}{1}\stateZ{14}{2}\stateZ{14}{3}\stateZ{14}{4}\stateZ{14}{5}
    \end{tikzpicture}\hskip 1cm
    \begin{tikzpicture}[scale=.4]
        \draw[thick,->] (2,-.5) -- (2,2) node[midway,sloped,above] {time} ;
        \draw[dotted] (1.5,0) -- (13,0);
    \stateZ{0}{0}\stateZ{0}{1}\stateZ{0}{2}\stateZ{0}{3}\stateZ{0}{4}\stateZ{0}{5}
      \stateZ{1}{0}\stateZ{1}{1}\stateZ{1}{2}\stateZ{1}{3}\stateZ{1}{4}\stateZ{1}{5}
      \stateZ{2}{0}\stateZ{2}{1}\stateZ{2}{2}\stateZ{2}{3}\stateZ{2}{4}\stateAB{2}{5}
      \stateZ{3}{0}\stateZ{3}{1}\stateZ{3}{2}\stateZ{3}{3}\stateAB{3}{4}\stateBB{3}{5}
      \stateZ{4}{0}\stateZ{4}{1}\stateZ{4}{2}\stateAB{4}{3}\stateBB{4}{4}\stateB{4}{5}
      \stateZ{5}{0}\stateZ{5}{1}\stateAB{5}{2}\stateBB{5}{3}\stateB{5}{4}\stateBB{5}{5}
      \stateZ{6}{0}\stateAB{6}{1}\stateBB{6}{2}\stateB{6}{3}\stateBB{6}{4}\stateB{6}{5}
      \stateB{7}{0}\stateBB{7}{1}\stateB{7}{2}\stateBB{7}{3}\stateB{7}{4}\stateBB{7}{5}
      \stateZ{8}{0}\stateAB{8}{1}\stateBB{8}{2}\stateB{8}{3}\stateBB{8}{4}\stateB{8}{5}
      \stateZ{9}{0}\stateZ{9}{1}\stateAB{9}{2}\stateBB{9}{3}\stateB{9}{4}\stateBB{9}{5}
      \stateZ{10}{0}\stateZ{10}{1}\stateZ{10}{2}\stateAB{10}{3}\stateBB{10}{4}\stateB{10}{5}
      \stateZ{11}{0}\stateZ{11}{1}\stateZ{11}{2}\stateZ{11}{3}\stateAB{11}{4}\stateBB{11}{5}
      \stateZ{12}{0}\stateZ{12}{1}\stateZ{12}{2}\stateZ{12}{3}\stateZ{12}{4}\stateAB{12}{5}
      \stateZ{13}{0}\stateZ{13}{1}\stateZ{13}{2}\stateZ{13}{3}\stateZ{13}{4}\stateZ{13}{5}
      \stateZ{14}{0}\stateZ{14}{1}\stateZ{14}{2}\stateZ{14}{3}\stateZ{14}{4}\stateZ{14}{5}
    \end{tikzpicture}}
  \end{center}
  
  \caption{Some space-time diagrams of $\Upsilon$ (state $(0,0)$ is not represented). The right bottom figure represents the morphism transforming the state at the central cell into the state at each space-time coordinate.}%%%%%%%%%%%ANAHI%%%%%%%%%%%%%%
\label{fig:stdh}
\end{figure}
\begin{prop}
  $\Upsilon$ is not $k$-expansive when $k\geq 2$ and in particular $\Upsilon$ is not pre-expansive.
  \label{prop:hnoprexp}
\end{prop}
\begin{proof}
  Let $k\geq 2$ be fixed and for each $z\in \Z$ define the configuration $c^z$ by:
  \[c^z(z') =
  \begin{cases}
    (0,0) &\text{ if }z'<z,\\
    (0,1) &\text{ if }z'=z,\\
    (1,1) &\text{ if }z<z'\le z+k-2,\\
    (1,0) &\text{ if }z'=z+k-1,\\
    (0,0) &\text{ if }z'\geq z+k.\\
  \end{cases}
  \]
  We have $c^z\di{k}\overline{(0,0)}$ and it is straightforward to
  check that $\Upsilon(c^z) = c^{z-1}$. We conclude that $\Upsilon$ is not
  $k$-expansive by Lemma~\ref{lem:leftright}.
\end{proof}

%We are now ready to prove the main result announced at the beginning of this section.
With these two examples we have proven two of the items of Theorem 1, this together with the preliminary results of section 3 and 4 allow us to conclude.
%%%%% AG: cambiié acá la redacción para ponerle más color, a mi juicio.

\begin{proof}[Proof of Theorem \ref{theo:examples}]
  Let $F$ be any irreversible and positively expansive CA. It holds:
  \begin{itemize}
  \item $\Psi$ is pre-expansive (by Proposition~\ref{prop:gprexp}) and it is
    reversible and not positively expansive (by
    Proposition~\ref{prop:secondorder} and \cite{posexponetoone});
  \item therefore $\Psi\times F$ is pre-expansive and irreversible and not positively expansive;
  \item $\Upsilon$ is $1$-expansive and reversible (by Proposition~\ref{prop:secondorder})
    but it is not pre-expansive (by Proposition~\ref{prop:hnoprexp});
  \item therefore $\Upsilon\times F$ is $1$-expansive and irreversible and not pre-expansive;
  \item finally Proposition~\ref{prop:onexpnonsurj} gives an example of $1$-expansive CA which is not surjective.
  \end{itemize}
\end{proof}

% \begin{lemma}
% \label{lemmecon}
% Let $u$ and $v$ be words that coincide outside of $\interval{a}{b}$.
% If $F$ has the property \autobis{x,y,l,r}, then for all $z\in\interval{b-l}{a+r}$, if $u_z\neq v_z$, then $F_{z+x}^y(u)\neq F_{z+x}^y(v)$. 
% \end{lemma}

% \begin{proof}
% This is the same idea as the lemma~1 from \cite{revsimu}.
% \end{proof}

% \begin{prop}[pretty obvious statement]
% Suppose that for every $\delta$, 

% \begin{itemize}
% \item there exist $x,y,l,r$ with $x,r+l>\delta$ such that $F$ fulfills \autobis{x,y,l,r};
% \item there exist $x,y,l,r$ with $-x,r+l>\delta$ such that $F$ fulfills \autobis{x,y,l,r}.
% \end{itemize}
% Then $F$ is preexpansive and Lebesgue-expansive. 
% \end{prop}

% \begin{proof}
% Let $u,v$ be almost equal words, and suppose, without loss of generality, that they coincide outside of $\interval{a}{b}$, with $a<b<0$.
% According to Lemma~\ref{lemmecon}, under the hypotheses of the proposition, there exists $x>0$ such that for some $y$, $F_{x}^y(u)\neq F_{x}^y(v)$.  Since the discrepancy between $u$ and $v$ has had to travel from the negative cells to the positive ones, there must be some $x,y$, with $\abs{x}$ smaller than the radius of $F$, such that $F_x^y(u)\neq F_x^y(v)$; hence the result.
% \end{proof}

% \begin{cor}
% If $X_p(F)$ contains points of negative and positive abscissa, then $F$ is preexpansive and Lebesgue-expansive.
% \end{cor}

% In particular, $0$ is an example of a CA that is reversible, hence nonexpansive, but preexpansive.

\subsection{Non linear examples: multiplication CA} %AQUI%

Now we exhibit a non linear family of reversible CA that will provide us with new examples of pre-expansive CA. This family has already been considered with different points of view in the literature \cite{blanchardmaass,Kari12} and underlined for its links with some Furstenberg problems in ergodic theory \cite{Pivato2009}.

Given two natural numbers $k$ and $k'$, let us consider the cellular automaton $F_{k,k'}$ on the state set $\Z_m$, with $m=kk'$, defined as follows 

\[F_{k,k'}(c)_i=kc_i\%m+\lfloor kc_{i+1}/m\rfloor\]

where $i\%j$ denotes ${i\bmod j}$ with operation precedence as follows: ${a+b\%c}$ means ${a + (b\bmod c)}$ and ${ab\%c}$ means ${(ab) \bmod c}$. Note that ${F_{k,k'}(c)_i}$ is always in $\Z_m$ because ${kc_i\%m \leq k(k'-1)}$ and ${\lfloor kc_{i+1}/m\rfloor<k}$. $F_{k,k'}$ can be seen as a multiplication by $k$ in base $kk'$, and the fact that there is no carry propagation ensures that it is a CA.
 Figure~\ref{fig:32} shows the evolution of a finite configuration in a background of $0$ under $F_{3,2}$.

\begin{figure}
  \centering
  \begin{tabular}{cp{.5cm}c}
  \includegraphics[width=.4\textwidth]{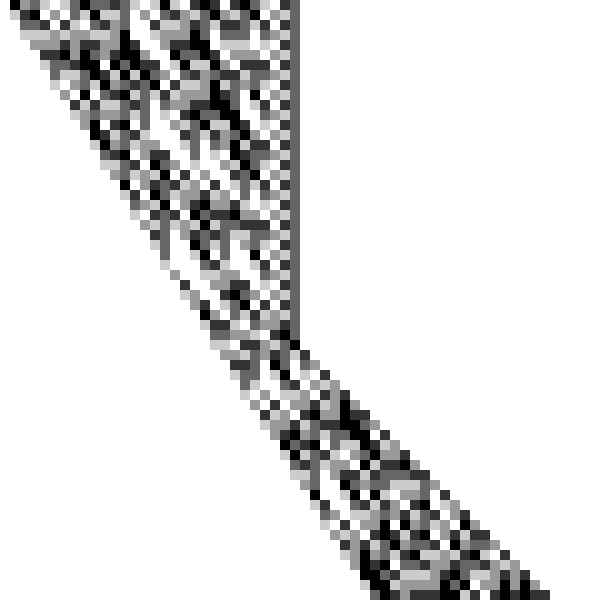} && \includegraphics[width=.4\textwidth]{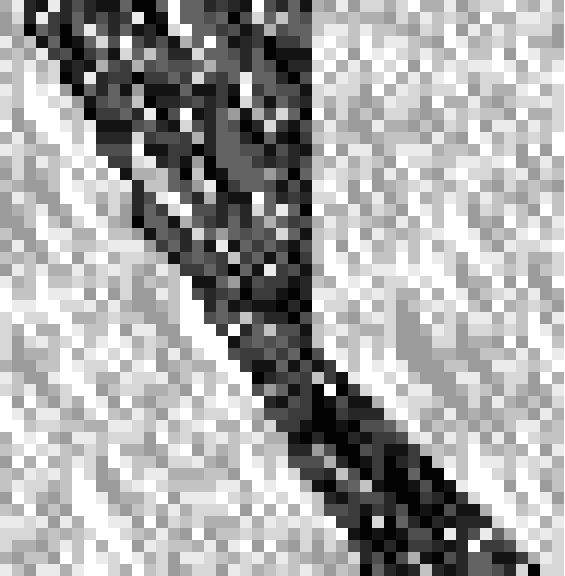}\\
  (a) && (b)\\
  \end{tabular}
  \caption{. (a) Space-time diagram of $F_{3,2}$ starting from a configuration with a finite number of non-zero cells (time goes from bottom to top, states are represented by tones of grey from white for 0 to black for 5).
  (b) Space-time diagrams of $F_{3,2}$ starting from \emph{two} random asymptotic configurations (the coincident cells are colored with light grey tones, the effect of the perturbation is colored with dark grey tones)}
  \label{fig:32}
\end{figure}

\begin{prop}\label{lem:bij}
$F_{k,k'}$ is bijective and $F^{-1}_{k,k'}=F_{k',k}\circ \sigma^{-1}$.
\end{prop}
\begin{proof}
We will show that $F_{k,k'}\circ F_{k',k}=\sigma$.
\begin{eqnarray*}
F_{k,k'}\circ F_{k',k}(c)_i&=& (kF_{k',k}(c)_i)\%m+\left\lfloor \frac{kF_{k',k}(c)_{i+1}}{m}\right\rfloor\\
&=& 0+c_{i+1}-c_{i+1}\% k+\left\lfloor \frac{(mc_{i+1})\%(mk) +k\lfloor c_{i+2}/k\rfloor}{m}\right\rfloor\\
&=& c_{i+1}-c_{i+1}\% k+\left\lfloor \frac{m(c_{i+1}-k\lfloor c_{i+1}/k\rfloor) +k\lfloor c_{i+2}/k\rfloor}{m}\right\rfloor\\
&=& c_{i+1}-c_{i+1}\% k+c_{i+1}-k\left\lfloor \frac{c_{i+1}}{k}\right\rfloor+\left\lfloor\frac{c_{i+2}-c_{i+2}\%k}{m}\right\rfloor\\
&=& c_{i+1}+0+0
\end{eqnarray*}
from here we obtain that $F_{k,k'}$ is surjective and $F_{k',k}$ is injective, but exchanging the roles of $k$ and $k'$ we obtain that both are bijective and we are done.
\end{proof}

Next lemma establishes some elementary bounds on the way perturbations propagates through $\Z$.

\begin{lemm}\label{lem:tech}
  Given a sequence $a\asy \overline{0}$ in $\{-m+1,\dots,0,\dots,m-1\}^\Z$, we define the rational number:
  \[g(a)=\sum_{i\in\Z} a_im^{-i},\]
  then, given two configurations $c\asy d\in\Z_m^\Z$ with ${c\neq d}$, the next properties hold.
  \begin{enumerate}
  \item\label{it:ly} If '$-$' represents the substraction in $\Z$, then $|g(F_{k,k'}(c)-F_{k,k'}(d))|=k|g(c-d)|$.
  \item\label{it:dist} If $i=l_0^d(c)$, and $j=r_0^d(c)$, then $m^{-j}\le|g(c-d)|<m^{-i+1}$, moreover, $-j$ is the biggest integer such that $|g(c-d)|/m^{-j}\in\Z$.
  \item\label{it:l} $l_t^d(c)< r_0^d(c)+1- t \frac{\textrm{log}(k)}{\textrm{log}(m)}$.
  \item\label{it:r} $l_0^d(c)-1- t \frac{\textrm{log}(k)}{\textrm{log}(m)}< r_t^d(c)$.
  \end{enumerate}
\end{lemm}
\begin{proof}
  \begin{enumerate}
  \item The result is directly obtained from the definition of $F_{k,k'}$ and $g$.
    We can see that  $g(F_{k,k'}(c)-F_{k,k'}(d))$ is equal to:
    \begin{eqnarray*}
      &=& \sum_{i\in\Z} (kc_i\%m+\lfloor c_{i+1}/k'\rfloor-kd_i\%m-\lfloor d_{i+1}/k'\rfloor)m^{-i}\\
      &=&\sum_{i\in\Z} (kc_{i+1}\%m-kd_{i+1}\%m)m^{-(i+1)}+\sum_{i\in\Z} (\lfloor c_{i+1}/k'\rfloor-\lfloor d_{i+1}/k'\rfloor) m^{-i}\\
      &=& \sum_{i\in\Z} (kc_{i+1}\% m+m\lfloor kc_{i+1}/m\rfloor-(kd_{i+1}\% m+m\lfloor kd_{i+1}/m\rfloor))m^{-(i+1)}\\
      &=& \sum_{i\in\Z} k(c_{i+1}-d_{i+1})m^{-(i+1)}\\
      &=& kg(c-d). 
    \end{eqnarray*}
  \item $|g(c-d)|=|\sum_{n=i}^j (c_n-d_n)m^{-n}|$ from where the result is clear.
  \item From item~\ref{it:ly} and \ref{it:dist}, we have that $m^{-l_t^d(c)+1}>|g(F^t(c)-F^t(d))|=k^t|g(c-d)|\ge k^tm^{-r_0^d(c)}$, thus $m^{-l_t^d(c)+1}>k^tm^{-r_0^d(c)}$.
  \item Analogously, we have that $m^{-r_t^d(c)}\le|g(F^t(c)-F^t(d))|=k^t|g(c-d)|<k^tm^{-l_0^d(c)+1}$, thus $m^{-r_t^d(c)}<k^tm^{-l_0^d(c)+1}$, from where we can conclude.
  \end{enumerate}
\end{proof}

 \begin{theo}
   \label{theo:pq}
   Suppose that $p_1,\ldots,p_I$ are distinct prime numbers and $o$ and $q$ are co-prime integers (possibly $1$) which are also co-prime with $p_1\cdots p_I$.
   If $k=op_1^{e_1}\cdots p_I^{e_I}$ and $k'=p_1^{e'_1}\cdots p_I^{e'_I}q$, then $F_{k,k'}$ is $\alpha$-pre-expansive if and only if
   \begin{itemize}
   \item $\alpha\in]-\log_m(k),-min\{\frac{e_i}{e_i+e'_i}:i\in\{1,\dots,I\}\}[$ and $q=1$, or
   \item $\alpha\in]-\log_m(k),0[$ and $q\not=1$.
   \end{itemize}
   
 \end{theo}
 \begin{proof}   
   We need to prove that $l_t^d(c)-\lceil \beta t\rceil$ is not lower bounded and $r_t^d(c)-\lceil \beta t\rceil$ is not upper bounded for every $c\asy d$ if and only if $\beta$ in the given interval.
   The fact about $l_t^d(c)$ can be obtained directly from Lemma~\ref{lem:tech}: first $l_t^d(c)-\lceil \beta t\rceil<r_0^d(c)+1-t\textrm{log}_m(k)-\lceil \beta t\rceil$ is not lower bounded when $\beta>-\textrm{log}_m(k)$. Conversely, if $\beta\leq-\textrm{log}_m(k)$ then $l_t^d(c)-\lceil \beta t\rceil$ is lower bounded when considering $d=\overline{0}$ and $c$ the configuration with $c_0=1$ and $c_z=0$ for ${z\neq 0}$, because in this case ${g\bigl(F^t_{k,k'}(c)-F^t_{k,k'}(d)\bigr)= k^t}$ and ${t\textrm{log}_m(k)-1\leq -l_t^d(c)\leq t\textrm{log}_m(k)+1}$.
   
   In order to prove the second fact, we will first study the way $r_t^d(c)$ varies in time.
   Let us fist suppose that $q=q_1^{f_1}\cdots q_J^{f_J}$.
   Let us suppose, without loss of generality, that $r_0^d(c)=0$, and let $n=|g(c-d)|$, which must be a natural number.
   Finally, let us assume that $n=o''p_1^{e''_1}\cdots p_I^{e''_I}q_1^{f''_1}\cdots q_J^{f''_J}$, where $o''$ is co-prime with $p_1\cdots p_Iq_1\cdots q_J$.

   From Lemma~\ref{lem:tech}(\ref{it:ly},\ref{it:dist}), we know that $s_t=-r_t^d(c)$ is the largest natural number such that $m^{s_t} | k^tn$.
   But $m^{s_t}=o^{s_t}p_1^{s_t(e_1+e'_1)}\cdots p_I^{s_t(e_I+e'_1)}q_1^{s_tf_1}\cdots q_J^{s_tf_J}$ and $k^tn=o''o^tp_1^{e''_1+te_1}\cdots p_I^{e''_I+te_I}q_1^{f''_1}\cdots q_J^{f''_J}$, thus $m^{s_t} | k^tn$ if and only if
   \begin{enumerate}
   \item $o^{s_t} | o''o^t$
   \item $\forall i\in\{1,\dots,I\}, s_t(e_i+e'_i)\le te_i+e''_i$
   \item $\forall j\in\{1,\dots,J\}, s_tf_j\le f''_j$
   \end{enumerate}

   The first condition is asymptotically true, because we know from Lemma~\ref{lem:tech}(\ref{it:r}) that $s_t<-l_0^d(c)+1+t \textrm{log}_m(k)$ which is smaller than $t$ for bigger enough $t$.
   We distinguish two cases.
   \begin{description}
   \item[($q=1$)] In this case, condition 3 is empty, $I\ge 1$, and 
     $$ s_t= \textrm{min}\left\{\left\lfloor\frac{te_i+e''_i}{e_i+e'_i}\right\rfloor:i\in\{1,\dots,I\}\right\}, \textrm{for big enough $t$},$$
     thus $r_t^d(c)-\lceil\beta t\rceil=-s_t-\lceil\beta t\rceil$ is not upper bounded if and only if $\beta < -min\{\frac{e_i}{e_i+e'_i}:i\in\{1,\dots,I\}\}$.
     \item[($q>1$)] In this case $J\ge 1$, and the condition 3 is non empty, thus for some $j\in\{1,\dots,J\}$,  $-s_t-\lceil\beta t\rceil\ge -\frac{f''_j}{f_j}-\lceil\beta t\rceil$ which is not upper bounded if and only if $\beta<0$.
   \end{description}
 \end{proof}

It is said that two integers $n$ and $m$ are \textit{multiplicatively dependent} if there are positive integers $a$ and $b$ such that ${n^a=m^b}$, and \textit{multiplicatively independent} otherwise.

 \begin{coro}
   $F_{k,k'}$ has directions of pre-expansivity if and only if $k$ and $k'$ are multiplicatively independent. 
 \end{coro}
 \begin{proof}
   First, multiplicative independence of $k$ and $k'$ is equivalent to multiplicative independence of $k$ and $m=kk'$.
   Let us first suppose there are positive integers $a$ and $b$ such that ${k^a=m^b}$.
   It is clear that $a>b$ in this case.
Then  $m^b=o^bp_1^{be_1+be'_1}\cdots p_I^{be_I+be'_I}q^b=o^ap_1^{ae_1}\cdots p_I^{ae_I}$, and we deduce that $q=o=1$ and $\frac{b}{a}=\frac{e_i}{e_i+e'_i}=\log_m(k)$ for every $i$.
Then from theorem~\ref{theo:pq} there is no $\alpha$ such that $F_{k,k'}$ is $\alpha$-expansive.

%   Using Lemma~\ref{lem:tech}(\ref{it:mult}) we have ${g(F_{k,k'}^a(x)) = m^bg(x)}$ for any configuration ${x\asy \overline{0}}$ such that ${x_z=0}$ for ${z\leq a}$. It turns out that ${g(\sigma_b(x)) = m^bg(x)}$ and, since the map $g$ is injective over the set of configurations ${X = \{x : x\asy \overline{0}\}}$, we deduce that $F_{k,k'}^a=\sigma_b$ over $X$. This equality extends over the whole space since $X$ is dense. This in particular means that $F_{k,k'}$ is not sensitive to initial conditions in direction ${\beta=-b/a}$ because for any configuration $x$ it holds ${\sigma_{\lceil\beta t\rceil}\circ F^t_{k,k'}(x) = x}$ for any ${t=0\bmod a}$, and therefore by continuity of $F_{k,k'}$ there is a constant $\alpha$ such that \[\Delta\bigl(\sigma_{\lceil\beta t\rceil}\circ F^t_{k,k'}(x),\sigma_{\lceil\beta t\rceil}\circ F^t_{k,k'}(y)\bigr)\leq\alpha\Delta(x,y)\] for all configurations $x,y$ and all positive $t$.

Suppose now that $k$ and $k'$ are multiplicatively independent and let ${p_1,\ldots,p_I}$ be the common primes in the decomposition of $k$ and $k'$ so that we have $k=op_1^{e_1}\cdots p_I^{e_I}$ and $k'=p_1^{e'_1}\cdots p_I^{e'_I}q$ where $o$ and $q$ are co-prime (but possibly $1$) and co-prime with ${p_1\cdots p_I}$ as in the hypothesis of Theorem~\ref{theo:pq}.
First, if ${q\neq 1}$ then Theorem~\ref{theo:pq} shows that $F_{k,k'}$ has a non-empty set of directions of pre-expansivity.

Suppose now that ${q=1}$, then it holds 
   \[\frac{\log(k)}{\log(m)}=\frac{\log(o)+\sum_ie_i\log(p_i)}{\log(o)+\sum_i(e_i+e'_i)\log(p_i)}\geq\frac{M\log(o) + \sum_iM(e_i+e'_i)\log(p_i)}{\log(o)+\sum_i(e_i+e'_i)\log(p_i)}= M\]
   where ${M = \min_i\frac{e_i}{e_i+e'_i}}$ because ${M\leq 1}$ and for all $i$ we have ${M\frac{e_i+e'_i}{e_i}\leq 1}$. The equality holds only if ${o=1}$ and ${M=\frac{e_i}{e_i+e'_i}}$ for all $i$, which is not the case since $k$ and $m$ are multiplicatively independent. We deduce by Theorem~\ref{theo:pq} that the set of directions of pre-expansivity of ${F_{k,k'}}$ is not empty.
 \end{proof}

 Let us interpret this result in a more geometrical way. 
 Theorem~\ref{theo:pq} establishes that perturbations will diffuse (at least sparsely) inside a cone whose left slope is $-\log_m(k)$ and right slope is $0$ or $-M$ where $M=-min\{\frac{e_i}{e_i+e'_i}:i\in\{1,\dots,I\}$. 
 Lemma~\ref{lem:tech} expresses moreover that perturbations are actually  always present close to the left boundary of the cone (with a tolerance equal to the size of the initial perturbation). In the case $q\not=1$, perturbations also accumulate on the right boundary of the cone since the CA is one-way.
% No entiendo esta frase:
 % The case $q=1$ imply that for $k^M\ge m$, thus in $t\ge M$ steps the right boundary must shift in at least one cell to the left... if any right boundary exists.
 
 In~\cite{blanchardmaass}, it was proved that $F_{k,k'}$ is left-expansive if and only if $q=1$, where \emph{left-expansive} means $\{l_t^d(c)\}_{t\in\N}$ not lower bounded for every $d\not=c$.
 The proof is performed through an argument using interesting results about the entropy of left-expansive CA.
Let us stress that left-expansivity and directional pre-expansivity are independent properties (none of them implies the other), and that any reversible CA has directions of left expansivity (any direction greater than the radius of the CA). 
Left expansivity asks for left-propagation of any perturbation, including perturbations over infinitely many cells.
That is why Lemma~\ref{lem:tech} cannot be applied in this case, and in fact these perturbation may not propagate to the left when $q\not =1$.
The conflictive cases are those corresponding to configurations which are identical up to some cell at which one configuration has a $n\not=0$ followed by 0s and the other has $n-1$ followed by $m-1$s.
Until this boundary, both configurations will evolve in the same way, thus when $q\not=1$, the cell where their differences start, which is equal to the right boundary of the first with respect to $\overline{0}$, will stop shifting to the left.

On the other hand, if $k$ and $k'$ are multiplicatively dependent, there are natural numbers $a$, $b$ such that $F^a_{k,k'}=\sigma_b$, that is, $F_{k,k'}$ is a rational power of the shift (thus left-expansive), this means that all the perturbations propagate in a unique direction, in this case the direction is $\alpha=-\frac{b}{a}$, to the left. This direction is actually a direction of equicontinuity (see \cite{Sablik08}), a strong form of non-sensitivity and prevents such CA from having any direction of pre-expansivity (see Proposition~\ref{prop:directional}).

%%%%%%%%%%%%%%%%%%%%%%%%%%%%%%%%%%%%%%%%%%%%%%%%%%%%%%

%%%%%%%%Guille: No tengo ganas de dar detalles sobre la aventura Furstenberg, la ref de Pivato que damos al prinicio de esta section me parece suficiente.
%There are several conjetures by Furstenberg about the properties of invariant sets and measures under multiplications by multiplicative independent numbers.

As a final remark, let us recall that \cite{Kari12} establishes that $F_{3,2}$ is a universal pattern generator: precisely, from every initial configuration $c$ on $\Z_6^\Z$, and every finite pattern $p$ there exists an iteration $t$ such that $p$ appear in $F^t_{3,2}(c)$ at some position.
The theorem is stated for $k=3$ and $k'=2$ but the proof only uses that $log_k(k')$ is irrational, which is exactly the condition for $F_{k,k'}$ to admit pre-expansive directions.

 % Simulations suggest that these CA are ``very randomizing'', but proving such a thing would be very hard.
%%%%%%%Guille: es un tema muy interesante, los resultados existen pero tendria que leer mucho, mejor no hablar de eso creo.
 %%%%%AG: muy a lo bruto este párrafo, merece pulido con bibliorfía en mano, lo dejo para mediados de semana.

%%%%%%%%%%%%%%%%%%%%%%%%%%%%%%%%%%%%%%%%%%%%%%%%%%%%%%%\\
%%%%%%%%%%%%%%%%%%%%%%%%%%%%%%%%%%%%%%%%%%%%%%%%%%%%%%%

\section{Cellular automata over the free group}
\label{sec:freegroup}

Some of the properties proved in the last section come from the fact
that the graph $(\Z,\{(i,i+j)\ |\ j\in V\})$ can be always disconnected by extracting a
finite part from $\Z$. %%%%%%%%%%%%%%%%%%%%%%%%%%%%%%%%%%%%%%%%%%%%%%corr(20)
The graph of any free group where the edges are given by any finite neighborhood has
this feature, and we would be able to extend some of the previous
properties to the case of a cellular automaton over the free group.
In particular, the pre-expansivity constant is strictly related with the
neighborhood size, as in $\Z$, and it does not depend on $k$ for a
$k$-expansive CA.
We denote by $\F_n$ the free group with $n$ generators ($\F_1$ is
$\Z$).

\begin{prop}\label{lem:expFG}
If $F$ is a cellular automaton with a neighborhood $V\subseteq B_r(0)$ of radius $r$ over the free group $\F_n$ and it is $k'$-expansive for all $k'\le k$, then $F$ is $k$-expansive with pre-expansivity constant equal to $2^{-r}$.
\end{prop}
\begin{proof}
The proof takes the ideas of Lemma~\ref{lem:lr}.
Let $c\di{k}d$ be two configurations in $\F_n$.
Let us call $S$ the set of standard generators of $\F_n$, including their inverses (i. e. $|S|=2n$), and for each $s\in S$, let us call $R_s$ the branch of $\F_n$ that hangs from $s$; we mean the set of elements whose shortest description in terms of $S$ starts with $s$.
In this way $\F_n=\{0\}\sqcup\left(\sqcup_{s\in S} R_s\right)$.

Now let us define $D=\{i\in\F_n\ |\ c(i)\neq d(i)\}$ and $D_s=D\cap R_s$.
We want to prove that $T_r(c)\neq T_r(d)$ so let us suppose the opposite.
This implies that $D=\sqcup_{s\in S}D_s$, and we can consider $k_s=|D_s|\le |D|=k$.
As in the case of Lemma~\ref{lem:lr}, we define configurations $c^s$ which are equal to $d$ everywhere except on branch $s$, as follows.

\begin{equation*}
c^s(i)=\left\{\begin{array}{ll} c(i)&\textrm{if } i\in R_s\\ d(i)&\textrm{otherwise}\end{array}\right.
\end{equation*}

At the beginning, $c^s$ differs from $d$ only on branch $R_s$.
We will see that this will be always the case.
Let us suppose that, for some $t\in\N$, $F^t(c^s)_i=F^t(c)_i$ for all $i\in R_s$ and $F^t(c^s)_j=F^t(d)_j$, for all $j\not\in R_s$.
We remark that, we assumed that if $j\in B_r(0)$, $F^t(c^s)_j=F^t(d)_j=F^t(c)_j$.
Since $\F_n$ is a tree and $F$ is a CA of radius $r$, if $i\in R_s$, $B_r(i)\subset R_s\cup B_r(0)$, thus $F^{t+1}(c^s)_i=F^{t+1}(c)_i$. 
If $j\in \{0\}\sqcup\left(\sqcup_{s'\in S\setminus\{s\}} R_{s'}\right)$, $B_r(j)\subset \left(\sqcup_{s'\in S\setminus\{s\}} R_{s'}\right)\cup B_r(0)$, then $F^{t+1}(c^s)_j=F^{t+1}(d)_j$.

We conclude that $T_r(c^s)=T_r(c)=T_r(d)$, for every $s\in S$.

But we know, by hypothesis, that $F$ is $k'$-expansive, let us take its pre-expansivity constant as $\epsilon=2^{-m}$.
Let us consider now the configuration $\sigma_{-ms}(c^s)$, the shift of $c^s$ by $ms\in\F_n$.
By construction, $\sigma_{-ms}(c^s)$ and $\sigma_{-ms}(d)$ are equal over $B_m(0)$.  
By the $k'$-expansivity of $F$, there exists a time $t$ and $j\in B_m(0)$ such that $F^t(\sigma_{-ms}(c^s))_j\neq F^t(\sigma_{-ms}(d))_j$.
But, as shown before, ${F^t(c^s)}$ and ${F^t(d)}$ differ only on branch $R_s$.
Therefore ${F^t(\sigma_{-ms}(c^s))}$ and ${F^t(\sigma_{-ms}(d))}$ differ only on branch $R_{ms}$ which is disjoint from $B_m(0)$: this is a contradiction.
\end{proof}

The last proposition shows that being $k$-expansive for every $k\in\N$ is enough to being pre-expansive, as in $\F_1=\Z$.
But not all the properties survive from $\F_1$ to $\F_n$, when $n>1$; $k$-expansivity is possible for infinitely many $k$'s in $\F_n$ even without pre-expansivity, as the next example shows.

\begin{exam}[$\Lambda_n$]
  Let ${Q=\{0,1\}}$, $+$ be the addition modulo $2$, and $\Lambda_n$ be the
  CA defined over $Q^{\F_n}$ by 
  $$\Lambda_n(c)_i=c(i)+\sum_{j\in S} c(i+j).$$
\end{exam}

In this CA, a spot will produce a wave of 1s advancing with velocity 1 over the boundary of a ball, as the next lemma establishes.

\begin{lemm}\label{lem:couches}
If ${i,j\in\F_n}$ are such that $\norm{i}=\norm{j}$, then for every $t\in\N$ $\Lambda_n^t(c^1)_i=\Lambda_n^t(c^1)_j$, moreover $\Lambda_n^{\norm{i}}(c^1)_i=1$.
\end{lemm}
\begin{proof}
We prove by induction on $l$ that for every $t\le l$ and every $x$, $y\in B_l(0)$, $\left[\norm{i}=\norm{j}\Rightarrow \Lambda_n^t(c^1)_i=\Lambda_n^t(c^1)_j\right]$ and that $\Lambda_n^l(c^1)_i=1$ if $\norm{i}=l$. 

For $l=0$ is clear since in this case $i=0=j$ and $\Lambda_n^0(c^1)_0=1$.
Now let us suppose it true for some $l$, and let us prove it for $l+1$.
\begin{description}
\item{Case 1,} $t\le l$. By the induction hypothesis, we only need to verify the property for $i,j\in B_{l+1}(0)-B_l(0)$, but $\Lambda_n^t(c^1)_i=0=\Lambda_n^t(c^1)_j$ because at time $t\le l$ no perturbation at 0 has the time to arrive to these cells.
\item{Case 2,} $t=l+1$. We first remark that any cell $i$ in $\F_n$ has exactly $2n-1$ neighbors farther and exactly one neighbor closer than $i$ to $0$; we also remark that the local rule of $\Lambda_n$ is totalistic, only the quantity of neighbors at a given state counts.
If $i,j\in B_l(0)$, all of their neighbors are in $B_{l+1}(0)$, thus by Case 1, their state at time $l$ depends only on their distance to 0, thus $\Lambda_n^{l+1}(c^1)_i=\Lambda_n^{l+1}(c^1)_j$.
If $i,j\in B_{l+1}(0)\setminus B_{l}(0)$, then their neighbors outside $B_l(0)$ and themselves have all state 0 at time $l$; their unique neighbors in $B_l(0)$ have both state 1, by induction hypothesis. Thus, by the definition of $\Lambda_n$,  $\Lambda_n^{l+1}(c^1)_i=\Lambda_n^{l+1}(c^1)_j=1$.
\end{description}
\end{proof}

\begin{prop}
  $\Lambda_n$ is $k$-expansive for every $k$ odd with pre-expansivity constant equal to 1, and it is not $2$-expansive if $n\ge 2$.
  %%%%% AG: puse la cte de pre-expansividad.
\end{prop}
\begin{proof}
Let $c\di{k}\overline 0$.
Let $D=\{i\ |\ c(i)\neq 0\}$ and let $D_l=D\cap(B_l(0)\setminus B_{l-1}(0))$.
It is clear that $c=\sum_{i\in D}\sigma_{-i}(c^1)$.
Since $|D|=k$ is odd, there exists some $l$ such that $|D_l|$ is odd, let us take $\overline{l}$ as the smallest one.
For every $x,y\in D_l$, $T_0(\sigma_{-x}(c^1))=T_0(\sigma_{-y}(c^1))$, thanks to lemma~\ref{lem:couches}.
Therefore, given a cell $y\in D_{\overline{l}}$, 
\begin{eqnarray*}
\Lambda_n^{\overline{l}}(c)_0
&=&\Lambda_n^{\overline{l}}(\sum_{l\in\N}\sum_{x\in D_l}\sigma_{-x}(c^1))_0\\
&=&\Lambda_n^{\overline{l}}(\sum_{l=0}^{\overline{l}}\sum_{x\in D_l}\sigma_{-x}(c^1))_0\\
&=& \sum_{l=0}^{\overline{l}}\Lambda_n^{\overline{l}}(\sum_{x\in D_l}\sigma_{-x}(c^1))_0\\
&=& \Lambda_n^{\overline{l}}(\sum_{x\in D_{\overline{l}}}\sigma_{-x}(c^1))_0\\
&=& \Lambda_n^{\overline{l}}(\sigma_{-y}(c^1))_0.
\end{eqnarray*}
By Lemma~\ref{lem:couches}, this last term is equal to 1 which proves the $k$-expansivity when $k$ is odd.

The second part of the proposition is almost direct from lemma~\ref{lem:couches}.
In fact, let $m\in \N$ be any natural number and let us take $z=ms$ for some fixed generator $s$.
Now let $s'$ be another generator, different from $s$ and $-s$ and define $x=z+s'$ and $y=z-s'$.
This imply that $\norm{x}=\norm{y}=\norm{z}+1=m+1$.
Lemma~\ref{lem:couches} says that $T_0(\sigma_{x}(c^1))=T_0(\sigma_{y}(c^1))$, but also that $T_m(\sigma_{x}(c^1))=T_m(\sigma_{y}(c^1))$, because $x$ and $y$ are equidistant from $z$, as well as from all the other members of $B_m(0)$.
\end{proof}

\section{Cellular Automata on $\Z^n$, with $n\ge 2$}
\label{sec:2D}

Expansivity is not possible in dimension ${n\geq 2}$ or more, due to
combinatorial reasons: the number of possible $n$-dimensional patterns
grows too quickly to be uniformly conveyed into a 1-dimensional array
without loss (see \cite{Pivato11} for a general result of inexistence of
expansive CA). This argument does not apply to pre-expansivity because only finite differences
have to be propagated. Nevertheless, in Abelian  CA the information propagates in a very regular way, and pre-expansivity is impossible as we will show.

\subsection{No pre-expansivity for Abelian CA in dimension 2 or higher}

\begin{theo}
  \label{thm:abelian2D}
  No Abelian CA of dimension $d\geq 2$ is pre-expansive.
\end{theo}
\begin{proof}
  First, if $G$ is the Abelian group of the theorem it can be
  decomposed in a direct product $G=G_p\times G'$ where $G_p$ is a
  finite $p$-group for some prime $p$ and $G'$ is a group
  whose order $m$ is such that $p$ doesn't divide $m$ (structure
  theorem for finite Abelian groups, see
  \cite{textbookabeliangroups}). Then $F$ is isomorphic to ${F_p\times
    F'}$ according to Lemma~\ref{lem:groupdecomp}, where $F_p$ is
  linear over $G_p$. Moreover if $F$ is pre-expansive, then $F_p$ must
  also be pre-expansive (by Proposition~\ref{prop:general}). It is
  therefore sufficient to show the Theorem for $p$-groups.

  Now consider $F$ of dimension $d\geq 2$ linear over a $p$-group and
  some $m\geq0$. By Lemma~\ref{lem:prefijotraza}, we know that the
  trace $T_m$ of a finite configuration of size $n$ is determined by its prefix of size $\lambda(n)$ where
  ${\lambda\in O(n)}$. The number of such finite configurations grows
  like $\alpha^{n^d}$ for some $\alpha>0$ and the number of prefixes of $T_m$ of length $\lambda(n)$
  grows like $\beta^{\lambda(n)}$ for some $\beta>0$ which depends only on $m$, $G_p$ and $d$. 
  Since $d\geq 2$
  and $\lambda$ is linear we deduce for $n$ large enough that two
  finite configurations of size $n$ have the same trace
  $T_m$. Therefore $T_m$ is not pre-injective and by
  Proposition~\ref{prop:general}, $F$ is not pre-expansive.
\end{proof}

Note that this does not avoid a priori the existence of a linear CA
which is $k$-expansive for any $k\in \N$ or for infinitely many $k$.

\subsection{Simple Abelian CA}

In general, in a CA with neighborhood $V\subset \G$, we can remark that the influence of the cell $0$ is restricted to the set generated by linear combinations of $-V$.
More precisely, at time $t$, its influence is restricted to the following set:

\[ -V_t(0)=\left\{ \sum_{i=1}^t v_i\ |\ (\forall i\in\{1,..,t\})\ v_i\in -V\right\}
\]

A perturbation in a cell $u\in\G$ can produce a change in the state of cells in ${-V_t(u)=u-V_t(0)}$ up to time $t$.

If $\G$ is commutative, for example $\G=\Z^n$ and $V=\{v_1,..,v_m\}$, this set can be computed as follows.

\begin{eqnarray*}
-V_t(0)&=& \left\{ \sum_{k=1}^m n_k(-v_k)\ |\ \sum_{k=1}^m n_k=t \textrm{ and for each } k, n_k\in \N\right\}\\
&=&\left\{ \sum_{k=1}^m \frac{n_k}{t}(-tv_k)\ |\ \sum_{k=1}^m n_k=t \textrm{ and for each } k, n_k\in \N\right\}\\
&\subseteq &\left\{ \sum_{k=1}^m \lambda_k(-tv_k)\ |\ \sum_{k=1}^m \lambda_k=1\textrm{ and for each } k, \lambda_k\in[0,1],\right\}\\
&\subseteq & co(-tV)
\end{eqnarray*}

Where $tV=\{t v\ |\ v\in V\}$, and $co(\cdot)$ stands for the \emph{convex hull} (in $\R^n$).

In the simpler case where $G_p=\Z_p$, any linear CA $F$ can be expressed as

$$ F=\sum_{z\in V} a_z\sigma_z, $$

where $(a_z)_{z\in V}$ is a sequence of elements of $\Z_p$.
When $p$ is prime, the Frobenius endomorphism gives strong self-similar properties to linear CAs, more precisely:
\[ F^{p^k}=\sum_{z\in V} a_z\sigma_{p^kz}.\]

More generally, consider any Abelian CA $F$ with states $G_p$ such that ${p\cdot g=0}$ for all $g\in G_p$ and such that 

$$ F=\sum_{z\in V} h_z\circ\sigma_z, $$
where $h_z$ are commuting automorphisms of $G_p$. The Binomial formula and the fact that $p$ divides all ${p^k\choose i}$ for ${0<i<p^k}$ gives:
\[ F^{p^k}=\sum_{z\in V} h_z^{p^k}\sigma_{p^kz}.\]
Since the $h_z$ are automorphisms of $G_p$ there are infinitely many $k$ such that ${h_z=h_z^{p^k}}$ for all ${z\in V}$ and therefore
\[ F^{p^k}=\sum_{z\in V} h_z\sigma_{p^kz}.\]

These particular cases suggest the following definition.

\begin{defi}
  An Abelian CA ${F=\sum_{z\in V} h_z\circ\sigma_z}$ is \emph{simple} if it verifies
  \begin{equation}
    F^{M}=\sum_{z\in V} h_z\sigma_{Mz}\label{eq:Lucas}
  \end{equation}
  for arbitrarily large $M$.
\end{defi}

The next lemma establishes that the constant of $k$-expansivity in a simple Abelian CA on $\Z_p$ depends only on the radius of the neighborhood.
The radius of a neighborhood $V$ is the smallest integer $r$ such that $V\subseteq B_r(0)$.

%Thus, $F^{p^t}(c_a)$ will be equal to $a_za$ at $-p^tz$, for every $z\in V$ and 0 elsewhere.

\begin{lemm}[Amplification]
Let $F$ be a simple Abelian CA with neighborhood $V\subset\Z^n$ of radius $r$.
If there exists a configuration $c\not =\overline{0}$ such that $T_{r}(c)=0$, then for any $m\ge r$ there exists a configuration $c'\not =\overline{0}$ such that $T_m(c')=0$. 
\end{lemm}
\begin{proof}
Let $c$ be such that $T_{r}(c)=0$ and let $M$ be such that $m\le M-1$.
We define $c'$ by $c'_{Mx}=c_{x}$ for every $x\in\Z^n$ and 0 elsewhere.

From Equation~\ref{eq:Lucas}, it is easy to see that $F^{tM}(c')_{Mx}=F^t(c)_{x}$, and 0 elsewhere.
Therefore, for every $t\in\N$ and every $v\in B_{r}(0)$, $F^{tM}(c')_{Mv}=0$.

Now, between iterations $tM$ and $(t+1)M$, we know, from the former remarks, that only cells in $\Omega=\displaystyle{\bigcup_{x\not \in B_{r}(0)} (Mx-V_{M}(0))}$ can have a state different from 0.
Since $-V\subseteq B_r(0)$, we have that $-V_{M}(0)\subseteq B_{rM}(0)$, and the complement of $\Omega$ contains $B_{M-1}(0)$, which is what we were looking for, in fact,
\begin{eqnarray*}
y&\in &\Omega=\left(\displaystyle{\bigcup_{x\not \in B_{r}} Mx-V_{M}(0)}\right)\\
&\Rightarrow & (\exists x\not \in B_{r}(0))(\exists v\in -V_{M}(0))\ y=Mx+v\\
&\Rightarrow & \norm{y}\ge \norm{Mx}-\norm{v}\ge M(r+1)-Mr=M\\
&\Rightarrow & y\not \in B_{M-1}(0). 
\end{eqnarray*}
\end{proof}

The next corollary shows that in this case again the expansivity constant depends only on the neighborhood radius.
Let us remark that here it is a little bit stronger than in the 1-dimensional case because it does not need $k'$ expansivity for every $k'\le k$.
%%%%% AG: agregué este comentario.

\begin{coro}\label{coro:amplif}
  Let $F$ be a linear CA in $\Z_p$. It holds:
  \begin{itemize}
  \item $F$ is $k$-expansive, if and only if $F$ is
    $k$-expansive with pre-expansivity constant $2^{-r}$;
  \item $F$ is $k$-expansive for all $k\in\N$ if and only if $F$ is pre-expansive.
  \end{itemize}
\end{coro}

With this lemma we can establish $k$-expansivity just by looking at $T_{r}$.
We will show a CA in that setting which is 1-expansive, 3-expansive and non 2-expansive,
and another which is non 1-expansive (and so non $k$-expansive for every $k$).
%%%%%%%%%%%%%%%%%%%%%%%%%%%%%%%%%%%%%%%%%%%%%%

\subsubsection{The rule $\oplus_2$ with von Neumann neighborhood in $\Z^2$}

The rule that simply sums the state of its 5 neighbors in the von Neumann neighborhood: $(0,0), (0,1), (1,0), (0,-1), (-1,0)$ is not 2-expansive.
This can be seen through a simple picture: let us suppose that we start with the configuration $c$ that has a `1' in cell $(-2^k,2^{k-1})$ and in cell $(2^k,2^{k-1})$.
By symmetry, the vertical line $\{0\}\times \Z$ will be always null.
Thus, at iterations $t2^k$ only cells at $2^k(\Z\setminus\{0\})\times \Z$ will be activated.
Between iterations $t2^k$ and $(t+1)2^k$ these cells cannot influence the ball $B_{2^{k-1}-1}(0,0)$  % (see figure~\ref{fig:2vN}) 
and this ball will have a null trace: $T_{2^{k-1}}(c)=0$.

% \begin{figure}
%   \begin{center}
%     \includegraphics[width=8cm]{2vN}
%   \end{center}
% \caption{Potentially active cells at iteration $t2^k$ cannot influence $B_{2^{k-1}-1}(0)$ before iteration $(t+1)2^k$. Big dots represent the initially active cells.}
% \label{fig:2vN}
% \end{figure}

In order to establish the 3-expansivity of this CA, we will start by proving some lemmas that describe the form of the traces $T_1(\sigma_z(c^1))$ of the evolution of the configuration $c^1$ at the different points of $\Z^2$.
In order to achieve this, we start by computing the partial traces $T_0(\sigma_z(c^1))|_{[0,2^{k}-1]}$ and $T_0(\sigma_z(c^1))|_{[2^{k},2^{k+1}-1]}$.
We first give a way for compute them, and afterwards we prove they are effectively the partial traces.

\begin{defi}\label{def:subst}
Given $k\ge 0$ and $z\in B_{2^k-1}(0,0)$, we recursively define $u_k(z)$ and $v_k(z)$ as follows.
Let us define $S_k=\{(0,0),(0,2^k),(2^k,0),(0,-2^k),(-2^k,0)\}$, the active cells of iteration $2^k$.
\begin{eqnarray*}
u_0(z)&=&v_0(z)=1;\\
u_k(z)&=&\left\{\begin{array}{ll}
u_{k-1}(z)v_{k-1}(z) & \textrm{if }z\in B_{2^{k-1}-1}(0,0)\\
0^{2^{k-1}}u_{k-1}(z-x) & \textrm{if }z\in B_{2^{k-1}-1}(x)\setminus B_{2^{k-1}-1}(0,0)\textrm{ and } x\in S_{k-1}\\
0^{2^k}&\textrm{otherwise}
\end{array}\right.\\
v_k(z)&=&\left\{\begin{array}{ll}
u_{k-1}(z)u_{k-1}(z) & \textrm{if }z\in B_{2^{k-1}-1}(0,0)\\
u_{k-1}(z-x)u_{k-1}(z-x) & \textrm{if }z\in B_{2^{k-1}-1}(x)\setminus B_{2^{k-1}-1}(0,0)\textrm{ and } x\in S_k\\
0^{2^k}&\textrm{otherwise}
\end{array}\right.
\end{eqnarray*}
\end{defi}

\begin{lemm}\label{teo:uv}
If $z\in B_{2^k-1}(0,0)$, then $u_k(z)$ and $v_k(z)$ represent the trace of $z$ from $0$ to $2^k-1$ and from $2^k$ to $2^{k+1}-1$ respectively.
\end{lemm}

\begin{proof}
When $k=0$, $B_{0}(0,0)=\{(0,0)\}$, and the trace of $(0,0)$ is constant and equal to 1.

Let us suppose the lemma true for $k-1\ge 0$.
Let $z\in B_{2^k-1}(0,0)$.

Figure~\ref{fig:zones}(a) depicts the first two cases in the definition of $u_k(z)$, the last one corresponds to cells in the diagonals segments, which remains null by symmetry .
\begin{itemize}
\item{Case 1)} $z\in B_{2^{k-1}-1}(0,0)$.
In this case, the induction hypothesis says that the trace until $2^{k}-1$ is given by $u_{k-1}(z)v_{k-1}(z)$
\item{Case 2)} $z\in B_{2^k-1}(0,0)\setminus B_{2^{k-1}-1}(0,0)$.
From 0 to $2^{k-1}-1$, $z$ has not been touched yet, thus its trace until $2^{k-1}-1$ is $0^{2^{k-1}}$.
At iteration $2^{k-1}$, only the cells in $S_{k-1}$ are in state 1, thus $z$ is influenced by only one of the cells in $S_{k-1}$, say $x$, its trace from $2^{k-1}$ to $2^k-1$ is equal to the trace of the cell $z-x$ from 0 to $2^{k-1}-1$, thus by induction again, it is equal to $u_{k-1}(z-x)$.
\item{Case 3)} If $z$ belong to none of these balls, it belongs to one of the two diagonals lines that pass through $(0,0)$, and its trace is null.
\end{itemize}

Figure~\ref{fig:zones}(b) depicts the three cases in the definition of $v_k(z)$.
\begin{figure}
\begin{tabular}{cp{1cm}c}
\includegraphics[width=3cm]{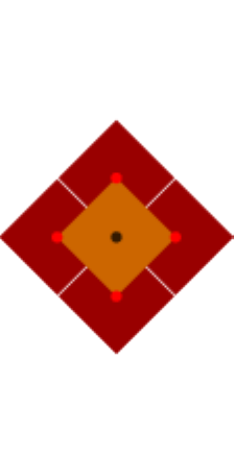} &&
\includegraphics[width=6cm]{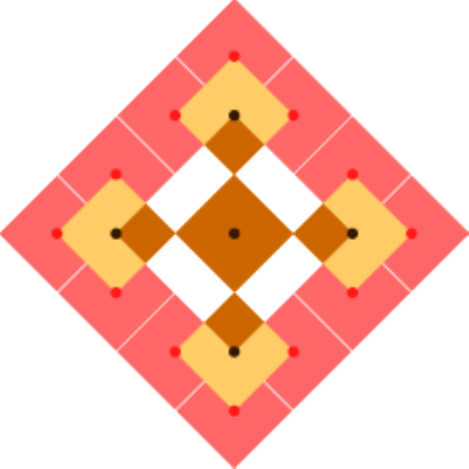}\\
(a) && (b)\\
\end{tabular}
\caption{(a) Represents the evolution from iteration 0 to $2^k$. Non-null cells at iterations 0 and  $2^{k-1}$ are marked with dark brown and red dots, respectively (cell (0,0) is active at both instants). The balls of radius $2^{k-1}$ around these cells are colored with similar lighter colors.
(b) Represents the evolution from iteration $2^k$ to $2^{k+1}$. Dark brown dots represent non-null cells at iterations $2^k$ and $2^k+2^{k-1}$, while red dots represent the cells which are non-null at iteration $2^k+ 2^{k-1}$. The balls of radius $2^{k-1}$ around these cells are colored with similar lighter colors. Faded colors represent the cells outside the ball of radius $2^k$.
}
\label{fig:zones}
\end{figure}
\begin{itemize}
\item{Case 1)} $z\in B_{2^{k-1}1}(0,0)$. At iteration $2^k$, only the cells in $S_{k}$ are in state 1. 
Therefore, from $2^k$ to $2^k+2^{k-1}-1$, $z$ will be influenced only by the cell $(0,0)$, and its trace will be equal to $u_{k-1}(z)$.
At iteration $2^k+2^{k-1}$, the active cells corresponds to \emph{red} and \emph{brown} cells in Figure~\ref{fig:zones}(b), and again only cell $(0,0)$ reaches $z$, the process is repeated.
\item{Case 2)} $z\in B_{2^k-1}(0,0)\setminus B_{2^{k-1}-1}(0,0)$. At iteration $2^k$, only the cells in $S_{k}$ are in state 1, thus, before iteration $2^k+2^{k-1}$, $z$ is touched only if it is at distance less than $2^{k-1}$ from one of the cells in $S_k$, say $x$, (orange zone in Figure~\ref{fig:zones}(b)), then its trace is equal to $u_{k-1}(z-x)$.
At iteration $2^k+2^{k-1}$, the active cells are far again, and $z$ is influenced only by $x$ again.
\item{Case 3)} If $z$ belong to none of these balls, its trace is null.
\end{itemize}
\end{proof}

This lemma proves that the traces can be obtained through the substitution $u\rightarrow uv$ and $v\rightarrow uu$.
The basic $u$ and $v$ for a given cell $z$ are obtained at iteration $2^{k+1}$ if $B_{2^k-1}(0,0)$ is the smallest ball containing $z$.
From the next lemma we can conclude the 1-expansivity of this automaton with pre-expansivity constant equal to $2^{-1}$.

\begin{lemm}\label{lem:odd}
If $i+j$ is odd and smaller than $2^k$, then the trace of the cell $z=(i,j)$, $T_{0}(\sigma_{z}(c^1))$ is not null and its first non null index is odd and smaller than $2^k$, in particular $u_k(i,j)$ is not null.
\end{lemm}
\begin{proof}
For $k=1$, the trace of the odd cells inside $B_1(0,0)$ from 0 to 1 is $u=01$, the result holds.
Let us assume the result true for $k\ge 1$, and let $z=(i,j)$ be an odd cell in $B_{2^{k+1}-1}(0,0)\setminus B_{2^{k}-1}(0,0)$.
Since the cell is odd, it is not over the diagonals, and it belongs to the ball of radius $2^k$ of one of the four cells of $S_k\setminus\{(0,0)\}$, say $x$.
Then, by lemma~\ref{teo:uv}, its trace from $2^k$ to $2^{k+1}-1$ is given by $u_{k}(z-x)$.
Since the cells of $S_k$ are even, $(z-x)$ is also odd, and the conclusion follows by induction.
\end{proof}

Now we will prove several properties that will be useful to prove 3-expansivity.

\begin{lemm}\label{lem:v}
Given $k>0$, if $z\in B_{2^k-1}(0,0)$ then $v_k(z)$ is a square.
\end{lemm}
\begin{proof}
It is clear from Definition~\ref{def:subst} and the fact that $k>0$.
\end{proof}

\begin{lemm}\label{lem:u}
Given $k>0$, if $z\in B_{2^k-1}(0,0)\setminus\{(0,0)\}$ and $u_k(z)\not = 0^{2^k}$, then $u_k(z)$ is not a square.
\end{lemm}
\begin{proof}
By induction on $k$.
For $k=1$, if $z$ is inside the von Neumann neighborhood of $(0,0)$, thus $u_1(z)=01$ which is not a square.
Let us suppose the assertion true for $k-1\ge 1$.
Since $k\ge 1$, from Definition~\ref{def:subst}, we recognize two cases for $u_k(z)$.
\begin{description}
\item[Case 1: $u_k(z)=0^{2^k}u_{k-1}(z-x)$, for some $x\in S_{k-1}$.]
The only way for $u_k(z)$ to be a square is to be equal to $0^{2^k}$.
\item[Case 2: $u_k(z)=u_{k-1}(z)v_{k-1}(z)$.] By the induction hypothesis $u$ is either null or not a square. From Lemma~\ref{lem:v} $v_{k-1}(z)$ is a square, then $u_k(z)$ is a square if and only if $u_k(z)=0^{2^k}$.
\end{description}
\end{proof}

\begin{lemm}\label{lem:diag}
If $|i|=|j|$ and $(i,j)\not=0$, then the trace of the cell $(i,j)$, $T_{0}(\sigma_{(i,j)}(c^1))$ is equal to $0^\omega$.
\end{lemm}
\begin{proof}
Cells in the diagonals systematically falls in the boundaries of the zones given by the substitution, then their traces are systematically assigned equal to 0.
\end{proof}

\begin{lemm}\label{lem:even}
If $i+j$ is even, smaller than $2^k$ and the trace of the cell $z=(i,j)$, $T_{0}(\sigma_{z}(c^1))$ is not null, then the first non null index of the trace is even and it is smaller than $2^k$.
\end{lemm}
\begin{proof}
For $k=0$, the trace of the even cell inside $B_0(0,0)$ from 0 to 0 is $u=1$, the result holds.
Let us assume the result true for $k\ge 0$, and let $z=(i,j)$ be an even cell in $B_{2^{k+1}-1}(0,0)\setminus B_{2^{k}-1}(0,0)$.
Since the cell is attained, from Lemma~\ref{lem:diag}, it is not over the diagonals, then it belongs to the ball of radius $2^k$ of one of the four cells in $S_k$, say $x$.
Then, its trace from $2^k$ to $2^{k+1}-1$ is given by $u_k(z-x)$.
Since the cells in $S_k$ are even, $z-x$ is also even, and the conclusion follows by induction.
\end{proof}

Now we are ready to prove that this automaton is 3-expansive.

\begin{lemm}\label{teo:3exp}
If $z_1$, $z_2$ and $z_3$ are three different cells and $T_{0}(\sigma_{z_1}(c^1))+T_{0}(\sigma_{z_2}(c^1))+T_{0}(\sigma_{z_3}(c^1))=0^\omega$, then there exists $z\in\{z_1,z_2,z_3\}$ such that $T_{0}(\sigma_{z}(c^1))=0^\omega$.
\end{lemm}
\begin{proof}
We will prove a stronger assertion:
\begin{quote}
If $z_1$, $z_2$ and $z_3$ are three different cells in $B_{2^k-1}(0,0)$ and $u_k(z_1)+u_k(z_2)+u_k(z_3)=0^{2^k}$, then there exists $z\in\{z_1,z_2,z_3\}$ such that $u_k(z)=0^{2^k}$.
\end{quote}
It is stronger because from Lemmas~\ref{lem:odd} and~\ref{lem:even}, if $z\in B_{2^k-1}(0,0)$ and $u_k(z)=0^{2^k}$, then $T_{0}(\sigma_{z}(c^1))=0^\omega$.

By contradiction, let $z_1$, $z_2$ and $z_3$ be three different cells in a ball $B_{2^k-1}(0,0)$ such that $u_k(z_1)+u_k(z_2)+u_k(z_3)=0^{2^k}$ and for all $i\in\{1,2,3\}$, $u_k(z_i)\not=0^{2^k}$.
Let us choose these cells such that $k$ is as small as possible.

Let be $t_{z_1}$, $t_{z_2}$ and $t_{z_3}$ the indices where the respective traces equals 1 by the first time.
It is clear that two of these numbers are equal and smaller than the third.
Let us suppose that $t_{z_1}=t_{z_2}<t_{z_3}$.
\begin{description}
\item[Case 1: $z_1$, $z_2\in B_{2^{k-1}-1}(0,0)$.] In this case, from Lemma~\ref{teo:uv}, the trace of $z_3$ from $2^{k-1}$ to $2^{k}-1$ is equal to $u_k(z_3-x)$, for some $x\in S_k$, thus $u_k(z_3)=0^{2^{k-1}}u_{k-1}(z_3-x)$.
On the other hand, $u_k(z_i)=u_{k-1}(z_i)v_{k-1}(z_i)$ for $i\in\{1,2\}$.
From Lemma~\ref{lem:v}, $v_{k-1}(z_1)$ and $v_{k-1}(z_2)$ are squares if $k>0$, which is the case, since $B_{0}(0,0)$ contains only one cell.
Thus $v_{k-1}(z_1)+v_{k-1}(z_2)$ is also a square, then it cannot be equal to $u_{k-1}(z_3-x)$ which is not a square thanks to Lemma~\ref{lem:u}.
\item[Case 2: $z_1$, $z_2$, $z_3\not\in B_{2^{k-1}-1}(0,0)$.] In this case, from Lemma~\ref{teo:uv}, for each $i\in\{1,2,3\}$ there exists some $x_i\in S_{k-1}$ such that $u_k(z_i)=0^{2^{k-1}}u_{k-1}(z_i-x_i)$.
Of course $z_i-x_i\in B_{2^{k-1}-1}(0,0)$ for each $i\in\{1,2,3\}$ and $u_{k-1}(z_1-x_1)+u_{k-1}(z_2-x_2)+u_{k-1}(z_3-x_3)=0^{2^{k-1}}$, which contradicts the minimality of $k$.
\end{description}
\end{proof}

\begin{theo}
The rule $\oplus_2$ with von Neumann neighborhood in $\Z^2$ is 1-expansive, 3-expansive but not $2k$-expansive for every $k\in\N$.
\end{theo}
\begin{proof}
  We have stated that the rule is not 2-expansive, through a counter example, and this fact can be extended to every even number $2k$, thanks to Proposition~\ref{prop:linear}.\\
  Since $B_1(0,0)$ always contains odd cells, from Lemma~\ref{lem:odd} we know that $T_1(\sigma_z(c^1)$ is never null, which proves that the CA is 1-expansive.\\
  Finally, Lemma~\ref{teo:3exp} shows that if three cells produce a null trace of radius 1, then one of them has a null trace of radius 0, this means that this cell is even, and its four neighbors are odd. 
When looking at the neighbors of these cells, their sum is also null, for each if its neighbors.
Thus at least one cell must have even neighbors with a null trace, but in this case, two of these cells are odd, and its four even neighbors cannot equal the odd neighbors of the first cell, and then the trace of radius 1 of the sum of the three cells cannot be null.
\end{proof}

%Unfortunately, the proof techniques cannot be generalised to every $k$ odd.

\subsubsection{The rule $\oplus_2$ with triangular neighborhood}

The last rule is 1- and 3-expansive, now we present a linear rule which is
even not 1-expansive. Thanks to Proposition~\ref{prop:linear}, it
implies in particular that is not $k$-expansive, for any $k\in \N$.
It correspond to addition modulo 2 as the last one but with a triangle shaped neighborhood: $N=\{(-1,1),(1,1),(0,0),(0,-1)\}$.
We only need to prove that it is not 1-expansive with expansivity constant equal to $2^{-2}$ thanks to Corollary~\ref{coro:amplif}.

\renewcommand\stateZ[2]{}
\renewcommand\stateO[2]{\draw[fill=red] (#1,#2) rectangle +(1,1);}
\renewcommand\stateZO[2]{\draw[fill=blue] (#1,#2) rectangle +(1,1);}
\renewcommand\stateOO[2]{\draw[fill=red] (#1,#2) rectangle +(1,1);}
\renewcommand\stateZZO[2]{}
\renewcommand\stateOZO[2]{\draw[fill=red] (#1,#2) rectangle +(1,1);}
\begin{figure}
  \begin{center}
    \begin{tikzpicture}[scale=.1]
\stateZ{0}{0}\stateZ{0}{1}\stateZ{0}{2}\stateZ{0}{3}\stateZ{0}{4}\stateZ{0}{5}\stateZ{0}{6}\stateZ{0}{7}\stateZ{0}{8}\stateZ{0}{9}\stateZ{0}{10}\stateZ{0}{11}\stateZ{0}{12}\stateZ{0}{13}\stateZ{0}{14}\stateZ{0}{15}\stateZ{0}{16}\stateZ{0}{17}\stateZ{0}{18}\stateZ{0}{19}\stateZ{0}{20}\stateZ{0}{21}\stateZ{0}{22}\stateZ{0}{23}\stateZ{0}{24}\stateZ{0}{25}\stateZ{0}{26}\stateZ{0}{27}\stateZ{0}{28}\stateZ{0}{29}\stateZ{0}{30}\stateZ{0}{31}\stateZ{0}{32}\stateZ{0}{33}\stateZ{0}{34}\stateZ{0}{35}\stateZ{0}{36}\stateZ{0}{37}\stateZ{0}{38}\stateZ{0}{39}\stateZ{0}{40}\stateZ{0}{41}\stateZ{0}{42}\stateZ{0}{43}\stateZ{0}{44}\stateZ{0}{45}\stateZ{0}{46}\stateZ{0}{47}\stateZ{0}{48}\stateZ{0}{49}\stateZ{0}{50}\stateZ{0}{51}
\stateZ{1}{0}\stateOO{1}{1}\stateZ{1}{2}\stateZ{1}{3}\stateZ{1}{4}\stateZ{1}{5}\stateZ{1}{6}\stateZ{1}{7}\stateZ{1}{8}\stateZ{1}{9}\stateZ{1}{10}\stateZ{1}{11}\stateZ{1}{12}\stateZ{1}{13}\stateZ{1}{14}\stateZ{1}{15}\stateZ{1}{16}\stateZ{1}{17}\stateZ{1}{18}\stateZ{1}{19}\stateZ{1}{20}\stateZ{1}{21}\stateZ{1}{22}\stateZ{1}{23}\stateZ{1}{24}\stateZ{1}{25}\stateZ{1}{26}\stateZ{1}{27}\stateZ{1}{28}\stateZ{1}{29}\stateZ{1}{30}\stateZ{1}{31}\stateZ{1}{32}\stateZ{1}{33}\stateZ{1}{34}\stateZ{1}{35}\stateZ{1}{36}\stateZ{1}{37}\stateZ{1}{38}\stateZ{1}{39}\stateZ{1}{40}\stateZ{1}{41}\stateZ{1}{42}\stateZ{1}{43}\stateZ{1}{44}\stateZ{1}{45}\stateZ{1}{46}\stateZ{1}{47}\stateZ{1}{48}\stateZ{1}{49}\stateZ{1}{50}\stateZ{1}{51}
\stateZ{2}{0}\stateZ{2}{1}\stateOO{2}{2}\stateOO{2}{3}\stateZ{2}{4}\stateZ{2}{5}\stateZ{2}{6}\stateZ{2}{7}\stateZ{2}{8}\stateZ{2}{9}\stateZ{2}{10}\stateZ{2}{11}\stateZ{2}{12}\stateZ{2}{13}\stateZ{2}{14}\stateZ{2}{15}\stateZ{2}{16}\stateZ{2}{17}\stateZ{2}{18}\stateZ{2}{19}\stateZ{2}{20}\stateZ{2}{21}\stateZ{2}{22}\stateZ{2}{23}\stateZ{2}{24}\stateZ{2}{25}\stateZ{2}{26}\stateZ{2}{27}\stateZ{2}{28}\stateZ{2}{29}\stateZ{2}{30}\stateZ{2}{31}\stateZ{2}{32}\stateZ{2}{33}\stateZ{2}{34}\stateZ{2}{35}\stateZ{2}{36}\stateZ{2}{37}\stateZ{2}{38}\stateZ{2}{39}\stateZ{2}{40}\stateZ{2}{41}\stateZ{2}{42}\stateZ{2}{43}\stateZ{2}{44}\stateZ{2}{45}\stateZ{2}{46}\stateZ{2}{47}\stateZ{2}{48}\stateZ{2}{49}\stateZ{2}{50}\stateZ{2}{51}
\stateZ{3}{0}\stateOO{3}{1}\stateZ{3}{2}\stateZO{3}{3}\stateZ{3}{4}\stateZ{3}{5}\stateZ{3}{6}\stateZ{3}{7}\stateZ{3}{8}\stateZ{3}{9}\stateZ{3}{10}\stateZ{3}{11}\stateZ{3}{12}\stateZ{3}{13}\stateZ{3}{14}\stateZ{3}{15}\stateZ{3}{16}\stateZ{3}{17}\stateZ{3}{18}\stateZ{3}{19}\stateZ{3}{20}\stateZ{3}{21}\stateZ{3}{22}\stateZ{3}{23}\stateZ{3}{24}\stateZ{3}{25}\stateZ{3}{26}\stateZ{3}{27}\stateZ{3}{28}\stateZ{3}{29}\stateZ{3}{30}\stateZ{3}{31}\stateZ{3}{32}\stateZ{3}{33}\stateZ{3}{34}\stateZ{3}{35}\stateZ{3}{36}\stateZ{3}{37}\stateZ{3}{38}\stateZ{3}{39}\stateZ{3}{40}\stateZ{3}{41}\stateZ{3}{42}\stateZ{3}{43}\stateZ{3}{44}\stateZ{3}{45}\stateZ{3}{46}\stateZ{3}{47}\stateZ{3}{48}\stateZ{3}{49}\stateZ{3}{50}\stateZ{3}{51}
\stateZ{4}{0}\stateZ{4}{1}\stateZ{4}{2}\stateZ{4}{3}\stateZO{4}{4}\stateZO{4}{5}\stateZ{4}{6}\stateZ{4}{7}\stateZ{4}{8}\stateZ{4}{9}\stateZ{4}{10}\stateZ{4}{11}\stateZ{4}{12}\stateZ{4}{13}\stateZ{4}{14}\stateZ{4}{15}\stateZ{4}{16}\stateZ{4}{17}\stateZ{4}{18}\stateZ{4}{19}\stateZ{4}{20}\stateZ{4}{21}\stateZ{4}{22}\stateZ{4}{23}\stateZ{4}{24}\stateZ{4}{25}\stateZ{4}{26}\stateZ{4}{27}\stateZ{4}{28}\stateZ{4}{29}\stateZ{4}{30}\stateZ{4}{31}\stateZ{4}{32}\stateZ{4}{33}\stateZ{4}{34}\stateZ{4}{35}\stateZ{4}{36}\stateZ{4}{37}\stateZ{4}{38}\stateZ{4}{39}\stateZ{4}{40}\stateZ{4}{41}\stateZ{4}{42}\stateZ{4}{43}\stateZ{4}{44}\stateZ{4}{45}\stateZ{4}{46}\stateZ{4}{47}\stateZ{4}{48}\stateZ{4}{49}\stateZ{4}{50}\stateZ{4}{51}
\stateZ{5}{0}\stateZ{5}{1}\stateZ{5}{2}\stateZO{5}{3}\stateZ{5}{4}\stateZO{5}{5}\stateZ{5}{6}\stateZO{5}{7}\stateZ{5}{8}\stateZ{5}{9}\stateZ{5}{10}\stateZ{5}{11}\stateZ{5}{12}\stateZ{5}{13}\stateZ{5}{14}\stateZ{5}{15}\stateZ{5}{16}\stateZ{5}{17}\stateZ{5}{18}\stateZ{5}{19}\stateZ{5}{20}\stateZ{5}{21}\stateZ{5}{22}\stateZ{5}{23}\stateZ{5}{24}\stateZ{5}{25}\stateZ{5}{26}\stateZ{5}{27}\stateZ{5}{28}\stateZ{5}{29}\stateZ{5}{30}\stateZ{5}{31}\stateZ{5}{32}\stateZ{5}{33}\stateZ{5}{34}\stateZ{5}{35}\stateZ{5}{36}\stateZ{5}{37}\stateZ{5}{38}\stateZ{5}{39}\stateZ{5}{40}\stateZ{5}{41}\stateZ{5}{42}\stateZ{5}{43}\stateZ{5}{44}\stateZ{5}{45}\stateZ{5}{46}\stateZ{5}{47}\stateZ{5}{48}\stateZ{5}{49}\stateZ{5}{50}\stateZ{5}{51}
\stateZ{6}{0}\stateZ{6}{1}\stateZ{6}{2}\stateZ{6}{3}\stateZ{6}{4}\stateZ{6}{5}\stateZO{6}{6}\stateZO{6}{7}\stateZO{6}{8}\stateZO{6}{9}\stateZ{6}{10}\stateZ{6}{11}\stateZ{6}{12}\stateZ{6}{13}\stateZ{6}{14}\stateZ{6}{15}\stateZ{6}{16}\stateZ{6}{17}\stateZ{6}{18}\stateZ{6}{19}\stateZ{6}{20}\stateZ{6}{21}\stateZ{6}{22}\stateZ{6}{23}\stateZ{6}{24}\stateZ{6}{25}\stateZ{6}{26}\stateZ{6}{27}\stateZ{6}{28}\stateZ{6}{29}\stateZ{6}{30}\stateZ{6}{31}\stateZ{6}{32}\stateZ{6}{33}\stateZ{6}{34}\stateZ{6}{35}\stateZ{6}{36}\stateZ{6}{37}\stateZ{6}{38}\stateZ{6}{39}\stateZ{6}{40}\stateZ{6}{41}\stateZ{6}{42}\stateZ{6}{43}\stateZ{6}{44}\stateZ{6}{45}\stateZ{6}{46}\stateZ{6}{47}\stateZ{6}{48}\stateZ{6}{49}\stateZ{6}{50}\stateZ{6}{51}
\stateZ{7}{0}\stateZ{7}{1}\stateZ{7}{2}\stateZO{7}{3}\stateZ{7}{4}\stateZO{7}{5}\stateZ{7}{6}\stateZO{7}{7}\stateZ{7}{8}\stateZ{7}{9}\stateZ{7}{10}\stateZO{7}{11}\stateZ{7}{12}\stateZ{7}{13}\stateZ{7}{14}\stateZ{7}{15}\stateZ{7}{16}\stateZ{7}{17}\stateZ{7}{18}\stateZ{7}{19}\stateZ{7}{20}\stateZ{7}{21}\stateZ{7}{22}\stateZ{7}{23}\stateZ{7}{24}\stateZ{7}{25}\stateZ{7}{26}\stateZ{7}{27}\stateZ{7}{28}\stateZ{7}{29}\stateZ{7}{30}\stateZ{7}{31}\stateZ{7}{32}\stateZ{7}{33}\stateZ{7}{34}\stateZ{7}{35}\stateZ{7}{36}\stateZ{7}{37}\stateZ{7}{38}\stateZ{7}{39}\stateZ{7}{40}\stateZ{7}{41}\stateZ{7}{42}\stateZ{7}{43}\stateZ{7}{44}\stateZ{7}{45}\stateZ{7}{46}\stateZ{7}{47}\stateZ{7}{48}\stateZ{7}{49}\stateZ{7}{50}\stateZ{7}{51}
\stateZ{8}{0}\stateZ{8}{1}\stateZ{8}{2}\stateZ{8}{3}\stateZO{8}{4}\stateZO{8}{5}\stateZ{8}{6}\stateZ{8}{7}\stateZO{8}{8}\stateZO{8}{9}\stateZ{8}{10}\stateZ{8}{11}\stateZO{8}{12}\stateZO{8}{13}\stateZ{8}{14}\stateZ{8}{15}\stateZ{8}{16}\stateZ{8}{17}\stateZ{8}{18}\stateZ{8}{19}\stateZ{8}{20}\stateZ{8}{21}\stateZ{8}{22}\stateZ{8}{23}\stateZ{8}{24}\stateZ{8}{25}\stateZ{8}{26}\stateZ{8}{27}\stateZ{8}{28}\stateZ{8}{29}\stateZ{8}{30}\stateZ{8}{31}\stateZ{8}{32}\stateZ{8}{33}\stateZ{8}{34}\stateZ{8}{35}\stateZ{8}{36}\stateZ{8}{37}\stateZ{8}{38}\stateZ{8}{39}\stateZ{8}{40}\stateZ{8}{41}\stateZ{8}{42}\stateZ{8}{43}\stateZ{8}{44}\stateZ{8}{45}\stateZ{8}{46}\stateZ{8}{47}\stateZ{8}{48}\stateZ{8}{49}\stateZ{8}{50}\stateZ{8}{51}
\stateZ{9}{0}\stateZ{9}{1}\stateZ{9}{2}\stateZO{9}{3}\stateZ{9}{4}\stateZ{9}{5}\stateZ{9}{6}\stateZO{9}{7}\stateZ{9}{8}\stateOO{9}{9}\stateZ{9}{10}\stateZO{9}{11}\stateZ{9}{12}\stateZO{9}{13}\stateZ{9}{14}\stateZO{9}{15}\stateZ{9}{16}\stateOO{9}{17}\stateZ{9}{18}\stateZ{9}{19}\stateZ{9}{20}\stateZ{9}{21}\stateZ{9}{22}\stateZ{9}{23}\stateZ{9}{24}\stateZ{9}{25}\stateZ{9}{26}\stateZ{9}{27}\stateZ{9}{28}\stateZ{9}{29}\stateZ{9}{30}\stateZ{9}{31}\stateZ{9}{32}\stateZ{9}{33}\stateZ{9}{34}\stateZ{9}{35}\stateZ{9}{36}\stateZ{9}{37}\stateZ{9}{38}\stateZ{9}{39}\stateZ{9}{40}\stateZ{9}{41}\stateZ{9}{42}\stateZ{9}{43}\stateZ{9}{44}\stateZ{9}{45}\stateZ{9}{46}\stateZ{9}{47}\stateZ{9}{48}\stateZ{9}{49}\stateZ{9}{50}\stateZ{9}{51}
\stateZ{10}{0}\stateZ{10}{1}\stateZ{10}{2}\stateZ{10}{3}\stateZ{10}{4}\stateZ{10}{5}\stateZ{10}{6}\stateZ{10}{7}\stateZ{10}{8}\stateZ{10}{9}\stateOO{10}{10}\stateOO{10}{11}\stateZO{10}{12}\stateZO{10}{13}\stateZO{10}{14}\stateZO{10}{15}\stateZO{10}{16}\stateZO{10}{17}\stateOO{10}{18}\stateOO{10}{19}\stateZ{10}{20}\stateZ{10}{21}\stateZ{10}{22}\stateZ{10}{23}\stateZ{10}{24}\stateZ{10}{25}\stateZ{10}{26}\stateZ{10}{27}\stateZ{10}{28}\stateZ{10}{29}\stateZ{10}{30}\stateZ{10}{31}\stateZ{10}{32}\stateZ{10}{33}\stateZ{10}{34}\stateZ{10}{35}\stateZ{10}{36}\stateZ{10}{37}\stateZ{10}{38}\stateZ{10}{39}\stateZ{10}{40}\stateZ{10}{41}\stateZ{10}{42}\stateZ{10}{43}\stateZ{10}{44}\stateZ{10}{45}\stateZ{10}{46}\stateZ{10}{47}\stateZ{10}{48}\stateZ{10}{49}\stateZ{10}{50}\stateZ{10}{51}
\stateZ{11}{0}\stateZ{11}{1}\stateZ{11}{2}\stateZO{11}{3}\stateZ{11}{4}\stateZ{11}{5}\stateZ{11}{6}\stateZO{11}{7}\stateZ{11}{8}\stateOO{11}{9}\stateZ{11}{10}\stateZO{11}{11}\stateZ{11}{12}\stateZO{11}{13}\stateZ{11}{14}\stateZO{11}{15}\stateZ{11}{16}\stateOO{11}{17}\stateZ{11}{18}\stateZ{11}{19}\stateZ{11}{20}\stateZ{11}{21}\stateZ{11}{22}\stateZ{11}{23}\stateZ{11}{24}\stateZ{11}{25}\stateZ{11}{26}\stateZ{11}{27}\stateZ{11}{28}\stateZ{11}{29}\stateZ{11}{30}\stateZ{11}{31}\stateZ{11}{32}\stateZ{11}{33}\stateZ{11}{34}\stateZ{11}{35}\stateZ{11}{36}\stateZ{11}{37}\stateZ{11}{38}\stateZ{11}{39}\stateZ{11}{40}\stateZ{11}{41}\stateZ{11}{42}\stateZ{11}{43}\stateZ{11}{44}\stateZ{11}{45}\stateZ{11}{46}\stateZ{11}{47}\stateZ{11}{48}\stateZ{11}{49}\stateZ{11}{50}\stateZ{11}{51}
\stateZ{12}{0}\stateZ{12}{1}\stateZ{12}{2}\stateZ{12}{3}\stateZO{12}{4}\stateZO{12}{5}\stateZ{12}{6}\stateZ{12}{7}\stateZO{12}{8}\stateZO{12}{9}\stateZ{12}{10}\stateZ{12}{11}\stateZO{12}{12}\stateZO{12}{13}\stateZ{12}{14}\stateZ{12}{15}\stateZ{12}{16}\stateZ{12}{17}\stateZ{12}{18}\stateZ{12}{19}\stateZ{12}{20}\stateZ{12}{21}\stateZ{12}{22}\stateZ{12}{23}\stateZ{12}{24}\stateZ{12}{25}\stateZ{12}{26}\stateZ{12}{27}\stateZ{12}{28}\stateZ{12}{29}\stateZ{12}{30}\stateZ{12}{31}\stateZ{12}{32}\stateZ{12}{33}\stateZ{12}{34}\stateZ{12}{35}\stateZ{12}{36}\stateZ{12}{37}\stateZ{12}{38}\stateZ{12}{39}\stateZ{12}{40}\stateZ{12}{41}\stateZ{12}{42}\stateZ{12}{43}\stateZ{12}{44}\stateZ{12}{45}\stateZ{12}{46}\stateZ{12}{47}\stateZ{12}{48}\stateZ{12}{49}\stateZ{12}{50}\stateZ{12}{51}
\stateZ{13}{0}\stateZ{13}{1}\stateZ{13}{2}\stateZO{13}{3}\stateZ{13}{4}\stateZO{13}{5}\stateZ{13}{6}\stateZO{13}{7}\stateZ{13}{8}\stateZ{13}{9}\stateZ{13}{10}\stateZO{13}{11}\stateZ{13}{12}\stateZO{13}{13}\stateZ{13}{14}\stateZO{13}{15}\stateZ{13}{16}\stateZ{13}{17}\stateZ{13}{18}\stateZ{13}{19}\stateZ{13}{20}\stateZ{13}{21}\stateZ{13}{22}\stateZ{13}{23}\stateZ{13}{24}\stateZ{13}{25}\stateZ{13}{26}\stateZ{13}{27}\stateZ{13}{28}\stateZ{13}{29}\stateZ{13}{30}\stateZ{13}{31}\stateZ{13}{32}\stateZ{13}{33}\stateZ{13}{34}\stateZ{13}{35}\stateZ{13}{36}\stateZ{13}{37}\stateZ{13}{38}\stateZ{13}{39}\stateZ{13}{40}\stateZ{13}{41}\stateZ{13}{42}\stateZ{13}{43}\stateZ{13}{44}\stateZ{13}{45}\stateZ{13}{46}\stateZ{13}{47}\stateZ{13}{48}\stateZ{13}{49}\stateZ{13}{50}\stateZ{13}{51}
\stateZ{14}{0}\stateZ{14}{1}\stateZ{14}{2}\stateZ{14}{3}\stateZ{14}{4}\stateZ{14}{5}\stateZO{14}{6}\stateZO{14}{7}\stateZO{14}{8}\stateZO{14}{9}\stateZ{14}{10}\stateZ{14}{11}\stateZ{14}{12}\stateZ{14}{13}\stateZO{14}{14}\stateZO{14}{15}\stateZO{14}{16}\stateZO{14}{17}\stateZ{14}{18}\stateZ{14}{19}\stateZ{14}{20}\stateZ{14}{21}\stateZ{14}{22}\stateZ{14}{23}\stateZ{14}{24}\stateZ{14}{25}\stateZ{14}{26}\stateZ{14}{27}\stateZ{14}{28}\stateZ{14}{29}\stateZ{14}{30}\stateZ{14}{31}\stateZ{14}{32}\stateZ{14}{33}\stateZ{14}{34}\stateZ{14}{35}\stateZ{14}{36}\stateZ{14}{37}\stateZ{14}{38}\stateZ{14}{39}\stateZ{14}{40}\stateZ{14}{41}\stateZ{14}{42}\stateZ{14}{43}\stateZ{14}{44}\stateZ{14}{45}\stateZ{14}{46}\stateZ{14}{47}\stateZ{14}{48}\stateZ{14}{49}\stateZ{14}{50}\stateZ{14}{51}
\stateZ{15}{0}\stateZ{15}{1}\stateZ{15}{2}\stateZO{15}{3}\stateZ{15}{4}\stateZO{15}{5}\stateZ{15}{6}\stateZO{15}{7}\stateZ{15}{8}\stateZ{15}{9}\stateZ{15}{10}\stateZO{15}{11}\stateZ{15}{12}\stateZO{15}{13}\stateZ{15}{14}\stateZO{15}{15}\stateZ{15}{16}\stateZ{15}{17}\stateZ{15}{18}\stateZO{15}{19}\stateZ{15}{20}\stateZ{15}{21}\stateZ{15}{22}\stateZ{15}{23}\stateZ{15}{24}\stateZ{15}{25}\stateZ{15}{26}\stateZ{15}{27}\stateZ{15}{28}\stateZ{15}{29}\stateZ{15}{30}\stateZ{15}{31}\stateZ{15}{32}\stateZ{15}{33}\stateZ{15}{34}\stateZ{15}{35}\stateZ{15}{36}\stateZ{15}{37}\stateZ{15}{38}\stateZ{15}{39}\stateZ{15}{40}\stateZ{15}{41}\stateZ{15}{42}\stateZ{15}{43}\stateZ{15}{44}\stateZ{15}{45}\stateZ{15}{46}\stateZ{15}{47}\stateZ{15}{48}\stateZ{15}{49}\stateZ{15}{50}\stateZ{15}{51}
\stateZ{16}{0}\stateZ{16}{1}\stateZ{16}{2}\stateZ{16}{3}\stateZO{16}{4}\stateZO{16}{5}\stateZ{16}{6}\stateZ{16}{7}\stateZ{16}{8}\stateZ{16}{9}\stateZ{16}{10}\stateZ{16}{11}\stateZO{16}{12}\stateZO{16}{13}\stateZ{16}{14}\stateZ{16}{15}\stateZO{16}{16}\stateZO{16}{17}\stateZ{16}{18}\stateZ{16}{19}\stateZO{16}{20}\stateZO{16}{21}\stateZ{16}{22}\stateZ{16}{23}\stateZ{16}{24}\stateZ{16}{25}\stateZ{16}{26}\stateZ{16}{27}\stateZ{16}{28}\stateZ{16}{29}\stateZ{16}{30}\stateZ{16}{31}\stateZ{16}{32}\stateZ{16}{33}\stateZ{16}{34}\stateZ{16}{35}\stateZ{16}{36}\stateZ{16}{37}\stateZ{16}{38}\stateZ{16}{39}\stateZ{16}{40}\stateZ{16}{41}\stateZ{16}{42}\stateZ{16}{43}\stateZ{16}{44}\stateZ{16}{45}\stateZ{16}{46}\stateZ{16}{47}\stateZ{16}{48}\stateZ{16}{49}\stateZ{16}{50}\stateZ{16}{51}
\stateZ{17}{0}\stateOO{17}{1}\stateZ{17}{2}\stateZO{17}{3}\stateZ{17}{4}\stateZ{17}{5}\stateZ{17}{6}\stateZ{17}{7}\stateZ{17}{8}\stateZ{17}{9}\stateZ{17}{10}\stateZO{17}{11}\stateZ{17}{12}\stateZ{17}{13}\stateZ{17}{14}\stateZO{17}{15}\stateZ{17}{16}\stateOO{17}{17}\stateZ{17}{18}\stateZO{17}{19}\stateZ{17}{20}\stateZO{17}{21}\stateZ{17}{22}\stateZO{17}{23}\stateZ{17}{24}\stateZ{17}{25}\stateZ{17}{26}\stateZ{17}{27}\stateZ{17}{28}\stateZ{17}{29}\stateZ{17}{30}\stateZ{17}{31}\stateZ{17}{32}\stateOO{17}{33}\stateZ{17}{34}\stateZ{17}{35}\stateZ{17}{36}\stateZ{17}{37}\stateZ{17}{38}\stateZ{17}{39}\stateZ{17}{40}\stateZ{17}{41}\stateZ{17}{42}\stateZ{17}{43}\stateZ{17}{44}\stateZ{17}{45}\stateZ{17}{46}\stateZ{17}{47}\stateZ{17}{48}\stateZ{17}{49}\stateZ{17}{50}\stateZ{17}{51}
\stateZ{18}{0}\stateZ{18}{1}\stateOO{18}{2}\stateOO{18}{3}\stateZ{18}{4}\stateZ{18}{5}\stateZ{18}{6}\stateZ{18}{7}\stateZ{18}{8}\stateZ{18}{9}\stateZ{18}{10}\stateZ{18}{11}\stateZ{18}{12}\stateZ{18}{13}\stateZ{18}{14}\stateZ{18}{15}\stateZ{18}{16}\stateZ{18}{17}\stateOO{18}{18}\stateOO{18}{19}\stateZO{18}{20}\stateZO{18}{21}\stateZO{18}{22}\stateZO{18}{23}\stateZO{18}{24}\stateZO{18}{25}\stateZ{18}{26}\stateZ{18}{27}\stateZ{18}{28}\stateZ{18}{29}\stateZ{18}{30}\stateZ{18}{31}\stateZ{18}{32}\stateZ{18}{33}\stateOO{18}{34}\stateOO{18}{35}\stateZ{18}{36}\stateZ{18}{37}\stateZ{18}{38}\stateZ{18}{39}\stateZ{18}{40}\stateZ{18}{41}\stateZ{18}{42}\stateZ{18}{43}\stateZ{18}{44}\stateZ{18}{45}\stateZ{18}{46}\stateZ{18}{47}\stateZ{18}{48}\stateZ{18}{49}\stateZ{18}{50}\stateZ{18}{51}
\stateZ{19}{0}\stateOO{19}{1}\stateZ{19}{2}\stateZ{19}{3}\stateZ{19}{4}\stateZ{19}{5}\stateZ{19}{6}\stateZ{19}{7}\stateZ{19}{8}\stateZ{19}{9}\stateZ{19}{10}\stateZO{19}{11}\stateZ{19}{12}\stateZ{19}{13}\stateZ{19}{14}\stateZO{19}{15}\stateZ{19}{16}\stateOO{19}{17}\stateZ{19}{18}\stateZO{19}{19}\stateZ{19}{20}\stateZO{19}{21}\stateZ{19}{22}\stateZO{19}{23}\stateZ{19}{24}\stateZ{19}{25}\stateZ{19}{26}\stateZO{19}{27}\stateZ{19}{28}\stateZ{19}{29}\stateZ{19}{30}\stateZ{19}{31}\stateZ{19}{32}\stateOO{19}{33}\stateZ{19}{34}\stateZO{19}{35}\stateZ{19}{36}\stateZ{19}{37}\stateZ{19}{38}\stateZ{19}{39}\stateZ{19}{40}\stateZ{19}{41}\stateZ{19}{42}\stateZ{19}{43}\stateZ{19}{44}\stateZ{19}{45}\stateZ{19}{46}\stateZ{19}{47}\stateZ{19}{48}\stateZ{19}{49}\stateZ{19}{50}\stateZ{19}{51}
\stateZ{20}{0}\stateZ{20}{1}\stateZ{20}{2}\stateZ{20}{3}\stateZ{20}{4}\stateZ{20}{5}\stateZ{20}{6}\stateZ{20}{7}\stateZ{20}{8}\stateZ{20}{9}\stateZ{20}{10}\stateZ{20}{11}\stateZO{20}{12}\stateZO{20}{13}\stateZ{20}{14}\stateZ{20}{15}\stateZO{20}{16}\stateZO{20}{17}\stateZ{20}{18}\stateZ{20}{19}\stateZO{20}{20}\stateZO{20}{21}\stateZ{20}{22}\stateZ{20}{23}\stateZ{20}{24}\stateZ{20}{25}\stateZ{20}{26}\stateZ{20}{27}\stateZO{20}{28}\stateZO{20}{29}\stateZ{20}{30}\stateZ{20}{31}\stateZ{20}{32}\stateZ{20}{33}\stateZ{20}{34}\stateZ{20}{35}\stateZO{20}{36}\stateZO{20}{37}\stateZ{20}{38}\stateZ{20}{39}\stateZ{20}{40}\stateZ{20}{41}\stateZ{20}{42}\stateZ{20}{43}\stateZ{20}{44}\stateZ{20}{45}\stateZ{20}{46}\stateZ{20}{47}\stateZ{20}{48}\stateZ{20}{49}\stateZ{20}{50}\stateZ{20}{51}
\stateZ{21}{0}\stateZ{21}{1}\stateZ{21}{2}\stateZ{21}{3}\stateZ{21}{4}\stateZ{21}{5}\stateZ{21}{6}\stateZ{21}{7}\stateZ{21}{8}\stateZ{21}{9}\stateZ{21}{10}\stateZO{21}{11}\stateZ{21}{12}\stateZO{21}{13}\stateZ{21}{14}\stateZO{21}{15}\stateZ{21}{16}\stateZ{21}{17}\stateZ{21}{18}\stateZO{21}{19}\stateZ{21}{20}\stateZO{21}{21}\stateZ{21}{22}\stateZO{21}{23}\stateZ{21}{24}\stateZ{21}{25}\stateZ{21}{26}\stateZO{21}{27}\stateZ{21}{28}\stateZO{21}{29}\stateZ{21}{30}\stateZO{21}{31}\stateZ{21}{32}\stateZ{21}{33}\stateZ{21}{34}\stateZO{21}{35}\stateZ{21}{36}\stateZO{21}{37}\stateZ{21}{38}\stateZO{21}{39}\stateZ{21}{40}\stateZ{21}{41}\stateZ{21}{42}\stateZ{21}{43}\stateZ{21}{44}\stateZ{21}{45}\stateZ{21}{46}\stateZ{21}{47}\stateZ{21}{48}\stateZ{21}{49}\stateZ{21}{50}\stateZ{21}{51}
\stateZ{22}{0}\stateZ{22}{1}\stateZ{22}{2}\stateZ{22}{3}\stateZ{22}{4}\stateZ{22}{5}\stateZ{22}{6}\stateZ{22}{7}\stateZ{22}{8}\stateZ{22}{9}\stateZ{22}{10}\stateZ{22}{11}\stateZ{22}{12}\stateZ{22}{13}\stateZO{22}{14}\stateZO{22}{15}\stateZO{22}{16}\stateZO{22}{17}\stateZ{22}{18}\stateZ{22}{19}\stateZ{22}{20}\stateZ{22}{21}\stateZO{22}{22}\stateZO{22}{23}\stateZO{22}{24}\stateZO{22}{25}\stateZ{22}{26}\stateZ{22}{27}\stateZ{22}{28}\stateZ{22}{29}\stateZO{22}{30}\stateZO{22}{31}\stateZO{22}{32}\stateZO{22}{33}\stateZ{22}{34}\stateZ{22}{35}\stateZ{22}{36}\stateZ{22}{37}\stateZO{22}{38}\stateZO{22}{39}\stateZO{22}{40}\stateZO{22}{41}\stateZ{22}{42}\stateZ{22}{43}\stateZ{22}{44}\stateZ{22}{45}\stateZ{22}{46}\stateZ{22}{47}\stateZ{22}{48}\stateZ{22}{49}\stateZ{22}{50}\stateZ{22}{51}
\stateZ{23}{0}\stateZ{23}{1}\stateZ{23}{2}\stateZ{23}{3}\stateZ{23}{4}\stateZ{23}{5}\stateZ{23}{6}\stateZ{23}{7}\stateZ{23}{8}\stateZ{23}{9}\stateZ{23}{10}\stateZO{23}{11}\stateZ{23}{12}\stateZO{23}{13}\stateZ{23}{14}\stateZO{23}{15}\stateZ{23}{16}\stateZ{23}{17}\stateZ{23}{18}\stateZO{23}{19}\stateZ{23}{20}\stateZO{23}{21}\stateZ{23}{22}\stateZO{23}{23}\stateZ{23}{24}\stateZ{23}{25}\stateZ{23}{26}\stateZO{23}{27}\stateZ{23}{28}\stateZO{23}{29}\stateZ{23}{30}\stateZO{23}{31}\stateZ{23}{32}\stateZ{23}{33}\stateZ{23}{34}\stateZO{23}{35}\stateZ{23}{36}\stateZO{23}{37}\stateZ{23}{38}\stateZO{23}{39}\stateZ{23}{40}\stateZ{23}{41}\stateZ{23}{42}\stateZO{23}{43}\stateZ{23}{44}\stateZ{23}{45}\stateZ{23}{46}\stateZ{23}{47}\stateZ{23}{48}\stateZ{23}{49}\stateZ{23}{50}\stateZ{23}{51}
\stateZ{24}{0}\stateZ{24}{1}\stateZ{24}{2}\stateZ{24}{3}\stateZ{24}{4}\stateZ{24}{5}\stateZ{24}{6}\stateZ{24}{7}\stateZ{24}{8}\stateZ{24}{9}\stateZ{24}{10}\stateZ{24}{11}\stateZO{24}{12}\stateZO{24}{13}\stateZ{24}{14}\stateZ{24}{15}\stateZ{24}{16}\stateZ{24}{17}\stateZ{24}{18}\stateZ{24}{19}\stateZO{24}{20}\stateZO{24}{21}\stateZ{24}{22}\stateZ{24}{23}\stateZO{24}{24}\stateZO{24}{25}\stateZ{24}{26}\stateZ{24}{27}\stateZO{24}{28}\stateZO{24}{29}\stateZ{24}{30}\stateZ{24}{31}\stateZO{24}{32}\stateZO{24}{33}\stateZ{24}{34}\stateZ{24}{35}\stateZO{24}{36}\stateZO{24}{37}\stateZ{24}{38}\stateZ{24}{39}\stateZO{24}{40}\stateZO{24}{41}\stateZ{24}{42}\stateZ{24}{43}\stateZO{24}{44}\stateZO{24}{45}\stateZ{24}{46}\stateZ{24}{47}\stateZ{24}{48}\stateZ{24}{49}\stateZ{24}{50}\stateZ{24}{51}
\stateZ{25}{0}\stateZ{25}{1}\stateZ{25}{2}\stateZ{25}{3}\stateZ{25}{4}\stateZ{25}{5}\stateZ{25}{6}\stateZ{25}{7}\stateZ{25}{8}\stateZ{25}{9}\stateZ{25}{10}\stateZO{25}{11}\stateZ{25}{12}\stateZ{25}{13}\stateZ{25}{14}\stateZ{25}{15}\stateZ{25}{16}\stateZ{25}{17}\stateZ{25}{18}\stateZO{25}{19}\stateZ{25}{20}\stateZ{25}{21}\stateZ{25}{22}\stateZO{25}{23}\stateZ{25}{24}\stateOO{25}{25}\stateZ{25}{26}\stateZO{25}{27}\stateZ{25}{28}\stateZO{25}{29}\stateZ{25}{30}\stateZO{25}{31}\stateZ{25}{32}\stateOO{25}{33}\stateZ{25}{34}\stateZO{25}{35}\stateZ{25}{36}\stateZO{25}{37}\stateZ{25}{38}\stateZO{25}{39}\stateZ{25}{40}\stateOO{25}{41}\stateZ{25}{42}\stateZO{25}{43}\stateZ{25}{44}\stateZO{25}{45}\stateZ{25}{46}\stateZO{25}{47}\stateZ{25}{48}\stateOO{25}{49}\stateZ{25}{50}\stateZ{25}{51}
\stateZ{26}{0}\stateZ{26}{1}\stateZ{26}{2}\stateZ{26}{3}\stateZ{26}{4}\stateZ{26}{5}\stateZ{26}{6}\stateZ{26}{7}\stateZ{26}{8}\stateZ{26}{9}\stateZ{26}{10}\stateZ{26}{11}\stateZ{26}{12}\stateZ{26}{13}\stateZ{26}{14}\stateZ{26}{15}\stateZ{26}{16}\stateZ{26}{17}\stateZ{26}{18}\stateZ{26}{19}\stateZ{26}{20}\stateZ{26}{21}\stateZ{26}{22}\stateZ{26}{23}\stateZ{26}{24}\stateZ{26}{25}\stateOO{26}{26}\stateOO{26}{27}\stateZO{26}{28}\stateZO{26}{29}\stateZO{26}{30}\stateZO{26}{31}\stateZO{26}{32}\stateZO{26}{33}\stateOO{26}{34}\stateOO{26}{35}\stateZO{26}{36}\stateZO{26}{37}\stateZO{26}{38}\stateZO{26}{39}\stateZO{26}{40}\stateZO{26}{41}\stateOO{26}{42}\stateOO{26}{43}\stateZO{26}{44}\stateZO{26}{45}\stateZO{26}{46}\stateZO{26}{47}\stateZO{26}{48}\stateZO{26}{49}\stateOO{26}{50}\stateOO{26}{51}
\stateZ{27}{0}\stateZ{27}{1}\stateZ{27}{2}\stateZ{27}{3}\stateZ{27}{4}\stateZ{27}{5}\stateZ{27}{6}\stateZ{27}{7}\stateZ{27}{8}\stateZ{27}{9}\stateZ{27}{10}\stateZO{27}{11}\stateZ{27}{12}\stateZ{27}{13}\stateZ{27}{14}\stateZ{27}{15}\stateZ{27}{16}\stateZ{27}{17}\stateZ{27}{18}\stateZO{27}{19}\stateZ{27}{20}\stateZ{27}{21}\stateZ{27}{22}\stateZO{27}{23}\stateZ{27}{24}\stateOO{27}{25}\stateZ{27}{26}\stateZO{27}{27}\stateZ{27}{28}\stateZO{27}{29}\stateZ{27}{30}\stateZO{27}{31}\stateZ{27}{32}\stateOO{27}{33}\stateZ{27}{34}\stateZO{27}{35}\stateZ{27}{36}\stateZO{27}{37}\stateZ{27}{38}\stateZO{27}{39}\stateZ{27}{40}\stateOO{27}{41}\stateZ{27}{42}\stateZO{27}{43}\stateZ{27}{44}\stateZO{27}{45}\stateZ{27}{46}\stateZO{27}{47}\stateZ{27}{48}\stateOO{27}{49}\stateZ{27}{50}\stateZ{27}{51}
\stateZ{28}{0}\stateZ{28}{1}\stateZ{28}{2}\stateZ{28}{3}\stateZ{28}{4}\stateZ{28}{5}\stateZ{28}{6}\stateZ{28}{7}\stateZ{28}{8}\stateZ{28}{9}\stateZ{28}{10}\stateZ{28}{11}\stateZO{28}{12}\stateZO{28}{13}\stateZ{28}{14}\stateZ{28}{15}\stateZ{28}{16}\stateZ{28}{17}\stateZ{28}{18}\stateZ{28}{19}\stateZO{28}{20}\stateZO{28}{21}\stateZ{28}{22}\stateZ{28}{23}\stateZO{28}{24}\stateZO{28}{25}\stateZ{28}{26}\stateZ{28}{27}\stateZO{28}{28}\stateZO{28}{29}\stateZ{28}{30}\stateZ{28}{31}\stateZO{28}{32}\stateZO{28}{33}\stateZ{28}{34}\stateZ{28}{35}\stateZO{28}{36}\stateZO{28}{37}\stateZ{28}{38}\stateZ{28}{39}\stateZO{28}{40}\stateZO{28}{41}\stateZ{28}{42}\stateZ{28}{43}\stateZO{28}{44}\stateZO{28}{45}\stateZ{28}{46}\stateZ{28}{47}\stateZ{28}{48}\stateZ{28}{49}\stateZ{28}{50}\stateZ{28}{51}
\stateZ{29}{0}\stateZ{29}{1}\stateZ{29}{2}\stateZ{29}{3}\stateZ{29}{4}\stateZ{29}{5}\stateZ{29}{6}\stateZ{29}{7}\stateZ{29}{8}\stateZ{29}{9}\stateZ{29}{10}\stateZO{29}{11}\stateZ{29}{12}\stateZO{29}{13}\stateZ{29}{14}\stateZO{29}{15}\stateZ{29}{16}\stateZ{29}{17}\stateZ{29}{18}\stateZO{29}{19}\stateZ{29}{20}\stateZO{29}{21}\stateZ{29}{22}\stateZO{29}{23}\stateZ{29}{24}\stateZ{29}{25}\stateZ{29}{26}\stateZO{29}{27}\stateZ{29}{28}\stateZO{29}{29}\stateZ{29}{30}\stateZO{29}{31}\stateZ{29}{32}\stateZ{29}{33}\stateZ{29}{34}\stateZO{29}{35}\stateZ{29}{36}\stateZO{29}{37}\stateZ{29}{38}\stateZO{29}{39}\stateZ{29}{40}\stateZ{29}{41}\stateZ{29}{42}\stateZO{29}{43}\stateZ{29}{44}\stateZ{29}{45}\stateZ{29}{46}\stateZ{29}{47}\stateZ{29}{48}\stateZ{29}{49}\stateZ{29}{50}\stateZ{29}{51}
\stateZ{30}{0}\stateZ{30}{1}\stateZ{30}{2}\stateZ{30}{3}\stateZ{30}{4}\stateZ{30}{5}\stateZ{30}{6}\stateZ{30}{7}\stateZ{30}{8}\stateZ{30}{9}\stateZ{30}{10}\stateZ{30}{11}\stateZ{30}{12}\stateZ{30}{13}\stateZO{30}{14}\stateZO{30}{15}\stateZO{30}{16}\stateZO{30}{17}\stateZ{30}{18}\stateZ{30}{19}\stateZ{30}{20}\stateZ{30}{21}\stateZO{30}{22}\stateZO{30}{23}\stateZO{30}{24}\stateZO{30}{25}\stateZ{30}{26}\stateZ{30}{27}\stateZ{30}{28}\stateZ{30}{29}\stateZO{30}{30}\stateZO{30}{31}\stateZO{30}{32}\stateZO{30}{33}\stateZ{30}{34}\stateZ{30}{35}\stateZ{30}{36}\stateZ{30}{37}\stateZO{30}{38}\stateZO{30}{39}\stateZO{30}{40}\stateZO{30}{41}\stateZ{30}{42}\stateZ{30}{43}\stateZ{30}{44}\stateZ{30}{45}\stateZ{30}{46}\stateZ{30}{47}\stateZ{30}{48}\stateZ{30}{49}\stateZ{30}{50}\stateZ{30}{51}
\stateZ{31}{0}\stateZ{31}{1}\stateZ{31}{2}\stateZ{31}{3}\stateZ{31}{4}\stateZ{31}{5}\stateZ{31}{6}\stateZ{31}{7}\stateZ{31}{8}\stateZ{31}{9}\stateZ{31}{10}\stateZO{31}{11}\stateZ{31}{12}\stateZO{31}{13}\stateZ{31}{14}\stateZO{31}{15}\stateZ{31}{16}\stateZ{31}{17}\stateZ{31}{18}\stateZO{31}{19}\stateZ{31}{20}\stateZO{31}{21}\stateZ{31}{22}\stateZO{31}{23}\stateZ{31}{24}\stateZ{31}{25}\stateZ{31}{26}\stateZO{31}{27}\stateZ{31}{28}\stateZO{31}{29}\stateZ{31}{30}\stateZO{31}{31}\stateZ{31}{32}\stateZ{31}{33}\stateZ{31}{34}\stateZO{31}{35}\stateZ{31}{36}\stateZO{31}{37}\stateZ{31}{38}\stateZO{31}{39}\stateZ{31}{40}\stateZ{31}{41}\stateZ{31}{42}\stateZ{31}{43}\stateZ{31}{44}\stateZ{31}{45}\stateZ{31}{46}\stateZ{31}{47}\stateZ{31}{48}\stateZ{31}{49}\stateZ{31}{50}\stateZ{31}{51}
\stateZ{32}{0}\stateZ{32}{1}\stateZ{32}{2}\stateZ{32}{3}\stateZ{32}{4}\stateZ{32}{5}\stateZ{32}{6}\stateZ{32}{7}\stateZ{32}{8}\stateZ{32}{9}\stateZ{32}{10}\stateZ{32}{11}\stateZO{32}{12}\stateZO{32}{13}\stateZ{32}{14}\stateZ{32}{15}\stateZO{32}{16}\stateZO{32}{17}\stateZ{32}{18}\stateZ{32}{19}\stateZO{32}{20}\stateZO{32}{21}\stateZ{32}{22}\stateZ{32}{23}\stateZ{32}{24}\stateZ{32}{25}\stateZ{32}{26}\stateZ{32}{27}\stateZO{32}{28}\stateZO{32}{29}\stateZ{32}{30}\stateZ{32}{31}\stateZ{32}{32}\stateZ{32}{33}\stateZ{32}{34}\stateZ{32}{35}\stateZO{32}{36}\stateZO{32}{37}\stateZ{32}{38}\stateZ{32}{39}\stateZ{32}{40}\stateZ{32}{41}\stateZ{32}{42}\stateZ{32}{43}\stateZ{32}{44}\stateZ{32}{45}\stateZ{32}{46}\stateZ{32}{47}\stateZ{32}{48}\stateZ{32}{49}\stateZ{32}{50}\stateZ{32}{51}
\stateZ{33}{0}\stateOO{33}{1}\stateZ{33}{2}\stateZ{33}{3}\stateZ{33}{4}\stateZ{33}{5}\stateZ{33}{6}\stateZ{33}{7}\stateZ{33}{8}\stateZ{33}{9}\stateZ{33}{10}\stateZO{33}{11}\stateZ{33}{12}\stateZ{33}{13}\stateZ{33}{14}\stateZO{33}{15}\stateZ{33}{16}\stateOO{33}{17}\stateZ{33}{18}\stateZO{33}{19}\stateZ{33}{20}\stateZO{33}{21}\stateZ{33}{22}\stateZO{33}{23}\stateZ{33}{24}\stateZ{33}{25}\stateZ{33}{26}\stateZO{33}{27}\stateZ{33}{28}\stateZ{33}{29}\stateZ{33}{30}\stateZ{33}{31}\stateZ{33}{32}\stateOO{33}{33}\stateZ{33}{34}\stateZO{33}{35}\stateZ{33}{36}\stateZ{33}{37}\stateZ{33}{38}\stateZ{33}{39}\stateZ{33}{40}\stateZ{33}{41}\stateZ{33}{42}\stateZ{33}{43}\stateZ{33}{44}\stateZ{33}{45}\stateZ{33}{46}\stateZ{33}{47}\stateZ{33}{48}\stateZ{33}{49}\stateZ{33}{50}\stateZ{33}{51}
\stateZ{34}{0}\stateZ{34}{1}\stateOO{34}{2}\stateOO{34}{3}\stateZ{34}{4}\stateZ{34}{5}\stateZ{34}{6}\stateZ{34}{7}\stateZ{34}{8}\stateZ{34}{9}\stateZ{34}{10}\stateZ{34}{11}\stateZ{34}{12}\stateZ{34}{13}\stateZ{34}{14}\stateZ{34}{15}\stateZ{34}{16}\stateZ{34}{17}\stateOO{34}{18}\stateOO{34}{19}\stateZO{34}{20}\stateZO{34}{21}\stateZO{34}{22}\stateZO{34}{23}\stateZO{34}{24}\stateZO{34}{25}\stateZ{34}{26}\stateZ{34}{27}\stateZ{34}{28}\stateZ{34}{29}\stateZ{34}{30}\stateZ{34}{31}\stateZ{34}{32}\stateZ{34}{33}\stateOO{34}{34}\stateOO{34}{35}\stateZ{34}{36}\stateZ{34}{37}\stateZ{34}{38}\stateZ{34}{39}\stateZ{34}{40}\stateZ{34}{41}\stateZ{34}{42}\stateZ{34}{43}\stateZ{34}{44}\stateZ{34}{45}\stateZ{34}{46}\stateZ{34}{47}\stateZ{34}{48}\stateZ{34}{49}\stateZ{34}{50}\stateZ{34}{51}
\stateZ{35}{0}\stateOO{35}{1}\stateZ{35}{2}\stateZO{35}{3}\stateZ{35}{4}\stateZ{35}{5}\stateZ{35}{6}\stateZ{35}{7}\stateZ{35}{8}\stateZ{35}{9}\stateZ{35}{10}\stateZO{35}{11}\stateZ{35}{12}\stateZ{35}{13}\stateZ{35}{14}\stateZO{35}{15}\stateZ{35}{16}\stateOO{35}{17}\stateZ{35}{18}\stateZO{35}{19}\stateZ{35}{20}\stateZO{35}{21}\stateZ{35}{22}\stateZO{35}{23}\stateZ{35}{24}\stateZ{35}{25}\stateZ{35}{26}\stateZ{35}{27}\stateZ{35}{28}\stateZ{35}{29}\stateZ{35}{30}\stateZ{35}{31}\stateZ{35}{32}\stateOO{35}{33}\stateZ{35}{34}\stateZ{35}{35}\stateZ{35}{36}\stateZ{35}{37}\stateZ{35}{38}\stateZ{35}{39}\stateZ{35}{40}\stateZ{35}{41}\stateZ{35}{42}\stateZ{35}{43}\stateZ{35}{44}\stateZ{35}{45}\stateZ{35}{46}\stateZ{35}{47}\stateZ{35}{48}\stateZ{35}{49}\stateZ{35}{50}\stateZ{35}{51}
\stateZ{36}{0}\stateZ{36}{1}\stateZ{36}{2}\stateZ{36}{3}\stateZO{36}{4}\stateZO{36}{5}\stateZ{36}{6}\stateZ{36}{7}\stateZ{36}{8}\stateZ{36}{9}\stateZ{36}{10}\stateZ{36}{11}\stateZO{36}{12}\stateZO{36}{13}\stateZ{36}{14}\stateZ{36}{15}\stateZO{36}{16}\stateZO{36}{17}\stateZ{36}{18}\stateZ{36}{19}\stateZO{36}{20}\stateZO{36}{21}\stateZ{36}{22}\stateZ{36}{23}\stateZ{36}{24}\stateZ{36}{25}\stateZ{36}{26}\stateZ{36}{27}\stateZ{36}{28}\stateZ{36}{29}\stateZ{36}{30}\stateZ{36}{31}\stateZ{36}{32}\stateZ{36}{33}\stateZ{36}{34}\stateZ{36}{35}\stateZ{36}{36}\stateZ{36}{37}\stateZ{36}{38}\stateZ{36}{39}\stateZ{36}{40}\stateZ{36}{41}\stateZ{36}{42}\stateZ{36}{43}\stateZ{36}{44}\stateZ{36}{45}\stateZ{36}{46}\stateZ{36}{47}\stateZ{36}{48}\stateZ{36}{49}\stateZ{36}{50}\stateZ{36}{51}
\stateZ{37}{0}\stateZ{37}{1}\stateZ{37}{2}\stateZO{37}{3}\stateZ{37}{4}\stateZO{37}{5}\stateZ{37}{6}\stateZO{37}{7}\stateZ{37}{8}\stateZ{37}{9}\stateZ{37}{10}\stateZO{37}{11}\stateZ{37}{12}\stateZO{37}{13}\stateZ{37}{14}\stateZO{37}{15}\stateZ{37}{16}\stateZ{37}{17}\stateZ{37}{18}\stateZO{37}{19}\stateZ{37}{20}\stateZ{37}{21}\stateZ{37}{22}\stateZ{37}{23}\stateZ{37}{24}\stateZ{37}{25}\stateZ{37}{26}\stateZ{37}{27}\stateZ{37}{28}\stateZ{37}{29}\stateZ{37}{30}\stateZ{37}{31}\stateZ{37}{32}\stateZ{37}{33}\stateZ{37}{34}\stateZ{37}{35}\stateZ{37}{36}\stateZ{37}{37}\stateZ{37}{38}\stateZ{37}{39}\stateZ{37}{40}\stateZ{37}{41}\stateZ{37}{42}\stateZ{37}{43}\stateZ{37}{44}\stateZ{37}{45}\stateZ{37}{46}\stateZ{37}{47}\stateZ{37}{48}\stateZ{37}{49}\stateZ{37}{50}\stateZ{37}{51}
\stateZ{38}{0}\stateZ{38}{1}\stateZ{38}{2}\stateZ{38}{3}\stateZ{38}{4}\stateZ{38}{5}\stateZO{38}{6}\stateZO{38}{7}\stateZO{38}{8}\stateZO{38}{9}\stateZ{38}{10}\stateZ{38}{11}\stateZ{38}{12}\stateZ{38}{13}\stateZO{38}{14}\stateZO{38}{15}\stateZO{38}{16}\stateZO{38}{17}\stateZ{38}{18}\stateZ{38}{19}\stateZ{38}{20}\stateZ{38}{21}\stateZ{38}{22}\stateZ{38}{23}\stateZ{38}{24}\stateZ{38}{25}\stateZ{38}{26}\stateZ{38}{27}\stateZ{38}{28}\stateZ{38}{29}\stateZ{38}{30}\stateZ{38}{31}\stateZ{38}{32}\stateZ{38}{33}\stateZ{38}{34}\stateZ{38}{35}\stateZ{38}{36}\stateZ{38}{37}\stateZ{38}{38}\stateZ{38}{39}\stateZ{38}{40}\stateZ{38}{41}\stateZ{38}{42}\stateZ{38}{43}\stateZ{38}{44}\stateZ{38}{45}\stateZ{38}{46}\stateZ{38}{47}\stateZ{38}{48}\stateZ{38}{49}\stateZ{38}{50}\stateZ{38}{51}
\stateZ{39}{0}\stateZ{39}{1}\stateZ{39}{2}\stateZO{39}{3}\stateZ{39}{4}\stateZO{39}{5}\stateZ{39}{6}\stateZO{39}{7}\stateZ{39}{8}\stateZ{39}{9}\stateZ{39}{10}\stateZO{39}{11}\stateZ{39}{12}\stateZO{39}{13}\stateZ{39}{14}\stateZO{39}{15}\stateZ{39}{16}\stateZ{39}{17}\stateZ{39}{18}\stateZ{39}{19}\stateZ{39}{20}\stateZ{39}{21}\stateZ{39}{22}\stateZ{39}{23}\stateZ{39}{24}\stateZ{39}{25}\stateZ{39}{26}\stateZ{39}{27}\stateZ{39}{28}\stateZ{39}{29}\stateZ{39}{30}\stateZ{39}{31}\stateZ{39}{32}\stateZ{39}{33}\stateZ{39}{34}\stateZ{39}{35}\stateZ{39}{36}\stateZ{39}{37}\stateZ{39}{38}\stateZ{39}{39}\stateZ{39}{40}\stateZ{39}{41}\stateZ{39}{42}\stateZ{39}{43}\stateZ{39}{44}\stateZ{39}{45}\stateZ{39}{46}\stateZ{39}{47}\stateZ{39}{48}\stateZ{39}{49}\stateZ{39}{50}\stateZ{39}{51}
\stateZ{40}{0}\stateZ{40}{1}\stateZ{40}{2}\stateZ{40}{3}\stateZO{40}{4}\stateZO{40}{5}\stateZ{40}{6}\stateZ{40}{7}\stateZO{40}{8}\stateZO{40}{9}\stateZ{40}{10}\stateZ{40}{11}\stateZO{40}{12}\stateZO{40}{13}\stateZ{40}{14}\stateZ{40}{15}\stateZ{40}{16}\stateZ{40}{17}\stateZ{40}{18}\stateZ{40}{19}\stateZ{40}{20}\stateZ{40}{21}\stateZ{40}{22}\stateZ{40}{23}\stateZ{40}{24}\stateZ{40}{25}\stateZ{40}{26}\stateZ{40}{27}\stateZ{40}{28}\stateZ{40}{29}\stateZ{40}{30}\stateZ{40}{31}\stateZ{40}{32}\stateZ{40}{33}\stateZ{40}{34}\stateZ{40}{35}\stateZ{40}{36}\stateZ{40}{37}\stateZ{40}{38}\stateZ{40}{39}\stateZ{40}{40}\stateZ{40}{41}\stateZ{40}{42}\stateZ{40}{43}\stateZ{40}{44}\stateZ{40}{45}\stateZ{40}{46}\stateZ{40}{47}\stateZ{40}{48}\stateZ{40}{49}\stateZ{40}{50}\stateZ{40}{51}
\stateZ{41}{0}\stateZ{41}{1}\stateZ{41}{2}\stateZO{41}{3}\stateZ{41}{4}\stateZ{41}{5}\stateZ{41}{6}\stateZO{41}{7}\stateZ{41}{8}\stateOO{41}{9}\stateZ{41}{10}\stateZO{41}{11}\stateZ{41}{12}\stateZO{41}{13}\stateZ{41}{14}\stateZO{41}{15}\stateZ{41}{16}\stateOO{41}{17}\stateZ{41}{18}\stateZ{41}{19}\stateZ{41}{20}\stateZ{41}{21}\stateZ{41}{22}\stateZ{41}{23}\stateZ{41}{24}\stateZ{41}{25}\stateZ{41}{26}\stateZ{41}{27}\stateZ{41}{28}\stateZ{41}{29}\stateZ{41}{30}\stateZ{41}{31}\stateZ{41}{32}\stateZ{41}{33}\stateZ{41}{34}\stateZ{41}{35}\stateZ{41}{36}\stateZ{41}{37}\stateZ{41}{38}\stateZ{41}{39}\stateZ{41}{40}\stateZ{41}{41}\stateZ{41}{42}\stateZ{41}{43}\stateZ{41}{44}\stateZ{41}{45}\stateZ{41}{46}\stateZ{41}{47}\stateZ{41}{48}\stateZ{41}{49}\stateZ{41}{50}\stateZ{41}{51}
\stateZ{42}{0}\stateZ{42}{1}\stateZ{42}{2}\stateZ{42}{3}\stateZ{42}{4}\stateZ{42}{5}\stateZ{42}{6}\stateZ{42}{7}\stateZ{42}{8}\stateZ{42}{9}\stateOO{42}{10}\stateOO{42}{11}\stateZO{42}{12}\stateZO{42}{13}\stateZO{42}{14}\stateZO{42}{15}\stateZO{42}{16}\stateZO{42}{17}\stateOO{42}{18}\stateOO{42}{19}\stateZ{42}{20}\stateZ{42}{21}\stateZ{42}{22}\stateZ{42}{23}\stateZ{42}{24}\stateZ{42}{25}\stateZ{42}{26}\stateZ{42}{27}\stateZ{42}{28}\stateZ{42}{29}\stateZ{42}{30}\stateZ{42}{31}\stateZ{42}{32}\stateZ{42}{33}\stateZ{42}{34}\stateZ{42}{35}\stateZ{42}{36}\stateZ{42}{37}\stateZ{42}{38}\stateZ{42}{39}\stateZ{42}{40}\stateZ{42}{41}\stateZ{42}{42}\stateZ{42}{43}\stateZ{42}{44}\stateZ{42}{45}\stateZ{42}{46}\stateZ{42}{47}\stateZ{42}{48}\stateZ{42}{49}\stateZ{42}{50}\stateZ{42}{51}
\stateZ{43}{0}\stateZ{43}{1}\stateZ{43}{2}\stateZO{43}{3}\stateZ{43}{4}\stateZ{43}{5}\stateZ{43}{6}\stateZO{43}{7}\stateZ{43}{8}\stateOO{43}{9}\stateZ{43}{10}\stateZO{43}{11}\stateZ{43}{12}\stateZO{43}{13}\stateZ{43}{14}\stateZO{43}{15}\stateZ{43}{16}\stateOO{43}{17}\stateZ{43}{18}\stateZ{43}{19}\stateZ{43}{20}\stateZ{43}{21}\stateZ{43}{22}\stateZ{43}{23}\stateZ{43}{24}\stateZ{43}{25}\stateZ{43}{26}\stateZ{43}{27}\stateZ{43}{28}\stateZ{43}{29}\stateZ{43}{30}\stateZ{43}{31}\stateZ{43}{32}\stateZ{43}{33}\stateZ{43}{34}\stateZ{43}{35}\stateZ{43}{36}\stateZ{43}{37}\stateZ{43}{38}\stateZ{43}{39}\stateZ{43}{40}\stateZ{43}{41}\stateZ{43}{42}\stateZ{43}{43}\stateZ{43}{44}\stateZ{43}{45}\stateZ{43}{46}\stateZ{43}{47}\stateZ{43}{48}\stateZ{43}{49}\stateZ{43}{50}\stateZ{43}{51}
\stateZ{44}{0}\stateZ{44}{1}\stateZ{44}{2}\stateZ{44}{3}\stateZO{44}{4}\stateZO{44}{5}\stateZ{44}{6}\stateZ{44}{7}\stateZO{44}{8}\stateZO{44}{9}\stateZ{44}{10}\stateZ{44}{11}\stateZO{44}{12}\stateZO{44}{13}\stateZ{44}{14}\stateZ{44}{15}\stateZ{44}{16}\stateZ{44}{17}\stateZ{44}{18}\stateZ{44}{19}\stateZ{44}{20}\stateZ{44}{21}\stateZ{44}{22}\stateZ{44}{23}\stateZ{44}{24}\stateZ{44}{25}\stateZ{44}{26}\stateZ{44}{27}\stateZ{44}{28}\stateZ{44}{29}\stateZ{44}{30}\stateZ{44}{31}\stateZ{44}{32}\stateZ{44}{33}\stateZ{44}{34}\stateZ{44}{35}\stateZ{44}{36}\stateZ{44}{37}\stateZ{44}{38}\stateZ{44}{39}\stateZ{44}{40}\stateZ{44}{41}\stateZ{44}{42}\stateZ{44}{43}\stateZ{44}{44}\stateZ{44}{45}\stateZ{44}{46}\stateZ{44}{47}\stateZ{44}{48}\stateZ{44}{49}\stateZ{44}{50}\stateZ{44}{51}
\stateZ{45}{0}\stateZ{45}{1}\stateZ{45}{2}\stateZO{45}{3}\stateZ{45}{4}\stateZO{45}{5}\stateZ{45}{6}\stateZO{45}{7}\stateZ{45}{8}\stateZ{45}{9}\stateZ{45}{10}\stateZO{45}{11}\stateZ{45}{12}\stateZ{45}{13}\stateZ{45}{14}\stateZ{45}{15}\stateZ{45}{16}\stateZ{45}{17}\stateZ{45}{18}\stateZ{45}{19}\stateZ{45}{20}\stateZ{45}{21}\stateZ{45}{22}\stateZ{45}{23}\stateZ{45}{24}\stateZ{45}{25}\stateZ{45}{26}\stateZ{45}{27}\stateZ{45}{28}\stateZ{45}{29}\stateZ{45}{30}\stateZ{45}{31}\stateZ{45}{32}\stateZ{45}{33}\stateZ{45}{34}\stateZ{45}{35}\stateZ{45}{36}\stateZ{45}{37}\stateZ{45}{38}\stateZ{45}{39}\stateZ{45}{40}\stateZ{45}{41}\stateZ{45}{42}\stateZ{45}{43}\stateZ{45}{44}\stateZ{45}{45}\stateZ{45}{46}\stateZ{45}{47}\stateZ{45}{48}\stateZ{45}{49}\stateZ{45}{50}\stateZ{45}{51}
\stateZ{46}{0}\stateZ{46}{1}\stateZ{46}{2}\stateZ{46}{3}\stateZ{46}{4}\stateZ{46}{5}\stateZO{46}{6}\stateZO{46}{7}\stateZO{46}{8}\stateZO{46}{9}\stateZ{46}{10}\stateZ{46}{11}\stateZ{46}{12}\stateZ{46}{13}\stateZ{46}{14}\stateZ{46}{15}\stateZ{46}{16}\stateZ{46}{17}\stateZ{46}{18}\stateZ{46}{19}\stateZ{46}{20}\stateZ{46}{21}\stateZ{46}{22}\stateZ{46}{23}\stateZ{46}{24}\stateZ{46}{25}\stateZ{46}{26}\stateZ{46}{27}\stateZ{46}{28}\stateZ{46}{29}\stateZ{46}{30}\stateZ{46}{31}\stateZ{46}{32}\stateZ{46}{33}\stateZ{46}{34}\stateZ{46}{35}\stateZ{46}{36}\stateZ{46}{37}\stateZ{46}{38}\stateZ{46}{39}\stateZ{46}{40}\stateZ{46}{41}\stateZ{46}{42}\stateZ{46}{43}\stateZ{46}{44}\stateZ{46}{45}\stateZ{46}{46}\stateZ{46}{47}\stateZ{46}{48}\stateZ{46}{49}\stateZ{46}{50}\stateZ{46}{51}
\stateZ{47}{0}\stateZ{47}{1}\stateZ{47}{2}\stateZO{47}{3}\stateZ{47}{4}\stateZO{47}{5}\stateZ{47}{6}\stateZO{47}{7}\stateZ{47}{8}\stateZ{47}{9}\stateZ{47}{10}\stateZ{47}{11}\stateZ{47}{12}\stateZ{47}{13}\stateZ{47}{14}\stateZ{47}{15}\stateZ{47}{16}\stateZ{47}{17}\stateZ{47}{18}\stateZ{47}{19}\stateZ{47}{20}\stateZ{47}{21}\stateZ{47}{22}\stateZ{47}{23}\stateZ{47}{24}\stateZ{47}{25}\stateZ{47}{26}\stateZ{47}{27}\stateZ{47}{28}\stateZ{47}{29}\stateZ{47}{30}\stateZ{47}{31}\stateZ{47}{32}\stateZ{47}{33}\stateZ{47}{34}\stateZ{47}{35}\stateZ{47}{36}\stateZ{47}{37}\stateZ{47}{38}\stateZ{47}{39}\stateZ{47}{40}\stateZ{47}{41}\stateZ{47}{42}\stateZ{47}{43}\stateZ{47}{44}\stateZ{47}{45}\stateZ{47}{46}\stateZ{47}{47}\stateZ{47}{48}\stateZ{47}{49}\stateZ{47}{50}\stateZ{47}{51}
\stateZ{48}{0}\stateZ{48}{1}\stateZ{48}{2}\stateZ{48}{3}\stateZO{48}{4}\stateZO{48}{5}\stateZ{48}{6}\stateZ{48}{7}\stateZ{48}{8}\stateZ{48}{9}\stateZ{48}{10}\stateZ{48}{11}\stateZ{48}{12}\stateZ{48}{13}\stateZ{48}{14}\stateZ{48}{15}\stateZ{48}{16}\stateZ{48}{17}\stateZ{48}{18}\stateZ{48}{19}\stateZ{48}{20}\stateZ{48}{21}\stateZ{48}{22}\stateZ{48}{23}\stateZ{48}{24}\stateZ{48}{25}\stateZ{48}{26}\stateZ{48}{27}\stateZ{48}{28}\stateZ{48}{29}\stateZ{48}{30}\stateZ{48}{31}\stateZ{48}{32}\stateZ{48}{33}\stateZ{48}{34}\stateZ{48}{35}\stateZ{48}{36}\stateZ{48}{37}\stateZ{48}{38}\stateZ{48}{39}\stateZ{48}{40}\stateZ{48}{41}\stateZ{48}{42}\stateZ{48}{43}\stateZ{48}{44}\stateZ{48}{45}\stateZ{48}{46}\stateZ{48}{47}\stateZ{48}{48}\stateZ{48}{49}\stateZ{48}{50}\stateZ{48}{51}
\stateZ{49}{0}\stateOO{49}{1}\stateZ{49}{2}\stateZO{49}{3}\stateZ{49}{4}\stateZ{49}{5}\stateZ{49}{6}\stateZ{49}{7}\stateZ{49}{8}\stateZ{49}{9}\stateZ{49}{10}\stateZ{49}{11}\stateZ{49}{12}\stateZ{49}{13}\stateZ{49}{14}\stateZ{49}{15}\stateZ{49}{16}\stateZ{49}{17}\stateZ{49}{18}\stateZ{49}{19}\stateZ{49}{20}\stateZ{49}{21}\stateZ{49}{22}\stateZ{49}{23}\stateZ{49}{24}\stateZ{49}{25}\stateZ{49}{26}\stateZ{49}{27}\stateZ{49}{28}\stateZ{49}{29}\stateZ{49}{30}\stateZ{49}{31}\stateZ{49}{32}\stateZ{49}{33}\stateZ{49}{34}\stateZ{49}{35}\stateZ{49}{36}\stateZ{49}{37}\stateZ{49}{38}\stateZ{49}{39}\stateZ{49}{40}\stateZ{49}{41}\stateZ{49}{42}\stateZ{49}{43}\stateZ{49}{44}\stateZ{49}{45}\stateZ{49}{46}\stateZ{49}{47}\stateZ{49}{48}\stateZ{49}{49}\stateZ{49}{50}\stateZ{49}{51}
\stateZ{50}{0}\stateZ{50}{1}\stateOO{50}{2}\stateOO{50}{3}\stateZ{50}{4}\stateZ{50}{5}\stateZ{50}{6}\stateZ{50}{7}\stateZ{50}{8}\stateZ{50}{9}\stateZ{50}{10}\stateZ{50}{11}\stateZ{50}{12}\stateZ{50}{13}\stateZ{50}{14}\stateZ{50}{15}\stateZ{50}{16}\stateZ{50}{17}\stateZ{50}{18}\stateZ{50}{19}\stateZ{50}{20}\stateZ{50}{21}\stateZ{50}{22}\stateZ{50}{23}\stateZ{50}{24}\stateZ{50}{25}\stateZ{50}{26}\stateZ{50}{27}\stateZ{50}{28}\stateZ{50}{29}\stateZ{50}{30}\stateZ{50}{31}\stateZ{50}{32}\stateZ{50}{33}\stateZ{50}{34}\stateZ{50}{35}\stateZ{50}{36}\stateZ{50}{37}\stateZ{50}{38}\stateZ{50}{39}\stateZ{50}{40}\stateZ{50}{41}\stateZ{50}{42}\stateZ{50}{43}\stateZ{50}{44}\stateZ{50}{45}\stateZ{50}{46}\stateZ{50}{47}\stateZ{50}{48}\stateZ{50}{49}\stateZ{50}{50}\stateZ{50}{51}
\stateZ{51}{0}\stateOO{51}{1}\stateZ{51}{2}\stateZ{51}{3}\stateZ{51}{4}\stateZ{51}{5}\stateZ{51}{6}\stateZ{51}{7}\stateZ{51}{8}\stateZ{51}{9}\stateZ{51}{10}\stateZ{51}{11}\stateZ{51}{12}\stateZ{51}{13}\stateZ{51}{14}\stateZ{51}{15}\stateZ{51}{16}\stateZ{51}{17}\stateZ{51}{18}\stateZ{51}{19}\stateZ{51}{20}\stateZ{51}{21}\stateZ{51}{22}\stateZ{51}{23}\stateZ{51}{24}\stateZ{51}{25}\stateZ{51}{26}\stateZ{51}{27}\stateZ{51}{28}\stateZ{51}{29}\stateZ{51}{30}\stateZ{51}{31}\stateZ{51}{32}\stateZ{51}{33}\stateZ{51}{34}\stateZ{51}{35}\stateZ{51}{36}\stateZ{51}{37}\stateZ{51}{38}\stateZ{51}{39}\stateZ{51}{40}\stateZ{51}{41}\stateZ{51}{42}\stateZ{51}{43}\stateZ{51}{44}\stateZ{51}{45}\stateZ{51}{46}\stateZ{51}{47}\stateZ{51}{48}\stateZ{51}{49}\stateZ{51}{50}\stateZ{51}{51}
    \end{tikzpicture}
  \end{center}
\caption{Simulation of $\oplus_2$ with a triangular neighborhood at iteration
  25 starting from an initial configuration with a single 1 inside 0s:
  cells in state 1 are in red,  cells that have been in state 1 between iteration 0 and 24 but not
  at 25 are in blue, and the others cells are not drawn.}
\label{fig:Cole}
\end{figure}
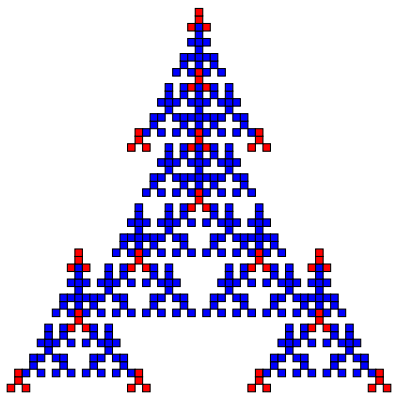

\begin{prop}
$T_{2}(\sigma_{(0,36)}(c^1))$ is null.
\end{prop}
\begin{proof}
We will prove, by induction, that $T_{2}(\sigma_{(0,36)}(c^1))$ is null from 0 to $2^k$.
For $k=0$ to 5 is clear since cell $(0,36)$ is too far to touch $B_2(0,0)$.
At iteration $2^k$ the only active cells are  $(-2^k,-2^k+36)$, $(2^k,-2^k+36)$, $(0,2^k+36)$ and $(0,36)$.
The first three are too far to touch $B_2(0,0)$ from iterations $2^k$ to $2^{k+1}$.
By induction hypothesis $(0,36)$ does not attain $B_2(0,0)$ in $2^k$ iterations, thus $B_2(0,0)$ remains null until iteration $2^{k+1}$.
\end{proof}

\section{Open Problems}

We showed in this paper that dynamics like pre-expansivity
or $k$-expansivity can exist without necessarily implying positive
expansivity. We also showed that some combinations of the space
structure and the local rule structure forbid pre-expansivity
(Theorem~\ref{thm:abelian2D}).
However, we left many open questions concerning pre-expansivity and $k$-expansivity: 
\begin{itemize}
\item is there a pre-expansive cellular automata on $\Z^d$ when $d\geq 2$?
\item is there a $2$-expansive cellular automata on $\Z^2$? on the free group? more
  generally what is the set of integers $k$ such that a given group
  admits $k$-expansive cellular automata?
\item a lot of results are known on traces of (positively) expansive cellular automata. What can we say for pre-expansive CA? for instance are they always transitive? what are their mixing properties?
\item is the property of pre-expansivity decidable?
\end{itemize}

\section{Acknowledgment}

We would like to thank the anonymous referees who pointed out several technical errors and made valuable comments that certainly improved the present paper.

\section*{References}

\bibliography{biblio}

\begin{thebibliography}{10}
\expandafter\ifx\csname url\endcsname\relax
  \def\url#1{\texttt{#1}}\fi
\expandafter\ifx\csname urlprefix\endcsname\relax\def\urlprefix{URL }\fi
\expandafter\ifx\csname href\endcsname\relax
  \def\href#1#2{#2} \def\path#1{#1}\fi

\bibitem{hedlund}
G.~A. Hedlund, Endomorphisms and automorphisms of the shift dynamical system,
  Mathematical Systems Theory 3 (1969) 320--375.

\bibitem{kurkabook}
P.~K\r{u}rka, Topological and symbolic dynamics, Soci\'et\'e Math\'ematique de
  France, 2003.

\bibitem{ergca}
M.~Pivato, The ergodic theory of cellular automata, Int. J. General Systems
  41~(6) (2012) 583--594.

\bibitem{cagroups}
T.~Ceccherini-Silberstein, M.~Coornaert, Cellular Automata and Groups,
  Springer, 2010.

\bibitem{moore}
E.~F. Moore, Machine models of self-reproduction, in: A.~W. Burks (Ed.), Essays
  on Cellular Automata, University of Illinois Press, 1970, pp. 187--203.

\bibitem{myhill}
J.~Myhill, The {C}onverse of {M}oore's {G}arden-of-{E}den {T}heorem, in:
  Proceedings of the American Mathematical Society, Vol.~14, American
  Mathematical Society, 1963, pp. 658--686.

\bibitem{AIF_1999}
T.~G. Ceccherini-Silberstein, A.~Machi, F.~Scarabotti,
  \href{https://aif.centre-mersenne.org/item/AIF_1999__49_2_673_0}{Amenable
  groups and cellular automata}, Annales de l'Institut Fourier 49~(2) (1999)
  673--685.
\newblock \href {http://dx.doi.org/10.5802/aif.1686}
  {\path{doi:10.5802/aif.1686}}.
\newline\urlprefix\url{https://aif.centre-mersenne.org/item/AIF_1999__49_2_673_0}

\bibitem{Bartholdi_2010}
L.~Bartholdi, \href{http://dx.doi.org/10.4171/JEMS/196}{Gardens of eden and
  amenability on cellular automata}, Journal of the European Mathematical
  Society (2010) 241–248\href {http://dx.doi.org/10.4171/jems/196}
  {\path{doi:10.4171/jems/196}}.
\newline\urlprefix\url{http://dx.doi.org/10.4171/JEMS/196}

\bibitem{Bartholdi16}
L.~Bartholdi, \href{http://arxiv.org/abs/1605.09133}{Amenability of groups is
  characterized by myhill's theorem}, CoRR abs/1605.09133.
\newblock \href {http://arxiv.org/abs/1605.09133} {\path{arXiv:1605.09133}}.
\newline\urlprefix\url{http://arxiv.org/abs/1605.09133}

\bibitem{machi}
A.~Machi, F.~Mignosi, Garden of eden configurations for cellular automata on
  {Cayley} graphs of groups, SIAM journal on discrete mathematics 6~(1) (1993)
  44.

\bibitem{BoyleOpen}
M.~Boyle, Open problems in symbolic dynamics, Contemporary Mathematics 469
  (2008) 69--118.

\bibitem{blanchardmaass}
F.~Blanchard, A.~Maass, Dynamical properties of expansive one-sided cellular
  automata, Israel Journal of Mathematics 99 (1997) 149--174.

\bibitem{nasu}
M.~Nasu, Nondegenerate q-biresolving textile systems and expansive cellular
  automata of onesided full shifts, Transactions of the American Mathematical
  Society 358 (2006) 871--891.

\bibitem{shereshevsky}
M.~A. Shereshevsky, Expansiveness, entropy and polynomial growth for groups
  acting on subshifts by automorphisms, Indagationes Mathematicae 4~(2) (1993)
  203 -- 210.
\newblock \href
  {http://dx.doi.org/http://dx.doi.org/10.1016/0019-3577(93)90040-6}
  {\path{doi:http://dx.doi.org/10.1016/0019-3577(93)90040-6}}.

\bibitem{Pivato11}
M.~Pivato, Positive expansiveness versus network dimension in symbolic
  dynamical systems, Theor. Comput. Sci. 412~(30) (2011) 3838--3855.
\newblock \href {http://dx.doi.org/10.1016/j.tcs.2011.02.021}
  {\path{doi:10.1016/j.tcs.2011.02.021}}.

\bibitem{Kurka97}
P.~K{\r{u}}rka, Languages, equicontinuity and attractors in cellular automata,
  Ergodic Theory and Dynamical Systems 17 (1997) 417--433.

\bibitem{JadurYazlle}
C.~Jadur, J.~Yazlle, On the dynamics of cellular automata induced from a prefix
  code, Advances in Applied Mathematics 38~(1) (2007) 27 -- 53.
\newblock \href {http://dx.doi.org/http://dx.doi.org/10.1016/j.aam.2005.11.004}
  {\path{doi:http://dx.doi.org/10.1016/j.aam.2005.11.004}}.

\bibitem{dilena06}
P.~Di~Lena, Decidable properties for regular cellular automata, in: G.~Navarro,
  L.~Bertossi, Y.~Kohayakawa (Eds.), Fourth IFIP International Conference on
  Theoretical Computer Science- TCS 2006, Vol. 209 of IFIP International
  Federation for Information Processing, Springer US, 2006, pp. 185--196.
\newblock \href {http://dx.doi.org/10.1007/978-0-387-34735-6_17}
  {\path{doi:10.1007/978-0-387-34735-6_17}}.

\bibitem{phdLukkarila}
V.~Lukkarila, On undecidable dynamical properties of reversible one-dimensional
  cellular automata, Ph.D. thesis, Turku Centre for Computer Science (2010).

\bibitem{posexponetoone}
E.~M. Coven, M.~Keane, Every compact metric space that supports a positively
  expansive homeomorphism is finite, Vol. Volume 48 of Lecture Notes--Monograph
  Series, Institute of Mathematical Statistics, Beachwood, Ohio, USA, 2006, pp.
  304--305.
\newblock \href {http://dx.doi.org/10.1214/074921706000000310}
  {\path{doi:10.1214/074921706000000310}}.

\bibitem{Kari94}
J.~Kari, Reversibility and surjectivity problems of cellular automata, J.
  Comput. Syst. Sci. 48~(1) (1994) 149--182.

\bibitem{Amoroso72}
S.~Amoroso, Y.~N. Patt, Decision procedures for surjectivity and injectivity of
  parallel maps for tesselation structures, Journal of Computer and System
  Sciences 6 (1972) 448--464.

\bibitem{Sablik08}
M.~Sablik, Directional dynamics for cellular automata: {A} sensitivity to
  initial condition approach, Theor. Comput. Sci. 400~(1-3) (2008) 1--18.
\newblock \href {http://dx.doi.org/10.1016/j.tcs.2008.02.052}
  {\path{doi:10.1016/j.tcs.2008.02.052}}.

\bibitem{fractal}
J.~G{\"u}tschow, V.~Nesme, R.~F. Werner, The fractal structure of cellular
  automata on abelian groups, in: Proceedings of Automata 2010, 2010, pp.
  55--74.

\bibitem{abelianrandomization}
B.~Hellouin De~Menibus, V.~Salo, G.~Theyssier, Characterizing asymptotic
  randomization in abelian cellular automata, Ergodic Theory and Dynamical
  Systems (2018) 1--30\href {http://dx.doi.org/10.1017/etds.2018.75}
  {\path{doi:10.1017/etds.2018.75}}.

\bibitem{itoosatonasu}
M.~Itô, N.~Ôsato, M.~Nasu, Linear cellular automata over {Zm}, Journal of
  Computer and System Sciences 27~(1) (1983) 125 -- 140.
\newblock \href
  {http://dx.doi.org/http://dx.doi.org/10.1016/0022-0000(83)90033-8}
  {\path{doi:http://dx.doi.org/10.1016/0022-0000(83)90033-8}}.

\bibitem{Manzinimargara}
G.~Manzini, L.~Margara, A complete and efficiently computable topological
  classification of d-dimensional linear cellular automata over {Zm},
  Theoretical Computer Science 221~(1–2) (1999) 157 -- 177.
\newblock \href
  {http://dx.doi.org/http://dx.doi.org/10.1016/S0304-3975(99)00031-6}
  {\path{doi:http://dx.doi.org/10.1016/S0304-3975(99)00031-6}}.

\bibitem{kari00}
J.~Kari, Linear cellular automata with multiple state variables, in:
  H.~Reichel, S.~Tison (Eds.), STACS 2000, Springer Berlin Heidelberg, Berlin,
  Heidelberg, 2000, pp. 110--121.

\bibitem{moore98}
C.~Moore, Predicting non-linear cellular automata quickly by decomposing them
  into linear ones, Physica D 111 (1998) 27--41.

\bibitem{DelacourtPST11}
M.~Delacourt, V.~Poupet, M.~Sablik, G.~Theyssier, Directional dynamics along
  arbitrary curves in cellular automata, Theor. Comput. Sci. 412~(30) (2011)
  3800--3821.
\newblock \href {http://dx.doi.org/10.1016/j.tcs.2011.02.019}
  {\path{doi:10.1016/j.tcs.2011.02.019}}.

\bibitem{Margolus1984}
N.~Margolus, Physics-like models of computation, Physica D: Nonlinear Phenomena
  10~(1–2) (1984) 81 -- 95.
\newblock \href
  {http://dx.doi.org/http://dx.doi.org/10.1016/0167-2789(84)90252-5}
  {\path{doi:http://dx.doi.org/10.1016/0167-2789(84)90252-5}}.

\bibitem{phdDiLena}
P.~Di~Lena, Decidable and computational properties of cellular automata, Ph.D.
  thesis, University of Bologna (2007).

\bibitem{kitchens}
B.~P. Kitchens, Expansive dynamics on zero-dimensional groups, Ergodic Theory
  and Dynamical Systems 7 (1987) 249--261.
\newblock \href {http://dx.doi.org/10.1017/S0143385700003989}
  {\path{doi:10.1017/S0143385700003989}}.

\bibitem{Kari12}
J.~Kari, Universal pattern generation by cellular automata, Theor. Comput. Sci.
  429 (2012) 180--184.
\newblock \href {http://dx.doi.org/10.1016/j.tcs.2011.12.037}
  {\path{doi:10.1016/j.tcs.2011.12.037}}.

\bibitem{Pivato2009}
M.~Pivato, Ergodic theory of cellular automata, in: R.~A. Meyers (Ed.),
  Encyclopedia of Complexity and Systems Science, Springer New York, New York,
  NY, 2009, pp. 2980--3015.
\newblock \href {http://dx.doi.org/10.1007/978-0-387-30440-3_178}
  {\path{doi:10.1007/978-0-387-30440-3_178}}.

\bibitem{textbookabeliangroups}
L.~Fuchs, Abelian groups, International series of monographs in pure and
  applied mathematics, Pergamon Press, 1960.

\end{thebibliography}

\end{document}